\documentclass[12pt]{article}

\usepackage{amsmath,amssymb,amsthm}
\usepackage{mathtools}
\usepackage[round,comma,authoryear]{natbib}

\usepackage{graphicx}
\usepackage{booktabs}
\usepackage{float}
\usepackage{algorithm}
\usepackage{algorithmic}
\usepackage{multirow}
\usepackage{threeparttable}

\usepackage[margin=1in]{geometry}
\usepackage{setspace}
\usepackage{enumitem}

\usepackage{xurl}
\usepackage{hyperref}

\mathtoolsset{showonlyrefs}

\floatstyle{ruled}
\restylefloat{algorithm}

\hypersetup{
    colorlinks=true,
    linkcolor=blue,
    citecolor=blue,
    urlcolor=blue
}

\newtheorem{theorem}{Theorem}[section]
\newtheorem{lemma}[theorem]{Lemma}
\newtheorem{proposition}[theorem]{Proposition}
\newtheorem{corollary}[theorem]{Corollary}

\theoremstyle{plain}
\newtheorem{assumption}{Assumption}

\theoremstyle{definition}
\newtheorem{definition}{Definition}

\theoremstyle{remark}
\newtheorem{remark}{Remark}

\newcommand{\corr}[1]{{\color{black}#1}}

\newcommand{\thetahat}{\widehat{\theta}}

\newcommand{\E}{\mathbb{E}}
\renewcommand{\P}{\mathbb{P}}
\newcommand{\R}{\mathbb{R}}
\newcommand{\N}{\mathbb{N}}

\newcommand{\Var}{\operatorname{Var}}
\newcommand{\Cov}{\operatorname{Cov}}

\DeclareMathOperator{\KL}{KL}

\newcommand{\Pnz}{P_n^0}

\newcommand{\PX}{P_X}
\newcommand{\QX}{Q_X}

\newcommand{\EP}{\E_P}
\newcommand{\EQ}{\E_Q}
\newcommand{\Em}{\E_m}
\newcommand{\En}{\E_n}
\newcommand{\Enk}[1]{\E_n^{(-#1)}}

\newcommand{\calC}{\mathcal{C}}
\newcommand{\calX}{\mathcal{X}}

\newcommand{\calI}{\mathcal{I}}
\newcommand{\calK}{\mathcal{K}}
\newcommand{\calT}{\mathcal{T}}

\newcommand{\rdag}{r^{\dagger}}

\newcommand{\muhat}{\hat{\mu}_0}
\newcommand{\rhat}{\hat{r}}
\newcommand{\that}{\widehat{\tau}}
\newcommand{\Lhat}{\widehat{L}}

\newcommand{\dr}{\delta r}
\newcommand{\dmu}{\delta\mu}

\newcommand{\norm}[2]{\left\|#1\right\|_{#2}}

\newcommand{\anon}{1}

\begin{document}

\if1\anon
  \title{\bf Transporting Randomized Trial Effects to Real-World Populations\\
via Riesz-Calibrated Optimal Transport}

  \author{
    Anik Burman$^{1}$,
    Margaret Gamalo$^{2}$,
    Promit Ghosal$^{3\ast}$,
    Prosenjit Kundu$^{2}$\thanks{Address for correspondence:
      \href{prosenjit.kundu@pfizer.com}{$^2$prosenjit.kundu@pfizer.com};
            \href{promit@uchicago.edu}{$^3$promit@uchicago.edu}.}\\[4pt]
    $^{1}$Department of Biostatistics, Johns Hopkins University\\
    $^{2}$Statistics, Inflammation \& Immunology, Pfizer Inc.\\
    $^{3}$Department of Statistics, University of Chicago
  }
  \date{}
  \maketitle
\fi

\if0\anon
  \bigskip\bigskip\bigskip
  \begin{center}
    {\LARGE\bf An Optimal-Transport Framework for Integrating
     RCT and Real-World Evidence}
  \end{center}
  \medskip
\fi

\begin{abstract}
\noindent
Randomized trials support causal inference, but differences between trial and target populations can limit the transportability of treatment effects to real-world settings. Many existing approaches model the propensity of trial participation and can therefore be sensitive to model misspecification and weak overlap of the covariate distributions. Optimal Transport (OT) offers a different route by comparing the trial and target populations directly in covariate space. We develop \textsc{RiCOT}, a Riesz-calibrated OT procedure transporting treatment effects to a treated target population. We consider a semi-unbalanced OT with entropic regularization where the source marginals are relaxed. We show that the uncalibrated OT introduces a bias which does not shrink with increasing sample size. \textsc{RiCOT} removes this bias by imposing calibration equations directly within the transport problem. With a growing calibration sieve, the calibrated weight consistently estimates the target-to-trial density ratio, equivalently the Riesz representer of the target expectation functional, even when the entropic and source-relaxation parameters remain fixed and positive. Combined with outcome regression, the resulting estimator is doubly robust and attains the semiparametric efficiency bound under suitable rate conditions. Its variance is estimated directly from the influence function, without resampling or repeated OT optimization. Simulations show low bias and near-nominal coverage across a range of overlap and misspecification settings, including settings in which sampling-score methods perform poorly. We illustrate \textsc{RiCOT} in a real-world application involving a rare progressive cardiomyopathy, comparing conventional IPW and AIPW estimators with our OT-based IPW and doubly robust estimators for transporting the randomized treatment effect to a real-world population receiving the same treatment.
\end{abstract}

\noindent%

\newpage
 
\section{Introduction}
\label{sec:intro}
 
Randomized controlled trials (RCTs) provide internally valid estimates of
treatment effects, but the patients enrolled in a trial may differ substantially
from those encountered in routine care
\citep{rothwell2005external, stuart2011use}. Eligibility restrictions,
selective participation, and under-representation of important subgroups can
limit the generalizability of trial results. Real-world data (RWD), including
electronic health records and registries, provide information on patients seen
in practice \citep{collins2020real}, but generally lack the randomization needed
for causal comparisons. Together, these challenges raise a natural question. What would the average causal treatment effect have been in the real-world population had treatment been randomized there?

We consider a two-sample setting in which an RCT provides baseline covariates, randomized treatment, and outcome, while an independent RWD sample represents a treated target population in which the same baseline covariates and the treated outcome are observed. Our target estimand is the average treatment effect in the treated real-world population, comparing the potential outcomes under treatment and control, denoted by $$\tau_Q
\coloneqq
\EQ\!\left[Y^{\mathrm Q}(1)-Y^{\mathrm Q}(0)\right],$$  where $Q$ denotes the law of the target observation, \(Y^{\mathrm Q}(a)\) denote the potential outcome under treatment level \(a\in\{0,1\}\) for the target individual. The outcome under treatment is observed in the target population, so the main identification problem is the corresponding counterfactual outcome under control. Under standard transportability conditions, this quantity can be obtained by averaging a trial control regression over \(Q_X\), or equivalently by reweighting the trial covariate law with the target-to-trial density ratio
$$\rdag\coloneq \frac{dQ_X}{dP_X},$$
 where $P_X$ and $Q_X$ denotes the trial and target covariate laws respectively.
The statistical problem is therefore not merely to align two covariate
distributions, but to estimate the particular population weight that makes
target expectations recoverable from the trial. This distinction is especially relevant when the population shift is modeled parametrically. The mechanism governing study participation is unknown, so it can be difficult to know in advance which main effects, interactions, or nonlinearities are needed in a sampling propensity score model. In regulatory settings, transparent specification of design and analysis choices is also important. Current FDA M14 guidance for non-interventional studies using RWD emphasizes early documentation of key design and analytic decisions and prespecified sensitivity analyses in the protocol \citep{fda2026m14}. Although that guidance is directed primarily to non-interventional safety studies, it notes that its principles may also be relevant to effectiveness studies. This motivates an analysis whose key design and tuning choices can be fixed in advance without requiring the correct parametric form of the study-participation mechanism.

We therefore seek an estimator that is robust to misspecification of the sampling propensity score model, adapts to imperfect overlap, and supports valid semiparametric inference. Optimal transport provides a natural starting point because it aligns the trial and target populations directly in covariate space. The methodological challenge is to ensure that the weights induced by regularized transport also target the density ratio required for efficient estimation of our target estimand. To set up our approach, we first consider how conventional weighting and balancing methods address the population shift and how optimal transport differs from them.
 
\subsection{Weighting, Balance, and Optimal Transport}

 The estimand fixes what the weight must accomplish, and existing approaches meet some of these aims but not all. The most direct is sampling propensity score weighting, which reweights trial participants toward the target through a scalar participation score \citep{stuart2011use,tipton2013improving,westreich2017transportability}. Outcome-regression methods instead model the conditional outcome in the trial and average the fitted values over the target population
 \citep{dahabreh2019extending}, while doubly robust methods combine the two strategies \citep{bang2005doubly,dahabreh2020extending}. Sampling propensity score weighting is appealing because it summarizes the population shift through a scalar score, but its performance depends on adequate specification of the sampling propensity score model. If that model fails to capture important nonlinearities or interactions, the resulting weights may converge to a misspecified limit, inducing bias in the transported treatment-effect estimate, while weak overlap can further amplify weight instability.

 Balancing methods take a different route, solving directly for weights that match chosen functions of the covariates between the two samples rather than fitting a propensity score model \citep{hainmueller2012,imai2014covariate,zubizarreta2015stable,chan2016globally}. This avoids committing to a score specification, but because balance is imposed only on a finite set of functions, the weight is constrained along those directions rather than across the full covariate law.

 Optimal transport (OT) addresses both limitations at once. Rather than reducing the covariates to a fitted score, it couples the trial and target populations directly in covariate space \citep{santambrogio2015optimal,peyre2019computational}. The mass leaving each trial observation induces a trial-side weight, a natural analogue of density-ratio weighting, and unlike local matching \citep{abadie2006large,portier2024nearest} the coupling is global rather than neighbourhood-by-neighbourhood. A source-relaxed entropic formulation sharpens this further. It is a semi-unbalanced problem in which the target marginal is held fixed, so the population defining the estimand does not move, while the trial marginal is relaxed so that poorly aligned trial observations may receive little mass. Entropic regularization smooths the coupling and makes Sinkhorn computation practical. Geometric alignment of this kind is exactly what we want to retain. It is not, however, the same as the weighting condition inference requires. Aligning the two populations well in covariate space does not guarantee that the induced weight is the density ratio the transported estimand depends on.
 
\subsection{The Inferential Gap}

At fixed entropic and source-relaxation parameters, regularized OT generally induces a density-ratio estimator that differs from the true target-to-trial density ratio. This discrepancy is a population-level bias and therefore does not disappear as the sample size increases. Although reducing the regularization can lessen this bias, taking the entropic parameter very small can lead to less favorable statistical and computational behavior
\citep{genevay2019sample,mena2019statistical,weed2019sharp}. Thus, increasing the sample size alone is not sufficient to recover the density ratio required for transporting the treatment effect.

This motivates imposing additional structure on the transport weights. The target-to-trial density ratio \(r^\dagger\) is characterized by the Riesz moment conditions
\begin{equation}
\E_{P_X}\bigl[\rdag(X)h(X)\bigr]
 =
\E_{Q_X}\bigl[h(X)\bigr],
\qquad h\in L^2(P_X).
\label{eq:riesz-intro}
\end{equation}
for suitable functions $h$.
These conditions have a direct balancing interpretation. After weighting the trial by \(r^\dagger\), averages of functions of the covariates reproduce their corresponding averages in the target population. The same density ratio also appears as the Riesz representer in the efficient influence function for
the transported treatment effect.

Regularized OT does not enforce these moment conditions directly because its weights are determined by a tradeoff between transport cost and regularization. We therefore calibrate the transport weights to the Riesz moment conditions. With a sufficiently rich calibration space, this removes the regularization-induced distortion and recovers the target-to-trial density ratio while allowing the entropic and source-relaxation parameters to remain
fixed and positive. The calibration step therefore closes the inferential gap between geometric alignment of the two populations and the weighting conditions required for semiparametric inference.
We summarize our contributions below.

\subsection{Contributions}
We develop \textsc{RiCOT}, a Riesz-calibrated semi-unbalanced (source-relaxed) entropic OT
procedure for transporting randomized treatment effects to a treated target population. The contribution is not entropic OT or calibration separately, but the characterization of the inferential distortion induced by regularized transport and a calibration mechanism that removes that distortion while retaining fixed positive transport regularization.

\textit{First}, we derive an exact nonlinear fixed-point representation of the population trial-side weight and use it to obtain a local regularization-bias expansion. The expansion separates the source-relaxation contribution from the entropic-smoothing contribution. We then translate the resulting weight error into an explicit violation of the Riesz moment equations. This identifies, at the population level, why a raw regularized transport weight need not be the weight required by the transported estimand.

\textit{Second}, we impose the Riesz equations directly within the transport program. The calibration multiplier appears as a multiplicative exponential tilt of a transport-generated offset, which preserves positivity and leads to a finite-dimensional convex estimation problem. With a growing calibration sieve, the calibrated weight converges in \(L^2(P_X)\) to \(\rdag\) even when the entropic and source-relaxation parameters remain fixed and positive.

\textit{Third}, we connect the calibrated weight to semiparametric estimation of the target estimand. The efficient influence function and an exact product remainder yield double robustness and the usual product-rate requirement for root-sample-size inference. We establish asymptotic normality and consistency of a direct influence-function variance estimator, and extend the construction to incidence-rate estimands.

\subsection{Related Work}

Our work connects the literatures on trial generalization, balancing and
calibration, causal optimal transport, unbalanced transport, and direct
estimation of Riesz representers.

Methods for extending randomized-trial effects to target populations include outcome standardization, inverse probability of sampling weighting, and doubly robust combinations of the two
\citep{cole2010generalizing,stuart2011use,tipton2013improving,
westreich2017transportability,dahabreh2019extending,
dahabreh2020extending}. Related work studies the use of observational or
population samples to adjust randomized trials for differences between trial participants and target populations \citep{hartman2015sate,lee2023improving}. These approaches establish the
identification and semiparametric foundations on which our inferential
construction builds. Our focus differs in that the principal nuisance object is estimated through a regularized transport problem rather than through an explicit model for the sampling propensity.

Balancing and calibration methods instead estimate weights by directly
matching functions of the covariates. Important examples include entropy
balancing, covariate balancing propensity scores, stable balancing weights, and globally efficient calibration estimators
\citep{hainmueller2012,imai2014covariate,zubizarreta2015stable,
chan2016globally}. Calibration has also been developed specifically for
transportability and trial generalization
\citep{josey2021transporting,lee2023improving,chen2023entropy}.
In particular, \citet{lee2023improving} propose calibration weights for
combining randomized and observational samples and develop a doubly robust, locally efficient augmented estimator with nonparametric sieve extensions. \citet{chen2023entropy} study entropy balancing for causal generalization when target information is available through covariate summaries. These results show that calibration itself is not new in trial generalization. Its role here is different. We begin with a weight induced by source-relaxed entropic transport, characterize the population distortion created by the transport regularization, and use the Riesz equations to correct that distortion while retaining the transport structure.

Our construction also builds on entropic and unbalanced OT. Entropic
regularization enables efficient Sinkhorn computation \citep{cuturi2013sinkhorn,peyre2019computational}, while unbalanced OT relaxes marginal constraints through divergence penalties
\citep{chizat2018unbalanced}. The source-relaxed formulation used here is a semi-unbalanced problem in which only the trial marginal is relaxed because the target distribution defines the estimand and must remain fixed. These OT ingredients are therefore not themselves contributions of the present paper. Our interest is in the statistical target of the trial-side marginal and in what fixed regularization implies for estimation of the target-to-trial density ratio.

Within the causal OT literature, \citet{dunipace2021optimal} is closest to the present work. Causal Optimal Transport constructs weights by minimizing a Sinkhorn discrepancy between weighted empirical distributions and can additionally impose balance restrictions on prespecified covariate functions. It connects OT weighting to importance weighting, distributional balance, and semiparametric causal estimation. \citet{gunsilius2022matching} use multimarginal unbalanced OT for causal matching, allowing poorly matched observations to receive less mass, while related work has developed OT methods for other causal estimands and counterfactual distributional problems \citep{torous2024optimal,gunsilius2025primer}. More recently,
\citet{yan2024reducing} relate Wasserstein discrepancy to causal balancing
error and jointly learn transport weights and representations. Our distinction from this literature is not the use of OT or covariate restrictions by themselves. We study the population weight produced by source-relaxed entropic OT at fixed regularization, quantify its discrepancy from the Riesz representer required by the transported estimand, and develop growing Riesz calibration that removes this discrepancy without requiring the transport regularization to vanish.

Finally, our analysis is related to direct estimation of Riesz representers. \citet{chernozhukov2021automatic} develop Riesz regression for automatic debiased machine learning, while \citet{chernozhukov2026adversarial} develop adversarial estimators of Riesz representers over general function classes with semiparametric inference guarantees. In our setting, the relevant Riesz representer has the explicit interpretation \(r^\dagger=dQ_X/dP_X\). Rather than estimating it through a separate regression or adversarial criterion, we begin with the weight induced by regularized OT and calibrate that weight to the Riesz equations. The resulting estimator therefore combines the geometric structure of transport with the moment representation required for efficient causal inference.
 
\subsection{Organization}
Section~\ref{sec:identification} establishes identification of the target
estimand, and Section~\ref{sec:eif} derives the efficient influence function and exact remainder. Section~\ref{sec:riesz} develops the Riesz-calibrated OT weight, followed by estimation and inference in Section~\ref{sec:estimation} and computation and tuning in Section~\ref{sec:implementation}.
Section~\ref{sec:rate} gives the rate-scale extension,
Sections~\ref{sec:sims} and \ref{sec:realdata} present the simulation and data analyses, and Section~\ref{sec:discussion} concludes. Proofs and the formal versions of all statements are in the appendices, whose organization is
described at the head of Appendix~\ref{app:A}.
 
\section{Causal Identification}
\label{sec:identification}
 
We observe two independent samples. The randomized trial consists of
\(O_i^{\mathrm R}=(X_i,T_i,Y_i)\), \(i=1,\ldots,n\), drawn independently from
\(P\), where \(X_i\) denotes baseline covariates, \(T_i\in\{0,1\}\) is
randomized treatment, and \(Y_i=Y_i(T_i)\) is the observed outcome. The target
RWD sample consists of \(O_j^{\mathrm Q}=(X^Q_j,Y_j^{\mathrm Q})\),
\(j=1,\ldots,m\), drawn independently from \(Q\); all target individuals are
treated, so \(Y_j^{\mathrm Q}=Y_j^{\mathrm Q}(1)\).
Our target estimand is
\begin{equation}
\tau_Q
\coloneqq
\EQ\!\left[Y^{\mathrm Q}(1)-Y^{\mathrm Q}(0)\right],
\label{eq:tauQ}
\end{equation}
the average treatment effect in the treated real-world target population.
Expectations under the two laws are written \(\EP\) and \(\EQ\), with
\((X,T,Y)\) reserved for a generic trial observation and
\((X^Q,Y^{\mathrm Q})\) for a generic target observation; thus \(\EP[f(X)]\)
integrates \(f\) against the trial covariate law and \(\EQ[f(X)]\) against the
target covariate law. Their empirical counterparts are the sample averages denoted by 
\(\E_n\) and \(\E_m\). Set \(N=n\wedge m\) and suppose
\(n/m\to\eta\in(0,\infty)\), so that the two samples grow at comparable rates.
For the trial, define $e_0(x)\coloneqq \P(T=0\mid X=x)$ to be the randomization probability of being assigned to control,
known by design in a trial; $A \coloneq \frac{1-T}{e_0(X)}$ is the corresponding inverse-probability
weight, which is zero on the treated arm and \(1/e_0(X)\) on the control arm;
and $\mu_0(x) \coloneqq \EP(Y\mid X=x,\,T=0)$ is the \emph{control response surface}, the object the trial
supplies and the target sample cannot.
 
\begin{assumption}[Causal identification]\label{ass:transport}
(i) \emph{Consistency:} \(Y=TY(1)+(1-T)Y(0)\) in the trial, and
\(Y^{\mathrm Q}=Y^{\mathrm Q}(1)\) in the target sample Q-almost surely. 
(ii)  \corr{$0<\underline e \le \overline e<1$},
\corr{$\underline e\le e_0(x) \le \overline e$} for $P_X$-almost every $x$. 
(iii) \emph{Transportability of the control response:} \(\EQ[Y^{\mathrm Q}(0)\mid X^Q=x]=\mu_0(x)\) for \(Q\)-almost every \(x\).
(iv)  \emph{Population overlap:} the target covariate law $Q_X$ is absolutely
continuous with respect to the trial covariate law $P_X$, with bounded density ratio
\(\rdag\). 
\end{assumption}
\corr{Moment conditions sufficient for the limit theory are collected in Assumption~\ref{ass:id} of the Supplementary Material, of which this assumption is an informal counterpart.}
 Conditions (i)--(ii) are internal to the trial and identify $\mu_0$ from
randomized data. Condition (iii) is the only cross-population assumption; it
restricts the control arm alone, because the treated response in the target
population is observed. Condition (iv) guarantees that $\mu_0$, identified
only where the trial has covariate support, is defined wherever the target
population lives. \corr{Randomization of treatment within the trial holds by
design and is therefore not listed as a separate assumption: $\mu_0$ is
defined from the observed data as the control-arm regression, and
randomization is what gives it the causal reading
$\mu_0(x)=\EP\{Y(0)\mid X=x\}$, making (iii) a substantive transportability
condition rather than a definition.}
 
\begin{proposition}[Identification]
\label{prop:id}
Under Assumption~\ref{ass:transport},
\begin{equation}
\tau_Q
=
\EQ\bigl[Y^{\mathrm Q}\bigr]-\EQ\bigl[\mu_0(X^Q)\bigr]
\label{eq:identification}
\end{equation}
\end{proposition}
 
The treated arm of \eqref{eq:tauQ} is observed in the
population where the estimand is defined, so consistency gives
\(\EQ[Y^{\mathrm Q}(1)]=\EQ[Y^{\mathrm Q}]\) directly. For the control arm,
condition on the target covariate and replace the unobserved conditional mean
by the trial regression $\EQ\bigl[Y^{\mathrm Q}(0)\bigr]
=\EQ\Bigl[\EQ\bigl\{Y^{\mathrm Q}(0)\mid X^Q\bigr\}\Bigr]
=\EQ\bigl[\mu_0(X^Q)\bigr],$ the first equality by the tower property and the second by
Assumption~\ref{ass:transport}(iii). Randomization\corr{ --- which holds by design ---} and positivity identify
\(\mu_0\) from the trial, and overlap makes it defined \(Q\)-almost everywhere.
See Proposition~\ref{prop:C-id} for the detailed proof. 
 
Proposition~\ref{prop:id} reduces identification to an observed treated mean in
the target population and a trial control regression averaged under the target
covariate law. The next section asks how efficiently that average can be
estimated, and thereby isolates the weighting function the problem requires.
 
\section{Efficient Estimation and the Riesz Representer}
\label{sec:eif}
 
\subsection{Influence functions and the efficiency bound}
 
We start with recalling the notions of the influence bound and the efficiency bound used throughout. An
estimator \(\hat\tau\) of \(\tau_Q\) is \emph{asymptotically linear} with
\emph{influence function} \(\varphi\) if it can be written, up to a negligible
term, as a sample average of a fixed function of the data $\{O_i, i = 1,\ldots,N\}$. i.e.
$$\hat\tau-\tau_Q=N^{-1}\sum_i\varphi(O_i)+o_p(N^{-1/2}), \qquad \text{with}\qquad \E[\varphi]=0.$$ The influence function determines the limit distributions, and by the central limit theorem, we then get asymptotic variance
\(\Var\{\varphi\}\). So comparing estimators amounts to comparing influence
functions.
 
In a \emph{nonparametric} model, where the data distribution is unrestricted,
there is a unique influence function that is a derivative of the target
parameter along every smooth one-dimensional submodel through the truth. It is
called the \emph{efficient influence function}, and its variance is the
\emph{semiparametric efficiency bound}.
Because our data arrive in two independent samples, one from RCT and the other from the RWD, the influence function has
two components, one per sample, and the bound combines them with a weight
reflecting the relative sample sizes.
 
The second notion is representation. If \(L\) is a bounded linear functional on
a Hilbert space \(H\), the Riesz representation theorem supplies a unique
\(r\in H\) with \(L(h)=\langle r,h\rangle\) for all \(h\in H\); \(r\) is the
\emph{Riesz representer} of \(L\). Take \(H=L^2(P_X)\) with inner product
\(\langle f,g\rangle=\E_{P_X}[f(X)g(X)]\), and let \(L\) be the functional that
averages a covariate function over the \emph{target} population,
\(L(h)=\E_{Q_X}[h(X)]\). Under Assumption~\ref{ass:transport}(iv) this functional is
bounded, and comparing with \eqref{eq:riesz-intro} shows that its Riesz
representer is exactly the density ratio \(\rdag\). Later Lemma~\ref{lem:A-riesz}
also shows it is the \emph{only} element of \(L^2(P_X)\) with that property. This
uniqueness is used in Section~\ref{sec:estimation}.

\subsection{The efficient influence function}
Related efficient-influence-function constructions for transported trial
effects are given by \citet{dahabreh2019extending,dahabreh2020extending}; the
present two-sample structure differs because the treated outcome is observed
directly in the target sample.
\begin{theorem}[Efficient Influence Function]
\label{thm:eif}
Under Assumption~\ref{ass:transport}, the efficient influence function for
\(\tau_Q\) has target and trial components
\[
\varphi_Q(O^{\mathrm Q})\coloneqq Y^{\mathrm Q}-\mu_0(X^Q)-\tau_Q,
\qquad
\varphi_P(O^{\mathrm R})\coloneqq -A\,\rdag(X)\{Y-\mu_0(X)\},
\]
and, under \(n/m\to\eta\in(0,\infty)\), the efficiency bound for
\(\sqrt m\,(\hat\tau-\tau_Q)\) is
\[
V_{\mathrm{eff}}
\coloneqq
\Var_Q\bigl[\varphi_Q(O^{\mathrm Q})\bigr]
+\eta^{-1}\Var_P\bigl[\varphi_P(O^{\mathrm R})\bigr].
\]
\end{theorem}
 
Perturbing the target law moves the outer expectation in \eqref{eq:identification} and yields \(\varphi_Q\) immediately. Perturbing the trial law moves only \(\mu_0\), and the resulting derivative \(\dot\mu_0\) must be transferred to the trial sample. The Riesz identity \eqref{eq:riesz-intro} carries the average from the target law to the trial law at the cost of a factor \(\rdag\), and the inverse-randomization identity of Lemma~\ref{lem:A-invrand} turns a conditional expectation given \(\{X,T=0\}\) into an unconditional one weighted by \(A\),
giving $\EQ\bigl[\dot\mu_0(X)\bigr]
=\EP\bigl[\rdag(X)\,\dot\mu_0(X)\bigr]
=\EP\bigl[A\,\rdag(X)\{Y-\mu_0(X)\}\,\zeta_P\bigr]$
where \(\zeta_P\) is the score of the trial submodel. The right-hand side is
\(-\EP[\varphi_P\zeta_P]\), which identifies \(\varphi_P\). See
Lemma~\ref{lem:C-deriv} and Theorem~\ref{thm:C-eif} and their proofs for more details.
 
Two features of Theorem~\ref{thm:eif} shape everything that follows. The trial
component is a weighted residual, so an estimator built from it will need an
estimate of \(\mu_0\). The weight in that component is \(\rdag\), and by
Remark~\ref{rem:C-riesz} it is there precisely as the Riesz representer of the
target averaging functional. 
 
\subsection{The estimating functional and its exact error}
 
Theorem~\ref{thm:eif} suggests estimating \(\tau_Q\) by the sample analogue of
the following functional, in which \(\mu\) and \(r\) stand for arbitrary
candidate values of the two nuisance functions \(\mu_0\) and \(\rdag\):
\begin{equation}
\Psi(\mu,r)
=
\EQ\bigl[Y^{\mathrm Q}-\mu(X^Q)\bigr]
-
\EP\bigl[A\,r(X)\{Y-\mu(X)\}\bigr].
\label{eq:onestep-functional}
\end{equation}
The first term standardizes the outcome model over the target population; the
second is a weighted average of trial control residuals, and it is what
corrects the first term when the outcome model is imperfect.
 
\begin{proposition}[Exact second-order remainder]
\label{prop:remainder}
Under Assumption~\ref{ass:transport}, for square-integrable \((\mu,r)\),
\begin{equation}
\Psi(\mu,r)-\tau_Q
=
\EP\bigl[\{r(X)-\rdag(X)\}\{\mu(X)-\mu_0(X)\}\bigr].
\label{eq:exact_remainder}
\end{equation}
\end{proposition}
 
 Randomization converts the trial term
into an average of \(r\) against the outcome-model error, whatever \(r\) may
be, and the Riesz identity converts the target term into an average of
\(\rdag\) against the same error, whatever \(\mu\) may be; subtracting leaves a
product:
\begin{equation*}
\underbrace{\EQ\bigl[\mu_0(X^Q)-\mu(X^Q)\bigr]}_{=\;\EP[\rdag(\mu_0-\mu)]}
\;-\;
\underbrace{\EP\bigl[A\,r(X)\{Y-\mu(X)\}\bigr]}_{=\;\EP[r(\mu_0-\mu)]}
\;=\;
\EP\bigl[(\rdag-r)(\mu_0-\mu)\bigr].
\end{equation*}
See Proposition~\ref{prop:C-rem} and its proof for more details. 

Notice that the identity \eqref{eq:exact_remainder} is exact. There is no first-order term and no expansion. Three
consequences, drawn out in Remark~\ref{rem:C-product}, organize the rest of the
paper. The error vanishes if \emph{either} factor does, which is double
robustness. Its size is the \emph{product} of the two nuisance errors, so
root-\(N\) inference needs only that product to be small, not each factor
separately. And \eqref{eq:exact_remainder} is the same identity whether \(r\)
comes from a fitted sampling score or from a transport plan, so the comparison
between such procedures reduces entirely to which approximates \(\rdag\) better
in \(L^2(P)\). The next section builds a transport-based estimate of \(\rdag\).

\section{Riesz-Calibrated Entropic Optimal Transport}
\label{sec:riesz}
\subsection{Couplings, entropic regularization, and smoothing}
 
A \emph{coupling} of the trial and target covariate distributions is a joint
distribution on \(\calX\times\calX\) whose two marginals are those
distributions; we represent it by a density \(\pi(x,z)\ge0\) with marginals
\(\pi_1\) and \(\pi_2\) as in \eqref{eq:marginals}. Optimal transport selects
the coupling minimizing an average cost \(\iint c(x,z)\pi(x,z)\,dx\,dz\). We will be using \(c(x,z)=\norm{x-z}{2}^2\) throughout.
 
Two modifications are needed. First, the linear program above is expensive and
its solution is a sparse, unstable object, so one adds an \emph{entropic
penalty}: for nonnegative functions \(f\) and \(g\) the Kullback--Leibler
divergence \(\KL(f\|g)\) of \eqref{eq:genKL} measures how far \(f\) is from
\(g\), and penalizing \(\KL(\pi\|p\otimes q)\) with weight \(\varepsilon>0\)
makes the problem strictly convex and solvable by the Sinkhorn algorithm
\citep{cuturi2013sinkhorn}. Second, as explained in
Section~\ref{sec:intro}, we hold the target marginal fixed at \(q\) and relax
the trial marginal by a further penalty \(\rho\,\KL(\pi_1\|p)\) with weight
\(\rho\ge0\). The resulting semi-unbalanced entropic problem is given as follows.
\begin{definition}[Unbalanced entropic transport]
\label{def:D-ubot}
Fix $\varepsilon>0$ and $\rho\ge0$. The unbalanced entropic plan is the
nonnegative density $\pi$ solving
\begin{equation}
  \min_{\pi\ge0,\ \pi_2=q}
  \Bigl\{\iint c(x,z)\,\pi(x,z)\,dx\,dz
        +\varepsilon\,\KL\bigl(\pi\,\|\,p\otimes q\bigr)
        +\rho\,\KL\bigl(\pi_1\,\|\,p\bigr)\Bigr\},
  \label{eq:ubot}
\end{equation}
where $(p\otimes q)(x,z)\coloneq p(x)q(z)$. The target marginal is constrained
exactly. The source marginal is free to deviate from $p$ at cost $\rho$. We
write
\begin{equation}
  r_{\varepsilon,\rho}(x)\coloneq \frac{\pi_1(x)}{p(x)}
  \label{eq:D-weightdef}
\end{equation}
for the induced \emph{source weight} (\emph{trial-side weight}), and $r_{\varepsilon,0}$ for the weight
obtained at $\rho=0$.
\end{definition}
Balanced OT fixes
\(\pi_1=p\) and hence yields the constant weight one. Relaxing the trial
marginal is what permits a non constant estimate of \(\rdag\). Throughout we
abbreviate $\gamma=\frac{\varepsilon}{\varepsilon+\rho}\in(0,1\corr{]}$\corr{, the exponent through which $\rho$ acts at a fixed dual potential ($\gamma=1$ corresponds to $\rho=0$)}.
 
For \(\varepsilon>0\), \emph{Gibbs smoothing}
\(S_\varepsilon\) is convolution with a Gaussian density of covariance
\((\varepsilon/2)I_d\) (explicitly defined in \eqref{eq:gibbs}); it averages a function over a neighbourhood of width \(\sqrt{\varepsilon/2}\). It appears because the exponential of the squared cost \(-c(x,z)/\varepsilon\) is a Gaussian kernel.
 
\subsection{Bias in uncalibrated weight estimates}
\begin{proposition}[Kernel representation]
\label{prop:nw}
\corr{Given the dual representation of Proposition~\ref{prop:D-closed}, t}he trial-side weight satisfies the exact fixed-point equation $r_{\varepsilon,\rho}
=
\Bigl[S_\varepsilon\Bigl\{
\frac{q}{S_\varepsilon(p\,r_{\varepsilon,\rho}^{-\rho/\varepsilon})}
\Bigr\}\Bigr]^{\gamma}$ and, when $\rho = 0$, $r_{\varepsilon,0}=S_\varepsilon\Bigl[\frac{q}{S_\varepsilon p}\Bigr]$. Moreover \(\EP[r_{\varepsilon,\rho}(X)]=1\) exactly\corr{. At a fixed dual
potential, $\rho$ acts through the exponent $\gamma$ as a power-law shrinkage
of the weight toward 1.
}.
\end{proposition}

The optimality conditions express the coupling as a Gibbs
kernel times a dual factor depending only on \(z\), and the constraint that the
target marginal equal \(q\) determines that factor as the reciprocal of a
smoothed, tilted trial density,
\begin{equation}
h(z)
=
\Bigl\{\kappa_\varepsilon\,
S_\varepsilon\bigl(p\,r_{\varepsilon,\rho}^{-\rho/\varepsilon}\bigr)(z)
\Bigr\}^{-1},
\label{eq:pf-nw}
\end{equation}
with \(\kappa_\varepsilon\) the Gaussian normalizing constant. Substituting
\eqref{eq:pf-nw} into the closed form of the trial-side weight makes
\(\kappa_\varepsilon\) cancel and produces the displayed fixed point. See
Propositions~\ref{prop:D-closed} and~\ref{prop:D-kernel}. 
 
Reading the \(\rho=0\) formula from the inside out: \(S_\varepsilon p\) smooths
the trial density at the target point, the ratio \(q/S_\varepsilon p\) is
therefore a density-ratio-like quantity with a smoothed denominator, and the
outer \(S_\varepsilon\) carries it back to the trial point.
Remark~\ref{rem:D-interp} explains in what sense this is, and is not, a
Nadaraya--Watson smoother, and why for \(\rho>0\) the equation is genuinely
nonlinear so that \(r_{\varepsilon,\rho}\) is not a fixed power of
\(r_{\varepsilon,0}\).
 
Both forms of regularization can move the population target of the weight. The
next result separates the two contributions and measures what they cost.

\begin{theorem}[Regularization bias at fixed \((\varepsilon,\rho)\)]
\label{thm:raw-bias}
Let
\(b_{\varepsilon,\rho}(x)
\coloneq r_{\varepsilon,\rho}(x)-r^\dagger(x)\).
\corr{Under the regularity and uniformity conditions of
Theorem~\ref{thm:D-laplace} in the Supplementary Material --- including the
uniform-in-\(\varepsilon\) control of the source-relaxation branch and of
\(\mathcal B_\varepsilon\) stated there ---}
as \(\varepsilon\to0\) and \(\rho/\varepsilon\to0\),
\[
b_{\varepsilon,\rho}(x)
=
\frac{\rho}{\varepsilon}\mathcal B_\varepsilon(x)
+
\frac{\varepsilon}{4}
\left\{
\Delta r^\dagger(x)
-
r^\dagger(x)\frac{\Delta p(x)}{p(x)}
\right\}
+
O\left\{
\left(\frac{\rho}{\varepsilon}\right)^2+\varepsilon^2
\right\},
\]
uniformly over compact subsets of the interior of \(\mathcal X\), where
\[
\mathcal B_\varepsilon(x)
=
S_\varepsilon\left[
\frac{
q\,S_\varepsilon\{p\log r_{\varepsilon,0}\}
}{
(S_\varepsilon p)^2
}
\right](x)
-
r_{\varepsilon,0}(x)\log r_{\varepsilon,0}(x).
\]
In particular,
\(r_{\varepsilon,\rho}-r^\dagger
=O(\rho/\varepsilon+\varepsilon)\)
\corr{uniformly on compact subsets of the interior of \(\calX\)}
in the stated joint regime\corr{; near the boundary of the covariate space
the entropic bias is generally of larger order (see
Remark~\ref{rem:D-boundary} in the Supplementary Material)}.
\end{theorem}
\corr{Theorem~\ref{thm:raw-bias} is diagnostic: it quantifies, in a
small-regularisation regime, the population distortion that calibration will
remove. It is invoked by no later result as this shows the need for caliberation, and the standing convention that
$(\varepsilon,\rho)$ remain fixed and positive is unaffected.}

Add and subtract \(r_{\varepsilon,0}\). The source-relaxation term is obtained by expanding the fixed-point equation of Proposition~\ref{prop:nw} in \(\rho/\varepsilon\) around zero, giving \(r_{\varepsilon,\rho}-r_{\varepsilon,0}
=(\rho/\varepsilon)\mathcal B_\varepsilon
+O\{(\rho/\varepsilon)^2\}\). The entropic term follows from the heat-kernel expansion applied to \(r_{\varepsilon,0}=S_\varepsilon\{q/(S_\varepsilon p)\}\), yielding the \(O(\varepsilon)\) term stated above. Adding the two expansions gives the result.

By the Riesz identity the weight bias passes directly into a violation of \eqref{eq:riesz-intro}, and by Proposition~\ref{prop:remainder} into the estimator through the product remainder. Sample size alone does not remove it;
Remark~\ref{rem:D-why} develops this point and explains what any repair must preserve. We now remove it by constraint.
 
\subsection{Riesz calibration}
 
Since \eqref{eq:riesz-intro} involves infinitely many test functions, we impose
it on a finite-dimensional subspace and let that subspace grow with the sample.
Such a growing family of finite-dimensional approximating spaces is called a
\emph{sieve}. Let \(b_J(x)=\{b_{J,1}(x),\ldots,b_{J,J}(x)\}^{\top}\) be a
vector of \(J\) basis functions with \(b_{J,1}\equiv1\), and let
\(H_J=\mathrm{span}\{b_{J,1},\ldots,b_{J,J}\}\). The \emph{calibration
equations} are
\begin{equation}
\EP\bigl[r(X)b_J(X)\bigr]=\EQ\bigl[b_J(X^Q)\bigr],
\label{eq:calib-main}
\end{equation}
equivalently, by Lemma~\ref{lem:A-riesz}, the statement that the weight error
\(r-\rdag\) is orthogonal to \(H_J\) in \(L^2(P)\). Calibration therefore
enforces \eqref{eq:riesz-intro} exactly on \(H_J\), without implying
\(r=\rdag\) pointwise; the consequences of that gap are taken up in
Section~\ref{sec:estimation}. Including the constant \(b_{J,1}\equiv1\) costs
nothing, since Proposition~\ref{prop:nw} shows the uncalibrated weight already
integrates to one.
\begin{definition}[Riesz-calibrated entropic OT]
\label{def:rcot}
The \textsc{RiCOT} plan minimizes the semi-unbalanced entropic objective
subject to the target-marginal constraint \emph{and} the calibration
constraints \eqref{eq:calib-main}; its trial-side weight is
\(r_{J,\varepsilon,\rho}=\pi_1/p\). The empirical version uses the entire
trial covariate sample, with trial reference masses \(a_i=n^{-1}\), target
masses \(\omega_j=m^{-1}\), and costs \(C_{ij}=c(X_i,X^Q_j)\), as given in
Definition~\ref{def:emp}. Its calibration constraint is exactly the sample
analogue
\[
\En\bigl[\rhat(X)b_J(X)\bigr]
=
\E_m\bigl[b_J(X^Q)\bigr]
\]
of \eqref{eq:calib-main}. Calibration is enforced \emph{simultaneously} with
the transport cost minimization; there is no post hoc recalibration step, and
it is this that produces the multiplicative structure below.
\end{definition}

\begin{proposition}[Multiplicative calibration tilt]
\label{prop:kkt}
At a strictly interior empirical solution there are target dual variables
\(v\in\R^m\) and a calibration multiplier \(\lambda\in\R^J\) with
\[
\rhat(x)
=
\Bigl[\sum_{j=1}^m\omega_j
\exp\Bigl\{\frac{v_j-c(x,X_j^Q)-\lambda^{\top}b_J(x)}{\varepsilon}\Bigr\}
\Bigr]^{\gamma}.
\]
Calibration therefore modifies the raw transport weight by a
\emph{multiplicative} exponential factor, preserving positivity and, up to a
bounded constant \corr{(a supremum-norm bound on the tilt, verified at the
solution in the Supplementary Material)}, the trial-side downweighting supplied by OT.
\end{proposition}

To see the proof of the above proposition, differentiate the Lagrangian of the empirical program in
\(\pi_{ij}\), attaching \(v_j\) to the target-marginal constraints and
\(\lambda\) to the calibration constraints. Stationarity at an interior
optimum reads
\begin{equation}
C_{ij}
+\varepsilon\log\frac{\pi_{ij}}{a_i\omega_j}
+\rho\log\frac{\alpha_i}{a_i}
+\lambda^{\top}b_J(X_i)-v_j=0,
\qquad \alpha_i=\sum_j\pi_{ij},
\label{eq:pf-kkt}
\end{equation}
where \(a_i=n^{-1}\). Solving \eqref{eq:pf-kkt} for \(\pi_{ij}\), summing
over \(j\), and using \(\alpha_i=a_i\rhat_i\) gives the display. See
Proposition~\ref{prop:E-kkt} and its proof for more details. 

That the tilt is exponential rather than additive is a consequence of the
Kullback--Leibler penalties, and it matters. Corollary~\ref{cor:E-tilt} and
Remark~\ref{rem:E-mult} show that a small baseline weight remains small after a
bounded multiplicative tilt, whereas a squared-loss calibration would add
\(\lambda^{\top}b_J\) to the weight and could return positive mass to units the
transport had deliberately downweighted. Proposition~\ref{prop:kkt} also
supplies the out-of-sample extension of \(\rhat\) to held-out trial covariates.

The next observation is the structural core of the construction. Recall that a
family of densities of the form \(e^{\,\text{offset}+\theta^{\top}b}\), indexed
by a finite-dimensional parameter \(\theta\) multiplying a fixed vector of
functions \(b\), is an \emph{exponential tilt} of the offset.

\begin{proposition}[Calibration is convex estimation in a tilt parameter]
\label{prop:tiltform}
Write \(\theta=-\lambda/(\varepsilon+\rho)\) and let \(\Lhat\) denote the
empirical transport offset \eqref{eq:offset-hat}. Then
\[
\rhat(x)=\exp\bigl\{\gamma\Lhat(x)+\theta^{\top}b_J(x)\bigr\},
\]
so the calibrated weight is an exponential tilt of the transport offset. Its
calibration equations are precisely the first-order conditions
\(\nabla M_n(\theta)=0\) of the convex objective
\[
M_n(\theta)
=
\En\bigl[e^{\gamma\Lhat(X)+\theta^{\top}b_J(X)}\bigr]
-\theta^{\top} \E_m \bigl[b_J(X^Q)\bigr].
\]
The objective is strictly convex \corr{exactly when the empirical Gram
matrix \(\En[b_Jb_J^\top]\) is nonsingular, and in that case} it has a
\corr{(unique)} minimizer \corr{if and only if}
\(\E_m[b_J(X^Q)]\) lies in the relative interior of the convex hull of the trial
design points \(\{b_J(X_i):i=1,\ldots,n\}\).
\end{proposition}

The multiplier \(\lambda^{\top}b_J(x)\) does not depend on
the target index \(j\), so it factors out of the sum in
Proposition~\ref{prop:kkt} before the outer power is taken:
\begin{equation}
\Bigl[e^{-\lambda^{\top}b_J(x)/\varepsilon}
\sum_{j}\omega_j e^{\{v_j-c(x,X^Q_j)\}/\varepsilon}\Bigr]^{\gamma}
=
e^{-\lambda^{\top}b_J(x)/(\varepsilon+\rho)}\;
e^{\gamma\Lhat(x)},
\label{eq:pf-tilt}
\end{equation}
where \(\gamma=\varepsilon/(\varepsilon+\rho)\). Differentiating \(M_n\)
then reproduces the calibration constraint, and the existence criterion is the
condition for \(M_n\) to be coercive. See
Proposition~\ref{prop:E-Mest} and its proof for more details. 

The decomposition into an \emph{offset} \(\gamma\Lhat\), carrying all the transport geometry, and a \emph{tilt} \(\theta^{\top}b_J\), carrying all the calibration, organizes the remaining analysis: smoothness is needed only of the offset, and is supplied by Lemmas~\ref{lem:A-lse} and~\ref{lem:dual-bdd}, while
the tilt is handled by finite-dimensional convex estimation. The feasibility criterion is also the practical diagnostic of Section~\ref{sec:implementation}: the calibration program stops having a solution when the target calibration moments leave the convex hull of the trial design points.

Let \( L \) denote the population transport offset~\ref{eq:offset}, and let \(a_{L,N}\downarrow0\) satisfy
\(\norm{\Lhat-L}{\infty}=O_p(a_{L,N})\).
Here, \(a_{L,N}\) measures the sampling error in estimating the transport offset. 

\newpage
\begin{theorem}[Calibration rate]
\label{thm:weight-rate}
Suppose \corr{Assumption~\ref{ass:id} holds, together with} the design, approximation, and dual-stability conditions of Assumption~\ref{ass:design} \corr{and two-sided overlap}. Let \(J_N\) denote the calibration dimension at sample
size \(N\). If, as \(N\to\infty\), along with
$
\sqrt{J_N}
\left\{
J_N^{-s/d}
+
\sqrt{\frac{J_N\log J_N}{N}}
+
a_{L,N}
\right\}
\to0,
$
then
\[
\|\rhat-\rdag\|_{L^2(P_X)}
=
O_p\left(
J_N^{-s/d}
+
\sqrt{\frac{J_N\log J_N}{N}}
+
a_{L,N}
\right).
\]
Balancing the first two terms gives
\(J_N\asymp
\left(\frac{N}{\log N}\right)^{d/(2s+d)}\) and thus we have
\(\|\rhat-\rdag\|_{L^2(P_X)}
=
O_p\left(
\left(\frac{\log N}{N}\right)^{s/(2s+d)}
+
a_{L,N}
\right)\).
\end{theorem}

To prove this result, we work with tilt parameter regime of Proposition~\ref{prop:tiltform}. Let
\(\theta_{J_N}\) approximate the population tilt
\(\log\rdag-\gamma L\) in \(H_{J_N}\). The gradient of \(M_n\) at
\(\theta_{J_N}\) satisfies
\[
\norm{\nabla M_n(\theta_{J_N})}{2}
=
O_p\left(
J_N^{-s/d}
+
\sqrt{\frac{J_N\log J_N}{N}}
+
a_{L,N}
\right),
\]
with contributions from sieve approximation, trial and target sampling variation, and estimation of the transport offset. The Hessian of \(M_n\) is a weighted empirical Gram matrix whose smallest eigenvalue is bounded away from zero with probability tending to one. Local strong convexity therefore gives
\begin{equation}
\norm{\hat\theta-\theta_{J_N}}{2}
\lesssim
\norm{\nabla M_n(\theta_{J_N})}{2}.
\label{eq:pf-rate}
\end{equation}
The growth condition in Theorem~\ref{thm:weight-rate} keeps the exponential tilt uniformly bounded, allowing this rate to be transferred from the log scale to \(\rhat\). See Theorem~\ref{thm:E-rate} and Corollaries~\ref{cor:E-J}--\ref{cor:E-window} and their proofs for more details. 

The three terms of the rate are, in order, sieve approximation error, calibration estimation error, and transport-offset estimation error. The regularization parameters affect the offset and the associated constants, but fixed positive \(\varepsilon\) and \(\rho\) do not generate a separate asymptotic bias term after calibration; see Remark~\ref{rem:E-three}.

Theorem~\ref{thm:weight-rate} shows that Riesz calibration converts the regularized transport weight into a consistent estimator of the target-to-trial density ratio.  
It separates the roles of the two kinds of regularization.Entropic regularization parameters \(\varepsilon\) and \(\rho\) stabilize the transport problem, while the growing calibration space controls the error that matters for inference. We now combine this weight with the outcome regression.
 
\section{Doubly Robust Estimation and Efficient Inference}
\label{sec:estimation}
 
\subsection{One-step estimation and cross-fitting}
 
A plug-in estimator of \(\tau_Q\) would substitute a fitted \(\muhat\) into
\eqref{eq:identification}. Such an estimator inherits the bias of \(\muhat\) at
first order. The \emph{one-step} construction removes that bias by adding the
sample average of the estimated influence function, which is a single Newton
step on the estimating equation; the resulting estimator is the sample analogue
of \(\Psi(\muhat,\rhat)\) in \eqref{eq:onestep-functional}, and by
Proposition~\ref{prop:remainder} its error is second order in the nuisances.
 
A second device is needed because \(\muhat\) is fitted on the same data used to
evaluate it, which can induce spurious correlation with the residuals it is
meant to center. \emph{Cross-fitting} avoids this: split the trial sample into
\(K\) folds, fit the nuisance on the other \(K-1\) folds, and evaluate it only
on the held-out fold. We use \(K\) fixed and average the fold-specific fits
where a single function is needed.
 
\subsection{The estimator}
 
Let \(I_1,\ldots,I_K\) partition the trial indices, let \(\muhat^{(-k)}\) be fitted on trial controls outside \(I_k\), and let \(\bar\mu_0(x)=K^{-1}\sum_k\muhat^{(-k)}(x)\). Let \(\rhat\) be the \textsc{RiCOT} weight of Definition~\ref{def:rcot}, fitted once using all trial
controls and all target covariates. The estimator is
\begin{equation}
\that
=
\frac1m\sum_{j=1}^m\bigl\{Y_j^{\mathrm Q}-\bar\mu_0(X_j^Q)\bigr\}
-\frac1n\sum_{k=1}^K\sum_{i\in I_k}
A_i\,\rhat(X_i)\bigl\{Y_i-\muhat^{(-k)}(X_i)\bigr\}.
\label{eq:drot}
\end{equation}
For inference, replace the unknowns in Theorem~\ref{thm:eif} by their estimates,
\[
\hat\varphi_{Q,j}=Y_j^{\mathrm Q}-\bar\mu_0(X_j^Q)-\that,
\qquad
\hat\varphi_{P,i}=-A_i\,\rhat(X_i)\bigl\{Y_i-\muhat^{(-k(i))}(X_i)\bigr\},
\]
and estimate the efficiency bound by the corresponding sample variances,
\begin{equation}
\hat V
=
\frac1m\sum_{j=1}^m\hat\varphi_{Q,j}^2
+\frac mn\cdot\frac1n\sum_{i=1}^n\hat\varphi_{P,i}^2 .
\label{eq:plugin}
\end{equation}
This requires a single pass over the data; no resampling and no re-solution of
the transport problem is involved. Write
\(a_N=\norm{\muhat-\mu_0}{L^2(P)}\) and \(b_N=\norm{\rhat-\rdag}{L^2(P)}\) for
the two nuisance errors.
 
\subsection{Double robustness}
 
An estimator is \emph{doubly robust} if it remains consistent when either one
of its two nuisance functions is estimated consistently, even if the other is
not.
 
\begin{theorem}[Double robustness]
\label{thm:dr-main}
Under Assumption~\ref{ass:transport} and foldwise laws of large numbers,
\(\that\xrightarrow{p}\tau_Q\) if either \(a_N=o_p(1)\) with \(b_N=O_p(1)\), or
\(b_N=o_p(1)\) with \(a_N=O_p(1)\). \corr{No rate condition on either
nuisance is needed beyond the displayed consistency and boundedness.}
\end{theorem}
 
The sample average converges to the population functional,
whose error is exactly the product \eqref{eq:exact_remainder}; Cauchy--Schwarz
bounds that product by the product of the two nuisance errors,
\begin{equation}
\bigl|\Psi(\muhat,\rhat)-\tau_Q\bigr|
\;\le\;
\norm{\rhat-\rdag}{L^2(P)}\;\norm{\muhat-\mu_0}{L^2(P)}
\;=\;b_N\,a_N,
\label{eq:pf-dr}
\end{equation}
which is \(o_p(1)\) as soon as one factor vanishes and the other is bounded.
See Theorem~\ref{thm:F-dr} and its proof for more details.
 
This branch of the theory survives a \emph{fixed} calibration dimension, at
which \(\rhat\) converges to some \(r_J\ne\rdag\) and \(b_N\) need not vanish.
 
\subsection{Asymptotic normality and inference}
 
\begin{theorem}[Efficient limit distribution and plug-in variance]
\label{thm:singlefit}
Suppose in addition that the nuisances satisfy the product-rate condition
\(a_Nb_N=o_p(N^{-1/2})\), that the trial residual is bounded, and that
\corr{\(J_N\log N=o(\sqrt N)\)} and \(\sqrt{J_N\log N}\,\max(a_N,b_N)=o_p(1)\). Then
\(\that\) is asymptotically linear with the influence functions of
Theorem~\ref{thm:eif}, and
\[
\sqrt m\,(\that-\tau_Q)\xrightarrow{d}\mathcal N(0,V_{\mathrm{eff}}),
\qquad
\hat V\xrightarrow{p}V_{\mathrm{eff}},
\]
so \(\that\pm z_{1-\alpha/2}\sqrt{\hat V/m}\) is an asymptotically valid Wald
interval.
\end{theorem}
 
The estimation error decomposes exactly into the two
efficient influence-function averages and three remainders,
\begin{equation}
\that-\tau_Q
=
(\Em-\EQ)\varphi_Q
+(\En-\EP)\varphi_P
+R_{1n}+R_{2n}+R_{3n},
\label{eq:pf-clt}
\end{equation}
where \(R_{3n}\) is the product remainder \eqref{eq:exact_remainder} and
\(R_{1n},R_{2n}\) are centered sample averages involving the fitted nuisances.
The product-rate condition kills \(R_{3n}\); cross-fitting makes the terms
involving \(\muhat\) conditionally centered with vanishing variance; and the
term involving \(\rhat\), which is \emph{not} cross-fitted, is controlled by a
maximal inequality over the exponential family of
Proposition~\ref{prop:tiltform}, which is a class of controlled complexity.
Two independent central limit theorems then give the limit. See
Theorem~\ref{thm:F-decomp}, Lemma~\ref{lem:F-rem} and
Theorem~\ref{thm:F-clt}; variance consistency is shown in Theorem~\ref{thm:F-var}, and
primitive sufficient conditions, amounting to \(\beta+s/(2s+d)>1/2\) with
\(s>d/2\) for an outcome-regression rate exponent \(\beta\), are discussed in
Corollary~\ref{cor:F-primitive}. 
 

\begin{remark}[Cross-fitting the transport weights]
\label{rem:main-crossfit}
Cross-fitting is essential for \(\muhat\), whose estimation error is otherwise correlated with the residuals it centers, but is not strictly required for \(\rhat\), which depends only on trial and target covariates. In the single-fit construction, the resulting empirical-process term is instead controlled by a maximal inequality over the exponential-tilting family. Cross-fitting \(\rhat\) nevertheless yields a simpler conditional argument, reduces sensitivity to in-sample overfitting and concentrated transport weights, and accommodates increasing calibration dimension or adaptive regularization without Donsker-type restrictions. It also removes the bounded-residual condition and the additional requirements
\(J_N\log N=o(\sqrt N)\), \(\sqrt{J_N\log N}\max(a_N,b_N)=o_p(1)\), the latter of which, under the rate-optimal choice of \(J_N\), entails the smoothness condition \(s>d/2\). The cost is that, with \(K\)-fold cross-fitting, the transport problem must be solved \(K\) times, where \(K\) denotes the number of folds.
\end{remark}

\begin{remark}[Practical choice of cross-fitting]
\label{rem:main-crossfit-practical}
With sufficiently large trial and target samples, we recommend cross-fitting both \(\muhat\) and \(\rhat\), typically using \(K=5\) or \(10\) folds, because the loss of training observations is modest and the foldwise construction provides greater robustness to overfitting. All tuning choices should be made within the training folds, and observations should be split at the independent-unit level. In small samples, however, cross-fitting may produce unstable nuisance estimates and transport weights because each training fold contains fewer observations. A single-fit weight estimator is then often preferable, provided that the complexity and rate conditions of Theorem~\ref{thm:singlefit} are plausible.
\end{remark}
 
 
 
\section{Computation and Tuning}
\label{sec:implementation}
 
To apply Riesz-Calibrated OT based estimation, three quantities must be chosen. The entropic parameter \(\varepsilon\), the
relaxation parameter \(\rho\), and the calibration dimension \(J\). The theory
assigns them sharply different roles, and this is the practical content of
Remark~\ref{rem:E-tuning}: \(J\) is the only parameter that appears in the
convergence rate, while \(\varepsilon\) and \(\rho\) select which stable member
of the calibrated family is computed.
 
\begin{algorithm}[ht]
\singlespacing
\caption{\textsc{RiCOT} with cross-fitted outcome regression}
\label{alg:main}
\begin{algorithmic}[1]
\REQUIRE Trial data $(X_i,T_i,Y_i)_{i=1}^n$; target data
$(X^Q_j,Y_j^{\mathrm Q})_{j=1}^m$; basis $b_J$; cost $c$; $\varepsilon,\rho$;
folds $K$
\STATE Standardize covariates using the prespecified trial-based scaling and
form the cost matrix $C_{ij}=c(X_i,X^Q_j)$.
\STATE Partition the trial indices into folds $I_1,\ldots,I_K$.
\FOR{$k=1,\ldots,K$}
  \STATE Fit $\muhat^{(-k)}$ using trial controls outside $I_k$.
\ENDFOR
\STATE Solve the calibrated transport problem of Definition~\ref{def:emp} once,
using all trial and the target real-world treated covariates, by alternating target-dual
Sinkhorn updates for $v$ with Newton updates for the tilt parameter $\theta$;
extract $\rhat$ and record the calibration residual and the plan and weight
ratios.
\STATE Average the fold-specific outcome predictions at each target covariate.
\STATE Compute $\that$ from \eqref{eq:drot} and $\hat V$ from
\eqref{eq:plugin}.
\RETURN $\that$, its Wald interval, and the diagnostics listed below.
\end{algorithmic}
\end{algorithm}
 
By Proposition~\ref{prop:tiltform}, step~5 is a strictly convex problem in the
\(J\)-dimensional tilt parameter for a given target dual, which is what makes
the alternating scheme well behaved; Proposition~\ref{prop:kkt} supplies the
out-of-sample extension of \(\rhat\). For the fully cross-fitted variant of
Remark~\ref{rem:main-crossfit}, the transport step is repeated within each
training fold while retaining the full target covariate sample.
 
\emph{Calibration dimension.} This is the parameter the theory determines.
Theorem~\ref{thm:weight-rate} trades approximation error \(J^{-s/d}\) against
estimation error \(\{J\log J/N\}^{1/2}\), balanced at the value of
\(J_N\asymp(N/\log N)^{d/(2s+d)}\) by Corollary~\ref{cor:E-J}; when a credible
outcome-regression rate exponent is available,
Corollary~\ref{cor:E-window} gives the admissible window directly. \corr{Three} upper
constraints bind. First, \(J_N^2\log J_N/N\to0\), that is
\(J_N=o(\sqrt{N/\log N})\), is required by Theorem~\ref{thm:weight-rate}
itself, because passing from the log scale on which the convex-estimation
argument operates back to the weight costs a factor \(\sqrt J\) in the supremum
norm. Second, by Proposition~\ref{prop:tiltform} the calibration program has a
solution only while \(\Em[b_J(Z)]\) remains inside the \corr{relative
interior of the} convex hull of the design points, which fails once the sieve
is too rich for the sample. \corr{Third, the offset error must stay small
relative to the sieve, $\sqrt{J}\,a_{L,N}\to0$; this is part of the growth
condition of Theorem~\ref{thm:weight-rate} and is automatic under root-$N$
offset stability.} These are
benchmarks rather than plug-in rules; we treat the displayed rate as a scale
and select within it using the diagnostics below.
 
\emph{Entropic regularization and source relaxation.} Neither is required to
vanish. By Remark~\ref{rem:E-tuning} they enter the asymptotics only through
the transport offset, whose stability is the content of
Assumption~\ref{ass:des-dual} and which contributes the term \(a_{L,N}\) to
Theorem~\ref{thm:weight-rate}; accordingly no condition of the form
\(\varepsilon_N\to0\) at a specified rate appears anywhere in the theory. We
select \(\varepsilon\) and \(\rho\) over prespecified bounded grids on the
scale of the transport cost, monitoring weight concentration, effective sample
size, and sensitivity of the final estimate. Should one wish to let either
parameter vary with \(N\), the stability constants in
Lemma~\ref{lem:dual-bdd} and the sequence \(a_{L,N}\) must be tracked
explicitly. No value of \(\rho\) substitutes for the support condition in
Assumption~\ref{ass:transport}(v).
 
\emph{Diagnostics.} Four quantities should be reported. The \emph{calibration
residual} of \eqref{eq:E-calib} is zero by construction when the constrained
problem is solved, so its departure from numerical zero is the numerical signature of the
convex-hull criterion in Proposition~\ref{prop:tiltform}, and is the natural
diagnostic for the upper edge of \(J\). The \emph{conditioning of the
calibration system}, that is the extreme eigenvalues of the empirical Gram
matrix, monitors Assumption~\ref{ass:des-basis}. The \emph{plan and weight
ratios} of Assumption~\ref{ass:des-dual} should be checked to be bounded away
from zero and infinity, since it is these quantitative bounds, and not mere
interiority of the solution, that Lemma~\ref{lem:dual-bdd} converts into the
offset stability the rate theory requires. Finally the \emph{effective sample
size} and covariate balance document the information cost of the overlap
restriction; for cross-fitted weights or an independent validation split,
held-out Riesz imbalance provides a further check.
Section~\ref{sec:sims} shows that effective sample size is not a safe
standalone diagnostic. We report sensitivity over a neighbourhood of
\((\varepsilon,\rho,J)\) alongside these quantities.
\section{Simulation Studies}
\label{sec:sims}
 
The theory yields four implications that we examine in finite samples. First,
Riesz calibration should recover the target weighting function even when a
parametric sampling-score model is misspecified. Second, with a sufficiently
rich calibration sieve, the plug-in variance should approach the
semiparametric efficiency bound, although overly rich calibration may lead to
instability. Third, when both working nuisance models are initially
misspecified, increasing the calibration dimension should restore valid
inference through consistency of the weighting branch, potentially at some
cost in efficiency. Finally, the source-relaxed estimator should remain stable
as overlap weakens, while commonly used weight diagnostics need not fully
reflect estimator performance. We study these implications under a common
simulation design, described in Section~\ref{sec:sim-design} and the
findings are summarized in Section~\ref{sec:sim-results}.
 
\subsection{Design}
\label{sec:sim-design}
 
Covariates are \(d=3\) dimensional, and the trial and target covariate laws
differ in one of three ways, ordered so that the Riesz representer \(\rdag\)
becomes progressively harder for a finite calibration space to represent. In
\(\mathrm{S}_0\) the two laws are Gaussian with common covariance and a mean
shift, so \(\log\rdag\) is affine and a logistic sampling score is correctly
specified; calibration ought to be unnecessary here, and the question is only
whether it costs anything. In \(\mathrm{S}_1\) a covariance tilt makes
\(\log\rdag\) quadratic, defeating the logistic score but still spanned exactly
by a degree-two Hermite basis. In \(\mathrm{S}_2\) the target is an asymmetric
two-component Gaussian mixture displaced from the trial law, so no finite basis
spans \(\log\rdag\) and the growing-sieve regime is genuine.
 
Outcomes are binary with control mean
\(\mu_0(x)=\mathrm{expit}(x^{\top}\vartheta+\tfrac12x^{\top}\Omega x)\) and an
\(x\)-varying treatment effect. Taking \(\Omega=0\) makes the fitted
main-effects logistic regression correct, while the starred scenarios
\(\mathrm{S}_0^{*},\mathrm{S}_1^{*},\mathrm{S}_2^{*}\) take \(\Omega\ne0\) so
that the same fitted model is wrong; since \(\mathrm{S}_1\) and
\(\mathrm{S}_2\) already defeat the logistic score, \(\mathrm{S}_1^{*}\) and
\(\mathrm{S}_2^{*}\) are the cells in which \emph{both} nuisances are
misspecified. An overlap index \(s_{\mathrm{ov}}\) scales the discrepancy
between the two laws --- the mean gap, the covariance tilt, and the mixture
displacement respectively --- with larger \(s_{\mathrm{ov}}\) meaning weaker
overlap and \(s_{\mathrm{ov}}=1\) the reference. Every configuration satisfies
\(\rdag\in L^2(P)\), and we verify the stronger fourth-moment condition that
stabilizes \(\hat V_P\).
 
The factorial crosses the six scenarios with
\(s_{\mathrm{ov}}\in\{0.5,1,1.5,2,2.5\}\), sample-size ratios
\(\eta=n/m\in\{0.5,2\}\), trial sizes \(n\in\{500,1000,2000,2500\}\) and
calibration dimensions \(J\in\{4,10,15,20\}\), giving \(960\) configurations at
\(500\) replicates each, or \(4.8\times10^{5}\) replicates in total, with
\(\P(T=0)=1/2\) throughout. We compare six estimators: the unadjusted
difference in means (\textsc{Naive}), outcome-regression standardization (\textsc{G-comp}), inverse
weighting by a logistic sampling-score density ratio (\textsc{IPW-PS}) or by
the calibrated OT weight (\textsc{IPW-OT}), and the two augmented estimators
built on those weights, \textsc{AIPW-PS} and the proposed \textsc{RiCOT}. The
outcome regression is cross-fitted with \(K=5\) folds; the OT weight is
computed once on the full trial sample, which Remark~\ref{rem:F-crossfit}
justifies. Both weightings are fit on the same source sample, so any difference
between them reflects the weighting method rather than the information used to
learn it. A detailed form of all the estimators have been provided in Section \ref{app:estimators} in the Appendix.
 
Two features of the implementation deserve mention because they affect what the
numbers can support. First, \(\tau_Q\), \(V_Q\) and \(V_P\) are computed
\emph{exactly} rather than by Monte Carlo: the tilted-Gaussian identity
\(q^2/p\propto\mathcal N(m_*,\Lambda^{-1})\) reduces \(V_P\) to the expectation
of a bounded function under a light-tailed proposal, whereas direct simulation
of \(V_P\) is dominated by the tail of \((\rdag)^2\) and converges far too
slowly to certify a variance ratio to three decimals. Second, the solver
returns the calibration residual \eqref{eq:E-calib}, which sits at numerical
zero when the calibration system is solvable and jumps to \(O(10^{-1})\) when
\(J\) is too large for the sample, exactly as Proposition~\ref{prop:tiltform}
predicts; we report only configurations whose residual falls below \(10^{-6}\)
and treat the failure boundary itself as an object of study. Reported
quantities are Monte Carlo bias, standard deviation, root mean squared error,
RMSE relative to the oracle standard deviation
\(\{V_{\mathrm{eff}}/m\}^{1/2}\), coverage of the nominal \(95\%\) Wald
interval built from \eqref{eq:plugin}, and the variance ratios
\(\hat V_Q/V_Q\), \(\hat V_P/V_P\) and \(\hat V/V_{\mathrm{eff}}\); with
\(500\) replicates a coverage estimate at the nominal level carries a standard
error of \(0.010\).
 
\subsection{Results}
\label{sec:sim-results}
 
\paragraph{Calibration restores consistency under sampling-score misspecification.}
Figure~\ref{fig:sim-benchmark} reports the reference configuration
\(s_{\mathrm{ov}}=1\), \(n=2000\), \(\eta=0.5\), \(J=10\).
 
Three patterns are visible. In \(\mathrm{S}_0\), where the logistic score is
correct, every adjusted estimator is unbiased and both augmented estimators
cover at the nominal rate; calibration costs nothing. In the unstarred
scenarios \(\mathrm{S}_1\) and \(\mathrm{S}_2\) the score is misspecified but
the outcome model is not, and double robustness protects \textsc{AIPW-PS}
exactly as Theorem~\ref{thm:dr-main} says it should --- its bias is negligible
even though \textsc{IPW-PS}, which has no outcome model to fall back on,
carries a bias of \(1.8\times10^{-2}\) in \(\mathrm{S}_2\).
 
The separation appears when both models fail. In \(\mathrm{S}_1^{*}\) the
\textsc{AIPW-PS} bias is \(-5.96\times10^{-2}\), some sixty times that of
\textsc{RiCOT}, and its intervals cover in \(2.2\%\) of replicates showing that the Wald
interval is not merely inefficient but inferentially void. \textsc{RiCOT}
retains a bias of \(-1.0\times10^{-3}\) and coverage of \(0.942\). The same
ordering, less extreme, holds in \(\mathrm{S}_2^{*}\) (\(0.880\) versus
\(0.946\)). Across all six scenarios \textsc{RiCOT} is the only estimator whose
bias never exceeds \(7\times10^{-3}\) and whose coverage stays within
\([0.940,0.956]\).
 
That \(\mathrm{S}_1^{*}\) is worse than \(\mathrm{S}_2^{*}\) is not an accident
of tuning. Double misspecification biases the one-step estimator through the
product remainder \eqref{eq:exact_remainder}, and its magnitude depends on
whether the two errors align. In \(\mathrm{S}_1^{*}\) the covariance tilt and
the outcome curvature push in the same direction and compound; in
\(\mathrm{S}_2^{*}\) the symmetric mixture leaves them partially orthogonal.
The lesson is that ``both models wrong'' is not a single regime, and reporting
only a favourable instance of it would understate the exposure of score-based
estimators.
 
\begin{figure}[t]
\centering
\includegraphics[width=\textwidth]{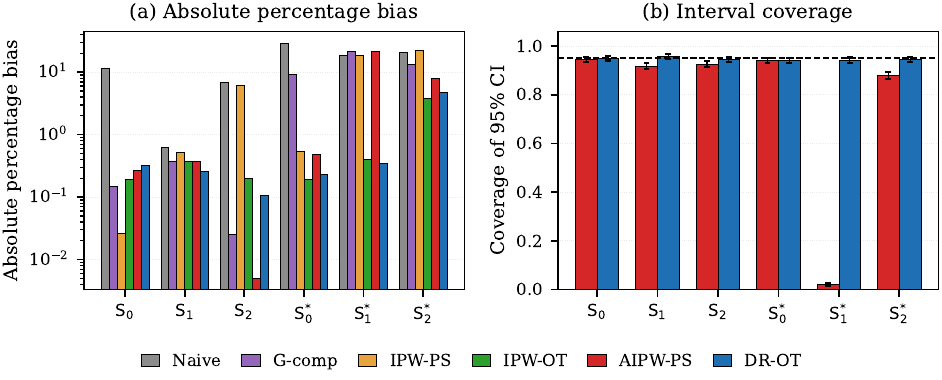}
\caption{Reference configuration ($s_{\mathrm{ov}}=1$, $n=2000$, $\eta=0.5$,
$J=10$; $500$ replicates). (a) Absolute bias on a logarithmic scale.
(b) Coverage of nominal $95\%$ Wald intervals for the two augmented estimators,
with Monte Carlo standard errors; the dashed line is the nominal level. Starred
scenarios have a misspecified outcome model. In $\mathrm{S}_1^{*}$ and
$\mathrm{S}_2^{*}$ both nuisances are misspecified.}
\label{fig:sim-benchmark}
\end{figure}
 
\paragraph{Plug-in variance recovery and the calibration-dimension tradeoff.}
Figure~\ref{fig:sim-efficiency} plots \(\hat V/V_{\mathrm{eff}}\) against
\(n\) for the correctly specified scenarios. For an adequate calibration space
the ratio approaches one from above as \(n\) grows. At \(J=10\) and \(n=2500\)
it is \(1.007\) in \(\mathrm{S}_1\), \(1.027\) in \(\mathrm{S}_0\) and
\(1.064\) in \(\mathrm{S}_2\), with coverage between \(0.935\) and \(0.953\)
throughout. The ordering matches the difficulty of the representer, and the
residual excess in \(\mathrm{S}_2\) is consistent with an approximation error
still shrinking at the largest sample size we simulate.
 
Two distinct failure modes bracket the useful range of \(J\), and they are
qualitatively different. A calibration space that is too coarse produces a
variance estimate biased \emph{downward} that does not improve with sample
size. At \(J=4\) the ratio \(\hat V_P/V_P\) equals \(0.794\) in
\(\mathrm{S}_1\) and \(0.827\) in \(\mathrm{S}_2\), drifting further from one
as \(n\) increases, because more data merely sharpens an estimate of a
functional the four-term basis cannot represent. 
The proportion of configurations with residual above \(10^{-6}\) rises from
\(0\%\) at \(J=4\), \(n\ge2000\) to \(45\%\) at \(J=20\), \(n=500\), and in
those cells \(\hat V_P/V_P\) takes values as large as four. This is the
empirical signature of the two upper constraints of
Section~\ref{sec:implementation}. The condition \(J_N^2\log J_N/N\to0\) of
Theorem~\ref{thm:weight-rate} and the convex-hull feasibility criterion of
Proposition~\ref{prop:tiltform}.
 
There is also a genuine cost to over-enrichment even when calibration succeeds.
In \(\mathrm{S}_0\), whose representer is log-linear,
\(\hat V/V_{\mathrm{eff}}\) deteriorates monotonically from \(0.962\) at
\(J=4\) to \(1.218\) at \(J=20\). The practical implication is
that \(J\) should track the complexity of the representer rather than be made
as large as the sample permits, and that the calibration residual --- available
at no extra cost --- is the natural diagnostic for the upper edge.
 
\begin{figure}[t]
\centering
\includegraphics[width=\textwidth]{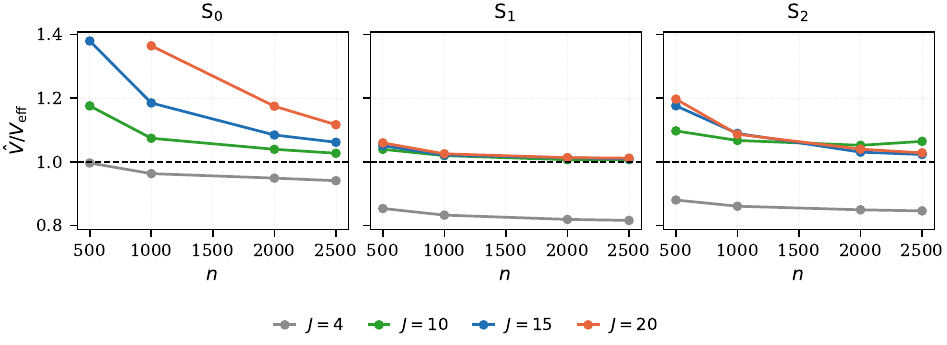}
\caption{Ratio of the plug-in variance to the semiparametric efficiency bound
against trial size, at $s_{\mathrm{ov}}=1$, for correctly specified outcome
models. Only configurations whose calibration residual is below $10^{-6}$ are
shown.}
\label{fig:sim-efficiency}
\end{figure}
 
\paragraph{Growing calibration restores validity but not efficiency under
outcome-model misspecification.}
The preceding comparison fixed \(J=10\). Because the calibrated weight is not a
fixed model but a sieve estimator, the interesting question in
\(\mathrm{S}_1^{*}\) and \(\mathrm{S}_2^{*}\) is what happens as \(J\) grows.
 
Coverage of \textsc{RiCOT} rises from \(0.823\) at \(J=4\) to \(0.954\) at
\(J=20\) in \(\mathrm{S}_1^{*}\), and from \(0.559\) to \(0.958\) in
\(\mathrm{S}_2^{*}\): enriching the calibration space repairs inference even
though the outcome model remains wrong throughout, because a consistent weight
is by itself sufficient for the one-step estimator. Coverage of
\textsc{AIPW-PS} is flat in \(J\) --- \(0.22\) in \(\mathrm{S}_1^{*}\) and
\(0.89\) in \(\mathrm{S}_2^{*}\) --- since no amount of computation improves a
parametric score that converges to the wrong limit.
 
The efficiency ratio behaves differently, and the difference is the point.
\(\hat V/V_{\mathrm{eff}}\) does not return to one; it plateaus strictly above,
at \(1.065\) in \(\mathrm{S}_1^{*}\) and \(1.182\) in \(\mathrm{S}_2^{*}\).
This is what Corollary~\ref{cor:F-fixedJ} and Remark~\ref{rem:F-summary} lead
one to expect that with a misspecified outcome regression the influence function is
no longer the canonical one, so the estimator is regular and asymptotically
linear but not efficient, and the plug-in variance correctly estimates the
larger sampling variance rather than the bound. The practical reading is that
under double misspecification the analyst recovers valid inference at a
quantified price of six to eighteen percent in variance, and that the price is
visible in the reported diagnostics rather than hidden.
 
\begin{figure}[t]
\centering
\includegraphics[width=\textwidth]{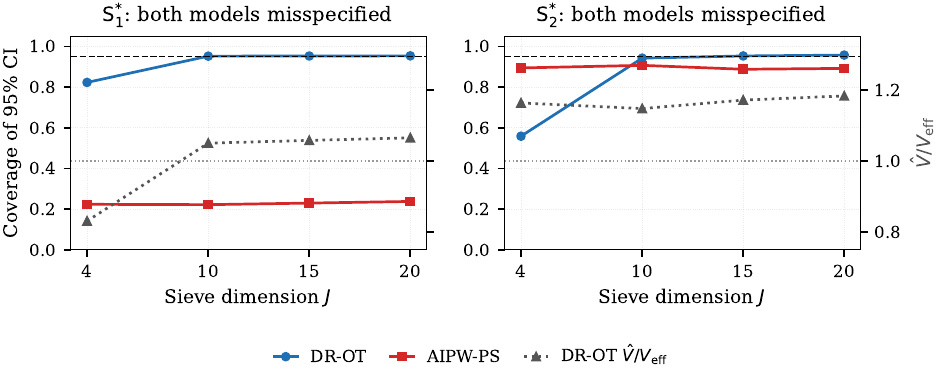}
\caption{Both nuisances misspecified ($s_{\mathrm{ov}}=1$, $n=2000$,
$\eta=0.5$). Left axis: coverage of nominal $95\%$ intervals against
calibration dimension. Right axis (grey, dotted): $\hat V/V_{\mathrm{eff}}$ for
\textsc{RiCOT}. Growing the calibration space restores nominal coverage for
\textsc{RiCOT} while the efficiency ratio plateaus above one; \textsc{AIPW-PS}
is unaffected by $J$.}
\label{fig:sim-doublemis}
\end{figure}
 
\paragraph{Robust performance under weakening overlap, and the limits of
effective sample size.}
Figure~\ref{fig:sim-overlap} follows the estimators as the overlap index
\(s_{\mathrm{ov}}\) increases from \(0.5\) to \(2.5\). The absolute bias of
\textsc{RiCOT} grows from \(8\times10^{-4}\) to \(7.8\times10^{-3}\) while its
coverage remains between \(0.930\) and \(0.955\); that of \textsc{AIPW-PS}
grows from \(1.1\times10^{-2}\) to \(5.2\times10^{-2}\) and its coverage falls
monotonically from \(0.779\) to \(0.579\). Root mean squared error of
\textsc{AIPW-PS} exceeds that of \textsc{RiCOT} by \(25\)--\(47\%\) across the
range. The multiplicative shrinkage of the semi-unbalanced formulation is doing
the work that Remark~\ref{rem:D-shrink} and Corollary~\ref{cor:E-tilt}
attribute to it.
 
Panel (c) contains a caution worth stating explicitly. As overlap deteriorates,
the effective sample size of the calibrated weights falls sharply, from \(828\)
to \(134\), while that of the sampling-score weights declines only from \(899\)
to \(644\). It would be natural to read the more stable diagnostic as the
better-behaved procedure. The opposite is true that the score weights remain
diffuse precisely because the logistic model fails to detect the covariate
shift, and their apparent stability accompanies a bias four to seven times
larger and coverage approaching \(0.58\). Effective sample size measures
concentration of the weights, not correctness, and under misspecification a
reassuring value can indicate under-correction rather than efficiency. We
recommend reporting it alongside the calibration residual and a balance check,
never alone.
 
\begin{figure}[t]
\centering
\includegraphics[width=\textwidth]{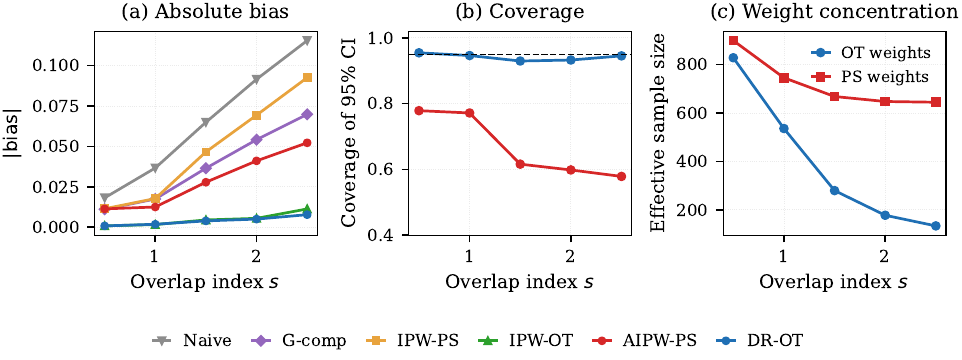}
\caption{Degradation with weakening overlap ($n=2000$, $\eta=0.5$, $J=10$;
averaged over scenarios). (a) Absolute bias. (b) Coverage of nominal $95\%$
intervals. (c) Effective sample size of the two weightings. Larger
$s_{\mathrm{ov}}$ means the target law sits further from the trial law.}
\label{fig:sim-overlap}
\end{figure}
 
\paragraph{Summary.}
The simulations support the main theoretical predictions, with two
finite-sample qualifications. Calibration is free when the sampling score is
correct and decisive when it is not; the plug-in variance attains the
efficiency bound for an adequate calibration space; growing that space restores
valid inference under double misspecification at a measurable efficiency cost;
and source relaxation delivers graceful degradation under weak overlap. The
qualifications are that \(J\) has two edges rather than one --- too coarse
leaves a variance bias that more data cannot remove, too rich renders the
calibration system unsolvable, with the residual signalling the latter --- and
that effective sample size is not a safe standalone diagnostic, since a
misspecified score can appear stable precisely because it under-corrects.
 
These simulations clarify when Riesz calibration improves transport and where
finite-sample instability can arise. We next apply the method to the motivating
real-data setting.

\section{Real Data Analysis}
\label{sec:realdata}
We applied the proposed method to all-cause mortality at 30 months in
transthyretin amyloid cardiomyopathy. The target population consisted of
patients from a real-world registry who were enrolled in 2019 or later and
received the study treatment. Because the real-world data contained treated
patients only, the corresponding placebo outcomes were transported from a
randomized controlled trial comparing the study treatment with placebo.
Detailed descriptions of the randomized trial and real-world registry,
including their study populations, eligibility criteria, and follow-up, have
been reported previously \citep{GarciaPavia2025, Maurer2018}. Baseline covariates from the entire randomized trial were used in the optimal
transport step to align the trial population with the real-world target
population, while placebo-arm outcomes were used to construct the transported
counterfactual control outcomes. Transport was based on age, sex, race/ethnicity, body
mass index, New York Heart Association class, TTR genotype, and exposure time.
The target estimand was the 30-month mortality risk difference in the treated
real-world population, defined as the difference in the proportions of
patients experiencing all-cause mortality under the study treatment and under
the corresponding transported placebo condition.

We considered six estimators, the unadjusted difference in observed outcomes
(\textsc{Naive}), outcome-regression standardization (\textsc{G-comp}),
inverse-probability weighting based on a parametric sampling propensity model
(\textsc{IPW-PS}), inverse weighting based on the calibrated OT weights
(\textsc{IPW-OT}), the conventional augmented inverse-probability weighted
estimator (\textsc{AIPW-PS}), and the proposed doubly robust
Riesz-calibrated OT estimator (\textsc{RICOT}).
Figure~\ref{fig:application-estimates} summarizes the estimated 30-month
mortality risk differences and corresponding 95\% confidence intervals for
all six methods. The unadjusted estimate was $-0.334$, while all adjusted
and transported estimators produced larger absolute risk differences. This
pattern is consistent with the difference in baseline risk between the two
populations. As shown in Figure~\ref{fig:application-overlap}, the treated
real-world population is shifted toward higher predicted placebo risk than
the randomized trial population.

Among the adjusted estimators, the \textsc{G-comp}, \textsc{AIPW-PS}, and
\textsc{RICOT} estimates were similar ($-0.534$, $-0.540$, and $-0.556$,
respectively), suggesting that the outcome regression captures much of the
relevant variation in placebo risk in this application. The \textsc{IPW-PS}
estimate was somewhat attenuated ($-0.476$), whereas \textsc{IPW-OT} produced
an estimate of $-0.578$, closer to the doubly robust estimates. This pattern
is consistent with the calibrated OT weights providing a better approximation
to the trial-to-target density ratio than the parametric sampling propensity
model in this application. The \textsc{IPW-OT} confidence interval was wider,
reflecting the greater variability of weighting alone, whereas augmentation
improved precision. The proposed \textsc{RICOT} estimator gave a 30-month
mortality risk difference of $-0.556$, with a 95\% confidence interval of
$(-0.617,-0.494)$. The sieve dimension was set to $J=11$. A sensitivity analysis across a range of values of 
$J$ is reported in Figure~\ref{fig:sieve-stable-real} in the appendix, where the estimates are seen to be stable.

\begin{figure}[ht]
\centering
\includegraphics[width=0.8\textwidth]{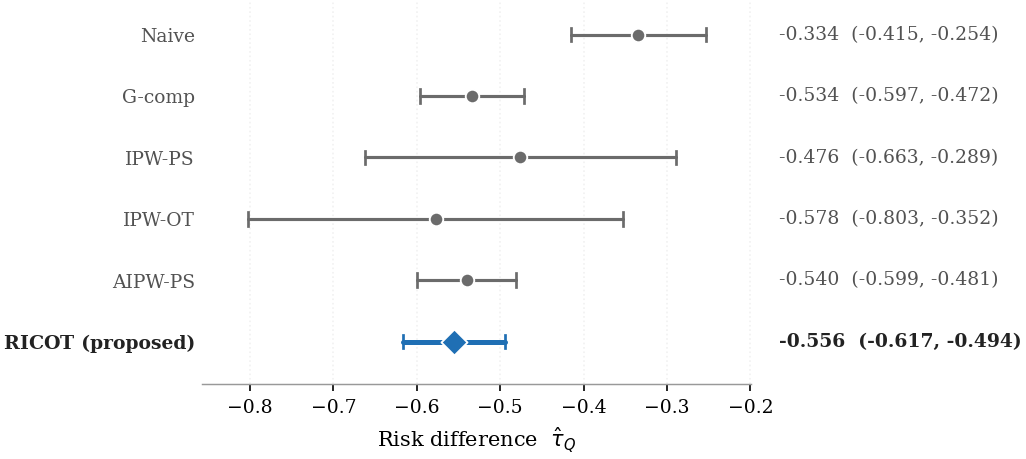}
\caption{Estimated 30-month mortality risk differences in the treated
real-world target population across the six estimators. Points denote point
estimates and horizontal lines denote 95\% confidence intervals. Negative
values favor the study treatment. 
}
\label{fig:application-estimates}
\end{figure}

\begin{figure}[H]
\centering
\includegraphics[width=0.8\textwidth]{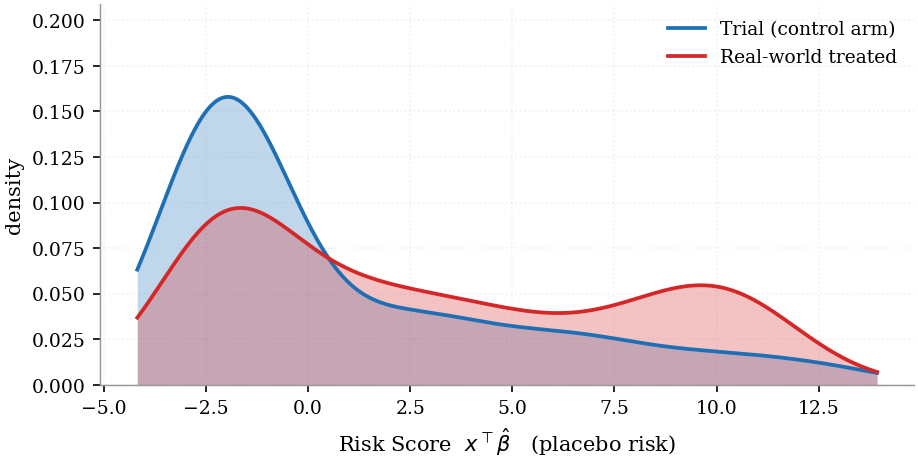}
\caption{Distribution of predicted placebo risk in the randomized trial and
the treated real-world target population. The rightward shift in the
real-world distribution indicates a greater prevalence of patients with
higher predicted placebo risk in the target population.}
\label{fig:application-overlap}
\end{figure}

\section{Extension to Rate-Scale Estimands}
\label{sec:rate}

Many clinical and real-world studies involve recurrent events, such as hospitalizations, observed over varying amounts of follow-up. In this setting, the treatment effect may be expressed through an incidence rate rather than a
fixed-time \corr{mean difference}. The population incidence rate is the expected number of events divided by the expected person-time at risk. It is therefore a ratio of population means and, in general, is not equivalent to
the average of individual event-to-time ratios. 
Rate-scale transport requires additional care because both the event count and the person-time may differ between the trial and target populations. For the treated target population, the recurrent-event count and person-time are
observed directly. Under the counterfactual control condition, however, they must be learned from the randomized trial and transported to the target population. When observed follow-up can depend on treatment, for example through death, treatment discontinuation, or other post-treatment events, person-time should not be included as a baseline transport variable. Instead, we \corr{transport the control event count and the control person-time with a single Riesz-calibrated OT weight built on the baseline covariates alone}, and combine the resulting transported population means to obtain the target control incidence rate. If person-time is predetermined and unaffected by treatment, the simpler construction that
incorporates it with the baseline covariates is sufficient.
Compared with the fixed-time \corr{mean difference}, the rate-scale analysis therefore introduces two additional statistical components. First, the numerator and denominator of the counterfactual control rate must both be estimated when person-time is treatment dependent. Second, uncertainty from both components must be propagated to the final incidence-rate contrast. \corr{The same augmented transport construction is used for each component, and double robustness operates through the shared weight.} Incidence-rate differences and incidence-rate ratios are then obtained from the transported rates, with inference based on the joint influence function and the delta method. For rate ratios, inference on the log scale provides a convenient way to construct positive confidence intervals.
Two practical points are important. The population incidence rate should not be replaced by the mean of individual event rates, since the two generally target different quantities. In addition, standardization of the counterfactual quantities must remain with respect to the target population;
moving the standardization to the weighted trial population eliminates the outcome-regression robustness branch. Formal identification conditions, double-robustness results, influence functions, and the treatment-dependent
person-time extension are provided in Appendix~\ref{app:G}. 

\section{Discussion}
\label{sec:discussion}
We developed a semiparametric framework for optimal transport weighting in trial-to-target generalization. A useful connection underlying the method is that the source marginal induced by source-relaxed entropic transport acts as a density-ratio estimator, while the same density ratio is the Riesz representer entering the efficient influence function for the transported treatment effect. Fixed regularization introduces a persistent bias in the density-ratio estimator and, in general, leads to bias in the resulting transported treatment-effect estimator. To correct this regularization-induced bias, we introduced Riesz calibration directly into the transport problem so that the resulting weights retain the flexibility of OT while satisfying the balance conditions required for semiparametric inference. With a sufficiently rich calibration space, the calibrated weights consistently estimate the target-to-trial density ratio even when the entropic and source-relaxation parameters remain fixed and positive.

We view the proposed method as a robust alternative to parametric
sampling-propensity weighting rather than as uniformly preferable. When the sampling propensity model is correctly specified, conventional AIPW can make efficient use of that parametric structure. In practice, however, the mechanism governing trial participation is unknown, and it can be difficult to know in advance which functional form adequately captures the differences between the trial and target populations. This is particularly relevant in regulatory settings, where key elements of the analysis are generally prespecified. A parametric sampling model must therefore be chosen without knowing whether it correctly represents the underlying population shift. Our approach offers an
alternative by targeting the density ratio through prespecified balance
conditions rather than relying on correct specification of a particular
parametric sampling model. The tradeoff is the slower rate and potentially greater finite-sample variability associated with nonparametric estimation.

The simulations reflected this tradeoff. Calibration reduced bias when the sampling propensity model was misspecified, while augmentation with outcome regression generally improved precision. They also showed that no single weight diagnostic tells the whole story. In practice, we recommend considering covariate balance, transport-weight behavior, effective sample size, and the calibration residual together, along with sensitivity to the transport and calibration tuning parameters.

Several boundaries of the current framework are worth stating plainly. The theory assumes a common baseline covariate space across trial and target and focuses primarily on a fixed-time outcome, and identification rests throughout on conditional transportability, an assumption that cannot be checked from the observed data. A formal sensitivity analysis for departures from this assumption is therefore the most important robustness gap, and we regard it as the natural companion to the estimation theory developed here. Two extensions stand out as substantive next steps. The rate-scale construction handles recurrent-event summaries but not general survival functionals, and extending the Riesz-calibration argument to time-to-event outcomes is the immediate methodological direction. Designs that use real-world controls to augment a randomized control arm are a second, since allowing real-world data into the control arm rather than treating it as the fixed target changes the estimand and the transport structure and raises new identification questions. Further directions include nonidentical covariate sets across trial and real-world data, high-dimensional representer estimation, and scalable transport solvers.

\bibliographystyle{plainnat}
\bibliography{refs}

\section*{Acknowledgments}
We thank Pfizer Inc. for its support. We also thank Kenneth K. Kwok, Rong Wang, and Martin Ove Carlsson for their assistance with data acquisition.

\section*{Declaration of interests}
Margaret Gamalo and Prosenjit Kundu are employees and stockholders of Pfizer Inc. The authors declare no other competing interests.

\newpage 
\renewcommand{\thesection}{S\arabic{section}}
\setcounter{section}{0}
\setcounter{equation}{0}
\numberwithin{equation}{section}
\renewcommand{\theequation}{\thesection.\arabic{equation}}
\setcounter{assumption}{0}
\renewcommand{\theassumption}{S\arabic{assumption}}
\setcounter{theorem}{0}
\renewcommand{\thetheorem}{ST\arabic{theorem}}
\setcounter{definition}{0}
\renewcommand{\thedefinition}{S\arabic{definition}}
\setcounter{remark}{0}
\renewcommand{\theremark}{SR\arabic{remark}}
\setcounter{example}{0}
\renewcommand{\theexample}{S\arabic{example}}
\makeatletter
\renewcommand{\theHsection}{S\arabic{section}}

\def\theHequation{S\arabic{section}.\arabic{equation}}
\def\theHtheorem{ST.\arabic{theorem}}
\def\theHassumption{S.\arabic{assumption}}
\def\theHdefinition{S.\arabic{definition}}
\def\theHremark{SR.\arabic{remark}}
\def\theHexample{S.\arabic{example}}
\makeatother

\begin{center}
{\large\bf SUPPLEMENT TO ``RIESZ-CALIBRATED OPTIMAL TRANSPORT FOR
TRANSPORTING TRIAL EVIDENCE TO REAL-WORLD POPULATIONS''}
\end{center}
 
\bigskip

This supplementary material contains the proofs of all results stated in the
main paper, together with the auxiliary results alluded to there and some
further discussion. The development can be summarized as follows: identification (Proposition~\ref{prop:C-id}) leads to the
efficient influence function (Theorem~\ref{thm:C-eif}), which exhibits the
density ratio as a Riesz representer and yields an exact product remainder
(Proposition~\ref{prop:C-rem}); the uncalibrated transport weight violates the
Riesz equations at fixed regularisation (Corollary~\ref{cor:D-imbalance});
calibration repairs this and, through an exponential-tilting representation
(Proposition~\ref{prop:E-Mest}), delivers an $L^2(P)$ rate
(Theorem~\ref{thm:E-rate}); and the two combine into an efficient central limit theorem (Theorem~\ref{thm:F-clt}).
 
\section{Notations, Assumptions, and Auxiliary Results}
\label{app:A}
 
\subsection{Notation and conventions}
\label{sub:A-notation}
 
\emph{Data.} Two independent samples are observed. The \emph{trial sample}
consists of $O^{\mathrm R}_i=(X_i,T_i,Y_i)$, $i=1,\dots,n$, drawn independently
from a law $P$ on $\calX\times\{0,1\}\times\R$, where $\calX\subset\R^d$,
$T_i\in\{0,1\}$ is the treatment indicator, $Y_i(a)$ denotes the potential
outcome under treatment $a\in\{0,1\}$, and $Y_i=Y_i(T_i)$. The \emph{target
sample} consists of $O^{\mathrm Q}_j=(X^Q_j,Y^{\mathrm Q}_j)$, $j=1,\dots,m$,
drawn independently from a law $Q$ on $\calX\times\R$; every target unit is
treated, so $Y_j^{\mathrm Q}=Y_j^{\mathrm Q}(1)$.
 
\emph{Expectations.} We write $\EP$ and $\EQ$ for expectations under $P$ and
$Q$, reserving the generic symbols $(X,T,Y)$ for a trial observation and
$(Z,Y^{\mathrm Q})$ for a target observation. Thus $\EP[f(X)]$ integrates $f$
against the covariate law of the trial and $\EQ[f(Z)]$ against that of the
target; no separate symbol for the covariate marginals is needed, and none is
used except in the definition of the density ratio below and in the transport
program, where measures are genuinely the object of study. The empirical
counterparts are
\begin{equation}
  \En[f(O^{\mathrm R})]\coloneq \frac1n\sum_{i=1}^n f(O_i^{\mathrm R}),
  \qquad
  \Em[f(O^{\mathrm Q})]\coloneq \frac1m\sum_{j=1}^m f(O_j^{\mathrm Q}),
  \label{eq:empirical}
\end{equation}
abbreviated $\En[f(X)]$ and $\Em[f(Z)]$ when the integrand depends only on
the covariates. For the associated centred empirical processes we write, for
instance, $(\En-\EP)\{f(O^{\mathrm R})\}
 =n^{-1}\sum_i\bigl\{f(O^{\mathrm R}_i)-\EP[f(O^{\mathrm R})]\bigr\}$.
 
\emph{Norms.} For $1\le t<\infty$ we write
$\norm{f}{L^t(P)}\coloneq \bigl(\EP[|f(X)|^t]\bigr)^{1/t}$ and
$\norm{f}{L^t(Q)}\coloneq \bigl(\EQ[|f(Z)|^t]\bigr)^{1/t}$, with $\norm{f}{\infty}$ the
supremum over $\calX$; when the argument is the full observation rather than
the covariate this is indicated explicitly. Euclidean and maximum norms on
$\R^J$ are $\norm{\cdot}{2}$ and $\norm{\cdot}{\infty}$; $\lambda_{\min}$ and
$\lambda_{\max}$ denote extreme eigenvalues of a symmetric matrix. We write
$W^{k,t}$ for the order-$k$ Sobolev space on $\calX$, and use $O_p$, $o_p$ with
respect to the joint law of the two samples, with $\xrightarrow{d}$ for weak
convergence.
 
\emph{Sample sizes.} Set $N\coloneq n\wedge m$. Throughout, the two sample sizes grow
proportionally,
\begin{equation}
  n/m\to\eta\in(0,\infty),
  \label{eq:ratio}
\end{equation}
so that $n$, $m$ and $N$ are of the same order and the distinction between
$o_p(n^{-1/2})$, $o_p(m^{-1/2})$ and $o_p(N^{-1/2})$ is immaterial; we use
whichever is most transparent at the point of use. Tuning parameters that vary
with sample size carry the index $N$, for example $J_N$ and $K_N$.
 
\emph{Nuisance quantities.} Three population quantities recur:
\begin{equation}
  e_0(x)\coloneq \P(T=0\mid X=x),
  \qquad
  A\coloneq \frac{1-T}{e_0(X)},
  \qquad
  \mu_0(x)\coloneq \EP(Y\mid X=x,\,T=0).
  \label{eq:nuisances}
\end{equation}
In a randomised trial $e_0$ is known by design; this is used repeatedly and is
what distinguishes the trial sample from an observational one. We write $p$
and $q$ for the densities of the covariate distributions under $P$ and $Q$ with
respect to a common dominating measure, and define the \emph{density ratio}
\begin{equation}
  \rdag(x)\coloneq \frac{q(x)}{p(x)},
  \qquad x\in\calX,
  \label{eq:rdag}
\end{equation}
with the convention that $\rdag$ is defined wherever $p(x)>0$.
Lemma~\ref{lem:A-riesz} restates the defining property of $\rdag$ purely in
terms of expectations, and it is that form which is used everywhere below.
 
\emph{Transport.} The cost is $c(x,z)=\norm{x-z}{2}^2$, the entropic
regularisation parameter is $\varepsilon>0$ and the marginal relaxation
parameter is $\rho>0$; neither is related to the estimand $\tau_Q$. We
abbreviate
\begin{equation}
  \gamma\coloneq \frac{\varepsilon}{\varepsilon+\rho}\in(0,1),
  \label{eq:gamma}
\end{equation}
which appears throughout as a shrinkage exponent. A coupling is represented by
its density $\pi(x,z)\ge0$ on $\calX\times\calX$, with marginals
\begin{equation}
  \pi_1(x)\coloneq \int\pi(x,z)\,dz,
  \qquad
  \pi_2(z)\coloneq \int\pi(x,z)\,dx;
  \label{eq:marginals}
\end{equation}
the letter $\pi$ is reserved for couplings, the Gaussian normalising constant
being absorbed into the density $\phi_\varepsilon$ of \eqref{eq:gibbs} below.
For nonnegative $f$ and $g$ we use the \emph{generalised} Kullback--Leibler
divergence
\begin{equation}
  \KL(f\,\|\,g)\coloneq \int\Bigl\{f\log\frac fg-f+g\Bigr\},
  \label{eq:genKL}
\end{equation}
with value $+\infty$ if $f>0$ somewhere that $g=0$, and the convention
$0\log0=0$; its first variation in $f$ is $\log(f/g)$.
 
\emph{Calibration.} The calibration basis is
$b_J(x)=(b_{J,1}(x),\dots,b_{J,J}(x))^\top$ with $b_{J,1}\equiv1$, and
$H_J\coloneq \mathrm{span}\{b_{J,1},\dots,b_{J,J}\}$. Scores of parametric submodels
are written $\zeta_P$ and $\zeta_Q$, reserving $s$ for the smoothness index;
regression residuals are denoted $\xi$, reserving $\varepsilon$ for the
entropic parameter; the calibration multiplier is $\theta$, reserving $\lambda$
for the incidence rates of Appendix~\ref{app:G}.
 
\emph{Cross-fitting.} Let $\{\calI_1,\dots,\calI_K\}$ partition the trial
indices into $K$ evaluation folds with $K$ fixed, put
$\calI_{-k}\coloneq \{1,\dots,n\}\setminus\calI_k$ and $n_{-k}\coloneq |\calI_{-k}|$, and
write
$\Enk{k}[f(O^{\mathrm R})]
 \coloneq n_{-k}^{-1}\sum_{i\in\calI_{-k}}f(O_i^{\mathrm R})$.
A superscript $(-k)$ on a fitted nuisance means that it was trained without the
observations in fold $k$; for $i\in\calI_k$ the shorthand $\muhat(X_i)$ means
$\muhat^{(-k)}(X_i)$, so that $\En$ applied to an expression involving fitted
nuisances denotes the assembled foldwise out-of-sample average.
 
\subsection{Standing assumptions}
\label{sub:A-assumptions}
 
The whole development rests on three assumptions, stated here once. The first concerns the data-generating mechanism and is what makes the estimand a functional of the observed laws; the second concerns the transport program and the calibration design; the third concerns the accuracy of the outcome regression. These are the global assumptions used for identification and for
the general construction, and they are in force throughout. Later appendices introduce two more hypotheses: two-sided overlap (Assumption~\ref{ass:E-twosided}) in Appendix~\ref{app:E}, and rate-scale transfer (Assumption~\ref{ass:G}) in Appendix~\ref{app:G}.
 
\begin{assumption}[Sampling and identification]
\label{ass:id}
The trial and target samples are mutually independent and \eqref{eq:ratio}
holds. Moreover:
\begin{enumerate}[label=\textup{(\roman*)},ref=\theassumption\textup{(\roman*)},
                  leftmargin=2.6em,itemsep=2pt,topsep=3pt]
\item\label{ass:id-cons} \textbf{Consistency.} $Y=TY(1)+(1-T)Y(0)$ almost
      surely in the trial, and $Y^{\mathrm Q}=Y^{\mathrm Q}(1)$ almost surely
      in the target sample.
\item\label{ass:id-pos} \textbf{Positivity.} There are constants
      $0<\underline e\le\overline e<1$ with
      $\underline e\le e_0(x)\le\overline e$ for $P$-almost every $x$.
\item\label{ass:id-trans} \textbf{Transportability of the control response.}
      $\EQ[Y^{\mathrm Q}(0)\mid Z=x]=\mu_0(x)$ for $Q$-almost every $x$.
\item\label{ass:id-ovl} \textbf{Overlap.} The covariate law under $Q$ is
      absolutely continuous with respect to that under $P$, and the density
      ratio \eqref{eq:rdag} is bounded,
      $\norm{\rdag}{\infty}\le\overline r<\infty$.
\item\label{ass:id-mom} \textbf{Moments.} There is $\delta>0$ with
      $\EQ|Y^{\mathrm Q}-\mu_0(Z)|^{2+\delta}<\infty$ and
      $\EP\bigl|A\,\rdag(X)\{Y-\mu_0(X)\}\bigr|^{2+\delta}<\infty$, and the
      conditional residual variance is bounded,
      $\EP[\{Y-\mu_0(X)\}^2\mid X,\,T=0]\le C_Y<\infty$ almost surely\corr{;
      moreover $\mu_0$ is bounded, $\norm{\mu_0}{\infty}\le C_0<\infty$
      (automatic for bounded outcomes)}.
\end{enumerate}
\end{assumption}
 
\begin{remark} Each condition plays a distinct role, and it is worth separating them before
they are used. Conditions \ref{ass:id-cons}--\ref{ass:id-pos} are internal to
the trial. Together they identify the control response surface $\mu_0$ on the
trial's covariate support from randomised data, with positivity guaranteeing
that the conditioning event $\{T=0,\,X=x\}$ has positive probability so that
$\mu_0(x)$ is well defined. Condition \ref{ass:id-trans} is the
cross-population assumption and is what allows exporting $\mu_0$ from the
trial to the target. It restricts only the control arm, because the treated
response in the target population is observed directly. Condition
\ref{ass:id-ovl} is a support condition ensuring that $\mu_0$, identified only where the trial has covariate support, is defined $Q$-almost everywhere. The
boundedness of $\rdag$ is slightly more than absolute continuity
requires, and is used only to transfer $L^2(P)$ rates to $L^2(Q)$ in
Appendix~\ref{app:F}. Such transfer is potentially feasible under \corr{much weaker conditions}, which we do not pursue. Condition \ref{ass:id-mom} supplies the Lindeberg and
law-of-large-numbers inputs\corr{; the boundedness of $\mu_0$ is
used only for the empirical-process envelopes of Appendix~\ref{app:F}, and in
particular gives $\mu_0\in L^2(P)$, which several proofs use tacitly.
Randomisation of treatment within the trial is a property of the design and
is deliberately not listed as an assumption. For instance, $\mu_0$ is defined directly as
the observed-data regression in \eqref{eq:nuisances}, every proof below uses
only that definition, and randomisation enters solely through the causal
reading $\mu_0(x)=\EP\{Y(0)\mid X=x\}$, which is what makes
\ref{ass:id-trans} a substantive cross-population restriction}.
\end{remark}
 
\begin{assumption}[Transport program and calibration design]
\label{ass:design}
The covariate space $\calX\subset\R^d$ is compact and convex with Lipschitz
boundary, and $c(x,z)=\norm{x-z}{2}^2$. The regularisation parameters are
\emph{fixed}, with $\varepsilon\in[\varepsilon_0,\varepsilon_1]$ for some
$0<\varepsilon_0\le\varepsilon_1<\infty$ and
$\rho/\varepsilon\le C_\rho<\infty$. In addition:
\begin{enumerate}[label=\textup{(\roman*)},ref=\theassumption\textup{(\roman*)},
                  leftmargin=2.6em,itemsep=2pt,topsep=3pt]
\item\label{ass:des-basis} \textbf{Basis.} $b_{J,1}\equiv1$; the coordinates
      are uniformly bounded, $\sup_{J,j,x}|b_{J,j}(x)|\le B<\infty$; and the
      Gram matrices $G_J\coloneq \EP[b_J(X)b_J(X)^\top]$ satisfy
      $0<c_G\le\lambda_{\min}(G_J)\le\lambda_{\max}(G_J)\le C_G<\infty$ uniformly
      in $J$.
\item\label{ass:des-approx} \textbf{Sieve approximation of the tilt.} With $L$
      the population offset defined in \eqref{eq:offset} below, there are
      \corr{$s>d/2$ (an integer, for notational simplicity)}, $C_s<\infty$ and vectors $\theta_J\in\R^J$ such that
      \[
        \bigl\|\log\rdag-\gamma L-\theta_J^\top b_J\bigr\|_{\infty}
        \;\le\; C_s\,J^{-s/d}
        \qquad\text{for every }J .
      \]
\item\label{ass:des-dual} \textbf{Quantitative interiority and dual
      stability.} With probability tending to one, the empirical calibrated
      program of Definition~\ref{def:emp} admits a solution whose plan and
      source ratios are bounded away from zero and infinity, that is,
      \[
        0<c_h\le\frac{\widehat\pi_{ij}}{a_i\corr{\omega_j}}\le C_h<\infty,
        \qquad
        0<c_r\le\widehat r(X_i)\le C_r<\infty
      \]
      for all \corr{trial units $i=1,\ldots,n$} and all $j\le m$.
      \corr{The population target dual potential $g$ of \eqref{eq:offset} is
      bounded, $\norm{g}{\infty}\le C_g<\infty$, and normalised so that
      $\EQ[g(X^Q)]=0$, matching the normalisation $\Em(v)=0$ of the empirical
      dual.} Moreover the empirical
      offset $\Lhat$ of \eqref{eq:offset-hat} satisfies
      $\norm{\Lhat-L}{\infty}=O_p(a_{L,N})$ for a deterministic sequence
      $a_{L,N}\downarrow0$.
\end{enumerate}
\end{assumption}
 
\begin{remark} Condition \ref{ass:des-basis} is the standard sieve design condition, satisfied
for instance by a tensor-product trigonometric basis on $\calX=[0,1]^d$ with
periodic boundary conditions when the trial covariate density is bounded away
from zero and infinity. Condition \ref{ass:des-approx} is the approximation
condition, imposed on the function $\log\rdag-\gamma L$ rather than on $\rdag$
itself; Remark~\ref{rem:why-tilt} explains why this is the right object, and
why an approximation condition on $\rdag$ together with a uniform smoothness
bound on the calibration span would be mutually inconsistent. \corr{The
smoothness floor $s>d/2$ is what allows the $\ell_2$ control of the tilt
coefficient to be converted into the supremum-norm control that the
exponential parametrisation requires (Step 4 of the proof of
Theorem~\ref{thm:E-rate}), and it is also the condition under which the
empirical-process arguments of Appendix~\ref{app:F} close.} Condition
\ref{ass:des-dual} is a stability requirement on the transport solve. Strict interiority alone gives positivity of the
fitted ratios but not uniform bounds, and it is the uniform bounds that
Lemma~\ref{lem:dual-bdd} converts into $\norm{v}{\infty}=O_p(1)$ and thence
into smoothness of the offset. This is the only place where the numerical
behaviour of the Sinkhorn iteration enters the statistical theory.
\end{remark}

\begin{assumption}[Cross-fitting and outcome regression]
\label{ass:nuisance}
The outcome regression $\muhat$ is cross-fitted over $K$ fixed folds as
described in Section~\ref{sub:A-notation}, so that, conditionally on the
training data, the fitted value $\muhat(X_i)$ is independent of
$Y_i-\mu_0(X_i)$ given $(X_i,T_i)$; if target outcomes are used in training,
the target sample is split so that each target observation entering $\Em$ is
out of sample for the regression evaluated at it. \corr{The fitted regression
takes values in a fixed bounded interval,
$\norm{\muhat}{\infty}\le C_\mu<\infty$ almost surely; this is automatic for
bounded outcomes and can always be arranged by truncating at a fixed interval containing the
range of $\mu_0$ --- a pointwise contraction --- without affecting the rate
below.} There is $\beta>0$ with
\[
  \norm{\muhat-\mu_0}{L^2(P)}
  =\bigl(\EP\bigl[\{\muhat(X)-\mu_0(X)\}^2\bigr]\bigr)^{1/2}
  =O_p(N^{-\beta}).
\]
\end{assumption}
 
\begin{remark} Assumption~\ref{ass:nuisance} controls only the $L^2(P)$ error; the
corresponding $L^2(Q)$ bound is not assumed but follows, since by
Lemma~\ref{lem:A-riesz} and \ref{ass:id-ovl},
\begin{equation}
  \EQ[f(Z)^2]=\EP[\rdag(X)f(X)^2]\le\overline r\;\EP[f(X)^2]
  \label{eq:PtoQ}
\end{equation}
for every $f\in L^2(P)$. The calibrated transport weight is a function of
covariates only, and is therefore \emph{not} required to be cross-fitted; the
consequences of this asymmetry are worked out in Appendix~\ref{app:F}.
\end{remark}

\subsection{Auxiliary results}
\label{sub:A-aux}
 
We collect here four elementary results used repeatedly below. Their proofs are
in Appendix~\ref{app:B}. The first is the device by which \corr{the known randomisation probability}
is converted into a statement about the control arm. It appears in every
influence-function and remainder computation later, and it is the reason
that the known randomisation probability $e_0$ is so valuable.
 
\begin{lemma}[Inverse-randomisation identity]
\label{lem:A-invrand}
Let $G$ be a measurable function of $O^{\mathrm R}=(X,T,Y)$ with
$\EP|A\,G(O^{\mathrm R})|<\infty$. Under
Assumptions~\ref{ass:id-cons}--\ref{ass:id-pos},
\begin{equation}
  \EP[A\,G(O^{\mathrm R})\mid X=x]
  \;=\;
  \EP[G(O^{\mathrm R})\mid X=x,\,T=0]
  \qquad\text{for $P$-almost every }x,
  \label{eq:invrand}
\end{equation}
and consequently
$\EP[A\,G(O^{\mathrm R})]
 =\EP\bigl(\EP[G(O^{\mathrm R})\mid X,\,T=0]\bigr)$.
In particular, for any \corr{$h\in L^2(P)$} and any square-integrable $\mu$,
\begin{equation}
  \EP\bigl[A\,h(X)\{Y-\mu(X)\}\bigr]
  \;=\;
  \EP\bigl[h(X)\{\mu_0(X)-\mu(X)\}\bigr],
  \label{eq:invrand2}
\end{equation}
\corr{both sides being finite under Assumption~\ref{ass:id-mom}.}
\end{lemma}
 
The second result records the sense in which the density ratio is a Riesz
representer, and does so entirely in terms of expectations. It is stated
separately because the uniqueness half is what forces the limit of any regular
procedure to be $\rdag$, a point taken up in Appendix~\ref{app:F}.
 
\begin{lemma}[Riesz representation]
\label{lem:A-riesz}
Under Assumption~\ref{ass:id-ovl}, $\rdag\in L^2(P)$ and
\begin{equation}
  \EQ[h(X)]=\EP[\rdag(X)\,h(X)]
  \label{eq:riesz}
\end{equation}
\corr{for every nonnegative measurable $h$ (both sides possibly infinite), and for every $h\in L^2(P)$, both sides then being finite; neither
class of test functions contains the other, and the nonnegative case is the
primitive one, the $L^2(P)$ case following from it by splitting $h$ into
positive and negative parts.}
and $\rdag$ is the unique element of $L^2(P)$ with this property. Moreover, if
$r\in L^2(P)$ satisfies the finite-dimensional calibration equations
\begin{equation}
  \EP[r(X)\,b_J(X)]=\EQ[b_J(Z)],
  \label{eq:calib-pop}
\end{equation}
then $\EP[r(X)h(X)]=\EQ[h(Z)]$ for every $h\in H_J$; equivalently,
$\EP\bigl[\{r(X)-\rdag(X)\}h(X)\bigr]=0$ for every $h\in H_J$.
\end{lemma}
 
The third result concerns the Gibbs smoothing operator that will be seen in
Appendix~\ref{app:D} to underlie the entropic transport weight. Let $\phi_\varepsilon$
denote the $\mathcal N\{0,(\varepsilon/2)I_d\}$ Lebesgue density on $\R^d$ and
define
\begin{equation}
  S_\varepsilon f\coloneq f*\phi_\varepsilon,
  \qquad\text{equivalently}\qquad
  S_\varepsilon f(x)=\E[f(x+U)],\quad U\sim\mathcal N\{0,(\varepsilon/2)I_d\},
  \label{eq:gibbs}
\end{equation}
so that $S_\varepsilon$ is Gaussian smoothing at bandwidth
$\sqrt{\varepsilon/2}$ per coordinate. Writing this operator as an expectation
rather than through its normalising constant avoids a collision with the
coupling $\pi$ and makes the expansion below transparent.
 
\begin{lemma}[Heat expansion]
\label{lem:A-heat}
Let $\mathcal K\subset\R^d$ be compact and let $f$ be four times continuously
differentiable with derivatives through order four bounded uniformly on a
neighbourhood of $\mathcal K$, and bounded globally. Then, uniformly on
$\mathcal K$,
\begin{equation}
  S_\varepsilon f=f+\frac{\varepsilon}{4}\,\Delta f+O(\varepsilon^2),
  \label{eq:heat}
\end{equation}
where $\Delta\coloneq \sum_{k=1}^d\partial^2/\partial x_k^2$.
\end{lemma}
 
The last auxiliary result is what allows us to control the \emph{offset} in the
exponential-tilting representation of the calibrated weight. Its point is that
a \corr{log-mixture-of-exponentials} of uniformly smooth functions is uniformly
smooth, with a bound \corr{that does not depend on the mixing distribution ---
so that the empirical offset built from $m$ target points, and the population
offset built from the target law itself, inherit the smoothness of the cost by
one and the same argument.}
 
\begin{lemma}[Stability of log-sum-exp]
\label{lem:A-lse}
{Let $k\in\N$, let $(\mathcal Z,\nu)$ be a probability space, and let
$\{f_z:z\in\mathcal Z\}$ be functions $f_z:\calX\to\R$, jointly measurable in
$(x,z)$, with $\sup_{z\in\mathcal Z}\norm{f_z}{W^{k,\infty}(\calX)}\le
F<\infty$. Then
\begin{equation}
  \Bigl\|\log\int_{\mathcal Z} e^{f_z}\,d\nu(z)\Bigr\|_{W^{k,\infty}(\calX)}
  \;\le\; C(k,d,F),
  \label{eq:lse}
\end{equation}
where the constant depends only on $k$, $d$ and $F$, and in particular not on
$\nu$. Taking $\nu$ discrete with weights $\omega_1,\dots,\omega_m$ recovers
the finite-sum bound
$\norm{\log\sum_{j\le m}\omega_je^{f_j}}{W^{k,\infty}(\calX)}\le C(k,d,F)$,
with a constant free of $m$ and of the weights.}
\end{lemma}
 
\subsection{Proofs of the Auxiliary Results}
\label{app:B}
 
\subsubsection{Proof of Lemma~\ref{lem:A-invrand}}
 
Fix $x$ with $e_0(x)>0$, which holds for $P$-almost every $x$ by
Assumption~\ref{ass:id-pos}, so that the conditioning events below have
positive probability. Conditioning on $X=x$ and expanding over the two values
of $T$,
\begin{align*}
  \EP[A\,G\mid X=x]
  &=\frac{1}{e_0(x)}\,\EP[(1-T)\,G\mid X=x]\\[2pt]
  &=\frac{1}{e_0(x)}\sum_{a\in\{0,1\}}\P(T=a\mid X=x)\,(1-a)\,
      \EP[G\mid X=x,\,T=a]\\[2pt]
  &=\frac{1}{e_0(x)}\cdot e_0(x)\cdot\EP[G\mid X=x,\,T=0],
\end{align*}
the term $a=1$ dropping out because of the factor $(1-a)$. This is
\eqref{eq:invrand}; taking expectations over $X$ gives the second assertion.
 
For \eqref{eq:invrand2}, take $G(O^{\mathrm R})=h(X)\{Y-\mu(X)\}$.
\corr{The integrability hypothesis holds. Applying the first display to the
nonnegative integrand $|h(X)\{Y-\mu(X)\}|$ which is legitimate by Tonelli's
theorem, and using $|Y-\mu|\le|Y-\mu_0|+|\mu_0-\mu|$ with
Assumption~\ref{ass:id-mom} and Cauchy--Schwarz,
$$
\begin{aligned}
  \EP\bigl|A\,h(X)\{Y-\mu(X)\}\bigr|
  &=\EP\Bigl[|h(X)|\;\EP\bigl\{|Y-\mu(X)|\bigm|X,\,T=0\bigr\}\Bigr]\\
  &\le
  \sqrt{C_Y}\,\norm{h}{L^2(P)}
   +\norm{h}{L^2(P)}\norm{\mu_0-\mu}{L^2(P)},
\end{aligned}
$$
is finite. Note that the conditional bound of Assumption~\ref{ass:id-mom} is a
control-arm bound, which is exactly what the conditioning produces; no moment
condition on the treated arm is used.} Since $h$
and $\mu$ are functions of $X$ alone they pass through the conditioning, so
\[
  \EP[G\mid X=x,\,T=0]
  =h(x)\bigl[\EP(Y\mid X=x,\,T=0)-\mu(x)\bigr]
  =h(x)\{\mu_0(x)-\mu(x)\},
\]
by the definition of $\mu_0$ in \eqref{eq:nuisances}. Taking expectations over
$X$ completes the proof. \hfill$\square$
 
\subsubsection{Proof of Lemma~\ref{lem:A-riesz}}
 
We first establish the representation and then deduce square-integrability from
it.
 
{\emph{Representation.} For nonnegative measurable $h$, definition
\eqref{eq:rdag} and the fact that $q=0$ wherever $p=0$ give
$\EQ[h(Z)]=\int hq=\int h\,\rdag\,p=\EP[\rdag(X)h(X)]$, both sides possibly
infinite, by Tonelli's theorem. Taking $h\equiv1$ yields $\EP[\rdag(X)]=1$
and hence, by Assumption~\ref{ass:id-ovl},
$\EP[\rdag(X)^2]\le\overline r\,\EP[\rdag(X)]=\overline r<\infty$, so
$\rdag\in L^2(P)$. For $h\in L^2(P)$, the nonnegative case applied to $|h|$
gives $\EQ|h(Z)|=\EP[\rdag(X)|h(X)|]
\le\norm{\rdag}{L^2(P)}\norm{h}{L^2(P)}<\infty$, so both sides of
\eqref{eq:riesz} are finite, and the identity follows by splitting $h$ into
its positive and negative parts.}
 
\emph{Uniqueness.} Suppose $r'\in L^2(P)$ also satisfies
$\EQ[h(Z)]=\EP[r'(X)h(X)]$ for all $h\in L^2(P)$. Subtracting
\eqref{eq:riesz} gives $\EP[\{r'(X)-\rdag(X)\}h(X)]=0$ for all such $h$. The
choice $h=r'-\rdag$, admissible because both functions lie in $L^2(P)$, yields
$\norm{r'-\rdag}{L^2(P)}^2=0$, so $r'=\rdag$ almost surely.
 
\emph{Sieve version.} If \eqref{eq:calib-pop} holds componentwise and
$h\in H_J$, write $h=a^\top b_J$ for some $a\in\R^J$. Then
\[
  \EP[r(X)h(X)]
  =a^\top\EP[r(X)b_J(X)]
  =a^\top\EQ[b_J(Z)]
  =\EQ[h(Z)].
\]
Combining with \eqref{eq:riesz} applied to the same $h$ gives
$\EP[\{r(X)-\rdag(X)\}h(X)]=0$. \hfill$\square$
 
\subsubsection{Proof of Lemma~\ref{lem:A-heat}}
 
Fix $x\in\mathcal K$ and write $D^kf(x)$ for the $k$-th Fréchet derivative of
$f$ at $x$, so that $D^2f(x)$ is the Hessian. By \eqref{eq:gibbs},
$S_\varepsilon f(x)=\E[f(x+U)]$ with
$U\sim\mathcal N\{0,(\varepsilon/2)I_d\}$. Taylor's theorem through third order
with a fourth-order remainder gives
\[
  f(x+U)=f(x)+Df(x)[U]+\tfrac12D^2f(x)[U,U]
         +\tfrac16D^3f(x)[U,U,U]+R_4(x,U),
\]
with $|R_4(x,U)|\le C\norm{U}{2}^4$,
valid whenever the segment from $x$ to $x+U$ remains in the neighbourhood on
which the derivative bounds hold. The complementary event has Gaussian
probability exponentially small in $1/\varepsilon$ and, $f$ being globally
bounded, contributes $o(\varepsilon^2)$\corr{; the Taylor-polynomial terms on
that event are handled by Cauchy--Schwarz, a Gaussian polynomial moment times
the square root of an exponentially small probability being
$O(e^{-c/\varepsilon})$}.
 
By symmetry of the centred Gaussian the first- and third-order terms have zero
mean. Since $\E(UU^\top)=(\varepsilon/2)I_d$,
\[
  \tfrac12\,\E[D^2f(x)[U,U]]
  =\tfrac12\operatorname{tr}\{D^2f(x)\,\E(UU^\top)\}
  =\frac{\varepsilon}{4}\,\Delta f(x).
\]
Finally $\E\norm{U}{2}^4=O(\varepsilon^2)$, so $\E|R_4(x,U)|=O(\varepsilon^2)$.
All derivative and tail bounds are uniform for $x\in\mathcal K$, which gives
\eqref{eq:heat} uniformly on $\mathcal K$. \hfill$\square$
 
\subsubsection{Proof of Lemma~\ref{lem:A-lse}}
 
{ Write $L(x)\coloneq \log\int_{\mathcal Z}e^{f_z(x)}\,d\nu(z)$ and, for each $x$, $\vartheta_z(x)\coloneq \frac{e^{f_z(x)}}{\int_{\mathcal Z}e^{f_w(x)}\,d\nu(w)},$ so that $\vartheta_z(x)\ge0$ and $\int_{\mathcal Z}\vartheta_z(x)\,d\nu(z)=1.$
Thus $\vartheta_\cdot(x)$ is a probability density with respect to $\nu$ for
every $x$. All interchanges of differentiation in $x$ and integration in $z$
below are justified by dominated convergence, every integrand being bounded,
uniformly in $(x,z)$, by a fixed polynomial in $F$ times $e^{2F}$.

\emph{Zeroth order.} Since $e^{-F}\le e^{f_z(x)}\le e^{F}$ for all $z$ and
$\nu$ is a probability measure, $-F\le L(x)\le F$, whence
$\norm{L}{\infty}\le F$.

\emph{First order.} Differentiating,
$DL(x)=\int_{\mathcal Z}\vartheta_z(x)\,Df_z(x)\,d\nu(z)$, an average of
vectors of norm at most \corr{$\sqrt d\,F$} under the probability density $\vartheta_\cdot(x)$,
so $\norm{DL}{\infty}\le \corr{\sqrt d\,F}$. Note also that
$D\vartheta_z=\vartheta_z\bigl(Df_z-\int\vartheta_w\,Df_w\,d\nu(w)\bigr)$, so
\corr{$|D\vartheta_z(x)|\le2\sqrt d\,F\,\vartheta_z(x)$ pointwise}.

\emph{Higher order.} We show by induction on $|\alpha|$ that
$\norm{D^\alpha L}{\infty}$ and
$\sup_z\norm{D^\alpha\vartheta_z/\vartheta_z}{\infty}$ are bounded by
constants depending only on $(k,d,F)$, for $1\le|\alpha|\le k$. The base case
was just verified. For the inductive step, differentiating
$DL=\int\vartheta_z\,Df_z\,d\nu(z)$ repeatedly produces, by the Leibniz rule,
finite sums of terms
$\int(D^{\alpha_1}\vartheta_z)(D^{\alpha_2}f_z)\,d\nu(z)$ with
$|\alpha_1|+|\alpha_2|\le k$, the number of such terms depending only on $k$
and $d$; writing
$D^{\alpha_1}\vartheta_z=\vartheta_z\,(D^{\alpha_1}\vartheta_z/\vartheta_z)$,
each such term is a $\vartheta$-weighted average of quantities bounded by the
inductive hypothesis and by $\sup_z\norm{f_z}{W^{k,\infty}}\le F$.
Differentiating
$D\vartheta_z=\vartheta_z(Df_z-\int\vartheta_w\,Df_w\,d\nu(w))$ produces terms
of the same shape by the usual multivariate Fa\`a di Bruno bookkeeping. The
essential point is that every integration over $z$ carries the weight
$\vartheta_z$ and is therefore an average rather than a sum. As a result, no factor
involving the size of the index set --- in the discrete case, no factor of $m$
or of the weights --- ever appears.

Combining the three displays gives $\norm{L}{W^{k,\infty}}\le C(k,d,F)$, which
is \eqref{eq:lse}; the finite-sum case is the special case in which $\nu$
places mass $\omega_j$ on the index $j$. \hfill$\square$}
 
\begin{remark}
\label{rem:B-lse}
Lemma~\ref{lem:A-lse} is applied in Appendix~\ref{app:E} \corr{twice, with
$k=s$: to the empirical offset, with $\nu$ the empirical distribution of the
target covariates and $f_j(\cdot)=\{v_j-c(\cdot,X^Q_j)\}/\varepsilon$, where
the hypothesis follows from \corr{the bound $\norm{v}{\infty}\le C_v$ of
Lemma~\ref{lem:dual-bdd}, valid on the event of
Assumption~\ref{ass:des-dual}}, the smoothness of the squared Euclidean cost,
and $\varepsilon\ge\varepsilon_0>0$; and to the population offset, with
$\nu=\QX$ and $f_z(\cdot)=\{g(z)-c(\cdot,z)\}/\varepsilon$, where it follows
from $\norm{g}{\infty}\le C_g$ (Assumption~\ref{ass:des-dual}).} It is worth noting that the above lemma is applied to the \emph{offset} only, and never
to the calibration tilt $\theta^\top b_J$ which was introduced later. The tilt is handled by the
$M$-estimation argument of Appendix~\ref{app:E}. Attempting to
bound the tilt in $W^{s,\infty}$ uniformly in $J$ would be inconsistent with
the Gram condition \ref{ass:des-basis}, as Remark~\ref{rem:why-tilt}
explains.
\end{remark}

\section{Identification, the Efficient Influence Function,
and the Exact Remainder}
\label{app:C}
 
This section does three things, in increasing order of consequence. It shows
that the transported estimand
\begin{equation}
  \tau_Q\coloneq \EQ\bigl[Y^{\mathrm Q}(1)-Y^{\mathrm Q}(0)\bigr]
  \label{eq:estimand}
\end{equation}
is a functional of the two observed data distributions; it computes the
efficient influence function of that functional, and finds the density ratio
$\rdag$ sitting inside it as a Riesz representer; and it establishes an exact
algebraic identity for the error of the associated estimating functional, an
identity which is the engine of every double-robustness statement in the paper
and which also fixes the terms on which transport-based and score-based
procedures are to be compared.
 
\subsection{Identification}
\label{sub:C-id}
 
\begin{proposition}[Identification of $\tau_Q$]
\label{prop:C-id}
Under Assumption~\ref{ass:id},
\begin{equation}
  \tau_Q
  \;=\;\EQ\bigl[Y^{\mathrm Q}\bigr]-\EQ\bigl[\mu_0(Z)\bigr]
  \;=\;\EQ\bigl[Y^{\mathrm Q}-\mu_0(Z)\bigr],
  \label{eq:id}
\end{equation}
a functional of the observed data distributions of the two samples.
\end{proposition}
 
\begin{proof}
The two arms of \eqref{eq:estimand} are identified by entirely different
routes, and it is worth keeping them apart.
 
\emph{Step 1:} Every unit in the target sample
receives treatment, so consistency (Assumption~\ref{ass:id-cons}) gives
$Y^{\mathrm Q}=Y^{\mathrm Q}(1)$ almost surely, and taking expectations under
$Q$ yields $\EQ[Y^{\mathrm Q}(1)]=\EQ[Y^{\mathrm Q}]$.
 
\emph{Step 2:} Fix $x$
with $e_0(x)>0$, which holds for $P$-almost every $x$ by Assumption~\ref{ass:id-pos},
so that the conditional expectations below are well defined. { The function $\mu_0(x)=\EP[Y\mid X=x,\,T=0]$ of \eqref{eq:nuisances}
depends only on the joint law of the observed $(X,T,Y)$, so it is identified
for $P$-almost every $x$; no further condition is involved. On the event
$\{T=0\}$, consistency gives $Y=Y(0)$ pointwise, so
$\mu_0(x)=\EP[Y(0)\mid X=x,\,T=0]$; because treatment is randomised in the
trial --- a property of the design rather than a listed assumption ---
conditioning on $T=0$ does not alter the conditional law of $Y(0)$ given
$X=x$, whence $\mu_0(x)=\EP[Y(0)\mid X=x]$. This last display is
interpretive: it gives Assumption~\ref{ass:id-trans} its causal content, but
the identification argument of Step 3 uses only the definition of $\mu_0$.}
 
\emph{Step 3:} The quantity
$\EQ[Y^{\mathrm Q}(0)]$ is the mean of an unobserved potential outcome.
Conditioning on the target covariate and using
Assumption~\ref{ass:id-trans},
\[
  \EQ\bigl[Y^{\mathrm Q}(0)\bigr]
  =\EQ\Bigl[\EQ\bigl[Y^{\mathrm Q}(0)\mid X^Q\bigr]\Bigr]
  =\EQ\bigl[\mu_0(X^Q)\bigr].
\]
It remains to check that this is well defined. By Step 2, $\mu_0$ is
identified off a set $\mathcal N_0$ of $P$-probability zero; by overlap
(Assumption~\ref{ass:id-ovl}) the covariate law under $Q$ is absolutely continuous with
respect to that under $P$, so $\mathcal N_0$ also has $Q$-probability zero, and
$\mu_0$ is therefore defined and identified $Q$-almost everywhere.
 
Subtracting Step 3 from Step 1 gives \eqref{eq:id}\corr{; both terms of the
split form are finite under Assumption~\ref{ass:id-mom}, whose boundedness of
$\mu_0$ gives $\EQ|\mu_0(X^Q)|<\infty$ and, with the $(2+\delta)$ moment,
$\EQ|Y^{\mathrm Q}|<\infty$}.
\end{proof}
 
\subsection{The efficient influence function}
\label{sub:C-eif}
 
We use the standard framework of regular parametric submodels. A regular
one-dimensional submodel through $(P,Q)$ is a pair of differentiable paths
$t\mapsto P_t$ and $t\mapsto Q_t$ with $P_0=P$, $Q_0=Q$, and densities $p_t$,
$q_t$ with respect to fixed dominating measures. The associated scores are
\[
  \zeta_P(o^{\mathrm R})\coloneq \frac{\partial}{\partial t}\log p_t(o^{\mathrm R})
  \Big|_{t=0},
  \qquad
  \zeta_Q(o^{\mathrm Q})\coloneq \frac{\partial}{\partial t}\log q_t(o^{\mathrm Q})
  \Big|_{t=0},
\]
which are square-integrable and centred, $\EP[\zeta_P]=0$ and
$\EQ[\zeta_Q]=0$, by differentiating the normalisation of the submodel
densities. \corr{Throughout, submodels are taken to be differentiable in
quadratic mean with locally dominated scores, so that differentiation under
integral signs below is justified; this standard regularity convention is not
repeated at each use.} Because the two samples are independent, the two scores may be
varied independently, and the pathwise derivative of $\tau_Q$ splits
accordingly into a target-law and a trial-law contribution. The functional
$\tau_Q$ of \eqref{eq:id} depends on the target law directly, through the outer
expectation, and on the trial law indirectly, through the regression $\mu_0$.
The next lemma computes the two derivatives separately.
 
\begin{lemma}[Pathwise derivatives]
\label{lem:C-deriv}
Let $\mu_{0,t}(x)\coloneq \E_{P_t}\{Y\mid X=x,\,T=0\}$. Under Assumption~\ref{ass:id}:
\begin{enumerate}[label=\textup{(\roman*)},leftmargin=2.6em,itemsep=2pt,topsep=3pt]
\item \emph{(Target channel.)} With the trial law held fixed,
\begin{equation}
  \frac{d}{dt}\Bigl[\E_{Q_t}\bigl\{Y^{\mathrm Q}-\mu_0(Z)\bigr\}\Bigr]_{t=0}
  =\EQ\Bigl[\bigl\{Y^{\mathrm Q}-\mu_0(Z)-\tau_Q\bigr\}\,
            \zeta_Q(O^{\mathrm Q})\Bigr].
  \label{eq:Qderiv}
\end{equation}
\item \emph{(Trial channel.)} For $P$-almost every $x$,
\begin{equation}
  \dot\mu_0(x)\coloneq \frac{d}{dt}\mu_{0,t}(x)\Big|_{t=0}
  =\EP\Bigl[\bigl\{Y-\mu_0(X)\bigr\}\,\zeta_P(O^{\mathrm R})
            \;\Bigm|\;X=x,\,T=0\Bigr].
  \label{eq:Pderiv}
\end{equation}
\end{enumerate}
\end{lemma}
 
\begin{proof}
\emph{(i).} Differentiation under the integral sign, justified by
square-integrability of $\zeta_Q$ and of $Y$, gives for any integrable $f$
\[
  \frac{d}{dt}\E_{Q_t}\{f\}\Big|_{t=0}
  =\int f\,\dot q_0\,d\nu
  =\EQ[f\,\zeta_Q],
\]
since $\dot q_0=q_0\,\zeta_Q$. Applying this with
$f=Y^{\mathrm Q}-\mu_0(Z)$ and then subtracting
$\tau_Q\,\EQ[\zeta_Q]=0$ yields \eqref{eq:Qderiv}. The subtraction is
cosmetic at this stage but is what makes the influence function centred.
 
\emph{(ii).} Let $p_t(y\mid x,0)$ denote the conditional density of $Y$ given
$(X,T)=(x,0)$ under $P_t$, so that
$\mu_{0,t}(x)=\int y\,p_t(y\mid x,0)\,d\nu(y)$. \corr{Writing
$p_t(y\mid x,0)=p_t(x,0,y)/p_t(x,0)$ and differentiating the logarithm, the
conditional score is $\zeta_P$ minus the score of the marginal law of
$(X,T)$, which equals $\EP[\zeta_P\mid X=x,T=0]$; differentiating
$\int p_t(y\mid x,0)\,d\nu(y)=1$ confirms the centring:}
\[
  \frac{\partial}{\partial t}\log p_t(y\mid x,0)\Big|_{t=0}
  =\zeta_P(y,x,0)-\EP[\zeta_P\mid X=x,\,T=0].
\]
Hence
\begin{align*}
  \dot\mu_0(x)
  &=\int y\,p_0(y\mid x,0)
     \bigl[\zeta_P(y,x,0)-\EP[\zeta_P\mid X=x,T=0]\bigr]d\nu(y)\\
  &=\EP[Y\zeta_P\mid X=x,T=0]-\mu_0(x)\,\EP[\zeta_P\mid X=x,T=0]\\
  &=\EP\bigl[\{Y-\mu_0(X)\}\zeta_P\mid X=x,\,T=0\bigr],
\end{align*}
which is \eqref{eq:Pderiv}.
\end{proof}
 
\begin{theorem}[Efficient influence function]
\label{thm:C-eif}
Under Assumption~\ref{ass:id}, the efficient influence function for $\tau_Q$ in the
two-sample nonparametric model is the pair
\begin{align}
  \varphi_Q(O^{\mathrm Q})&=Y^{\mathrm Q}-\mu_0(Z)-\tau_Q,
  \label{eq:eifQ}\\[2pt]
  \varphi_P(O^{\mathrm R})&=-\,A\,\rdag(X)\bigl\{Y-\mu_0(X)\bigr\},
  \qquad A=\frac{1-T}{e_0(X)},
  \label{eq:eifP}
\end{align}
in the sense that for every regular submodel with scores
$(\zeta_P,\zeta_Q)$,
\begin{equation}
  \dot\tau_Q
  =\EQ\bigl[\varphi_Q\,\zeta_Q\bigr]
   +\EP\bigl[\varphi_P\,\zeta_P\bigr].
  \label{eq:eif}
\end{equation}
Both contributions are centered, and under \eqref{eq:ratio} the semiparametric
efficiency bound for $\sqrt m\,(\that-\tau_Q)$ is
\begin{equation}
  V_{\mathrm{eff}}
  =\Var_Q\bigl[\varphi_Q(O^{\mathrm Q})\bigr]
   +\eta^{-1}\Var_P\bigl[\varphi_P(O^{\mathrm R})\bigr].
  \label{eq:Veff}
\end{equation}
\end{theorem}
 
\begin{proof}
\emph{Target contribution.} This is exactly
Lemma~\ref{lem:C-deriv}(i), which by \eqref{eq:eifQ} reads
$\EQ[\varphi_Q\zeta_Q]$.
 
\emph{Trial contribution.} Perturbing $P_t$ changes $\mu_{0,t}$ while $Q$ is
held fixed, so by \eqref{eq:id}
\[
  \frac{d}{dt}\Bigl[-\EQ\bigl[\mu_{0,t}(Z)\bigr]\Bigr]_{t=0}
  =-\EQ\bigl[\dot\mu_0(Z)\bigr].
\]
Two substitutions now convert this into a trial-sample expectation. First, the
Riesz identity of Lemma~\ref{lem:A-riesz} applied to $h=\dot\mu_0$
\corr{--- admissible, since $\dot\mu_0(x)^2\le C_Y\,\EP[\zeta_P^2\mid X=x,T=0]$
by Cauchy--Schwarz and Assumption~\ref{ass:id-mom}, and the right-hand side
integrates to at most $\underline e^{-1}\EP[\zeta_P^2]<\infty$, so
$\dot\mu_0\in L^2(P)$ ---} moves the expectation from
the target law to the trial law at the cost of the factor $\rdag$:
\[
  \EQ\bigl[\dot\mu_0(Z)\bigr]=\EP\bigl[\rdag(X)\,\dot\mu_0(X)\bigr].
\]
Second, by Lemma~\ref{lem:C-deriv}(ii) the inner object is a conditional
expectation given $\{X,T=0\}$, and the inverse-randomisation identity
\eqref{eq:invrand} of Lemma~\ref{lem:A-invrand}, applied with
$G(O^{\mathrm R})=\{Y-\mu_0(X)\}\zeta_P(O^{\mathrm R})$, rewrites it as \corr{a conditional expectation given $X=x$ alone, no longer
restricted to the control arm}:
\[
  \dot\mu_0(x)
  =\EP\Bigl[A\,\bigl\{Y-\mu_0(X)\bigr\}\zeta_P(O^{\mathrm R})\Bigm|X=x\Bigr].
\]
Combining the two displays and taking the expectation over $X$,
\[
  \EQ\bigl[\dot\mu_0(Z)\bigr]
  =\EP\Bigl[A\,\rdag(X)\bigl\{Y-\mu_0(X)\bigr\}\zeta_P(O^{\mathrm R})\Bigr],
\]
so that the trial contribution equals $\EP[\varphi_P\zeta_P]$ with
$\varphi_P$ as in \eqref{eq:eifP}. Adding the two contributions gives
\eqref{eq:eif}.
 
\emph{Centring.} $\EQ[\varphi_Q]=\EQ[Y^{\mathrm Q}-\mu_0(Z)]-\tau_Q=0$ by
Proposition~\ref{prop:C-id}, and by \eqref{eq:invrand2} with
$h=\rdag$ and $\mu=\mu_0$,
$\EP[\varphi_P]=-\EP[\rdag(X)\{\mu_0(X)-\mu_0(X)\}]=0$.
 
\emph{Efficiency.} The pair \eqref{eq:eifQ}--\eqref{eq:eifP} is a gradient of
$\tau_Q$ by \eqref{eq:eif}. { Write $L_0^2(P)$ for the mean-zero
square-integrable functions of $O^{\mathrm R}$ and decompose the trial score
space as
\[
  L_0^2(P)
  =\{a(X)\}\;\oplus\;\{s(X,T):\E[s\mid X]=0\}\;\oplus\;
   \{b(O^{\mathrm R}):\E[b\mid X,T]=0\},
\]
the three summands being the scores of the covariate law, of the treatment
mechanism, and of the conditional outcome law. In the model with known
randomisation probability $e_0$ the trial tangent space is the direct sum of
the first and third summands, and the target tangent space is all of
$L_0^2(Q)$. The canonical gradient is the projection of any gradient onto the
tangent space, so it suffices to check that $(\varphi_Q,\varphi_P)$ already
lies in it. The target component visibly does. For the trial component,
$\EP[\varphi_P\mid X,T]=0$ almost surely --- on $\{T=1\}$ because $A=0$, and
on $\{T=0\}$ because $\EP[Y-\mu_0(X)\mid X,T=0]=0$ by the definition of
$\mu_0$ --- so $\varphi_P$ lies in the third summand and is in particular
orthogonal to every treatment-mechanism score. Hence the pair is the
canonical, that is efficient, influence function; see
\citet[Section~25.4]{vaart1998asymptotic} and \citet{bickel1993efficient}
for the underlying projection calculus. Note that, since $\varphi_P$ has no
treatment-mechanism component, knowledge of $e_0$ does not change the
bound.} The identity \eqref{eq:Veff} then follows from the independence
of the two samples together with \eqref{eq:ratio}, the factor $\eta^{-1}$
arising because the trial contribution is averaged over $n$ observations while
the scaling is by $\sqrt m$\corr{; the formalisation of the two-sample bound
as the variance of the canonical gradient follows the convolution-theorem
argument of \citet{bickel1993efficient}}.
\end{proof}
 
\begin{remark}[The density ratio is a Riesz representer]
\label{rem:C-riesz}
The proof shows exactly where $\rdag$ enters. Recall that $\rdag$ is the object that converts
the target-population expectation $\EQ[h(Z)]$ into a trial-population
expectation, for the particular function $h=\dot\mu_0$ traced out by
perturbations of the trial law. Since the model is nonparametric, $\dot\mu_0$
ranges over a dense subset of $L^2(P)$ \corr{provided the conditional outcome
variance $\Var(Y\mid X,T=0)$ is positive $P$-almost everywhere --- a mild
nondegeneracy condition we tacitly assume in this remark ---} so the
conversion must hold for all such $h$\corr{, which} by the uniqueness half of
Lemma~\ref{lem:A-riesz} pins down the weight uniquely as $\rdag$. 
\end{remark}
 
\subsection{The exact second-order remainder}
\label{sub:C-rem}
 
\corr{The population error of the one-step estimating functional is
controlled by the following identity, which is exact --- no expansion and no
approximation; beyond Assumption~\ref{ass:id} it uses only square-integrability
of the candidate nuisances.}
 
\begin{proposition}[Second-order remainder]
\label{prop:C-rem}
For square-integrable candidate nuisances $(\mu,r)$, define the population
estimating functional
\begin{equation}
  \Psi(\mu,r)
  \coloneq \EQ\bigl[Y^{\mathrm Q}-\mu(Z)\bigr]
    -\EP\Bigl[A\,r(X)\bigl\{Y-\mu(X)\bigr\}\Bigr].
  \label{eq:Psi}
\end{equation}
Then, under Assumption~\ref{ass:id},
\begin{equation}
  \Psi(\mu,r)-\tau_Q
  \;=\;\EP\Bigl[\bigl\{r(X)-\rdag(X)\bigr\}\bigl\{\mu(X)-\mu_0(X)\bigr\}\Bigr].
  \label{eq:remainder}
\end{equation}
\end{proposition}
 
\begin{proof}
The proof consists in simplifying the two terms of \eqref{eq:Psi} by two
different devices and observing that what is left is a product.
 
\emph{Step 1:} By \eqref{eq:invrand2} of Lemma~\ref{lem:A-invrand}, applied
with $h=r$,
\begin{equation}
  \EP\Bigl[A\,r(X)\bigl\{Y-\mu(X)\bigr\}\Bigr]
  =\EP\Bigl[r(X)\bigl\{\mu_0(X)-\mu(X)\bigr\}\Bigr].
  \label{eq:rem-step1}
\end{equation}
Note what this does and does not use: it uses \corr{only the known
randomisation probability, through $A$, together with positivity}, and
it holds \emph{whatever} $r$ is. The candidate weight is simply carried along.
 
\emph{Step 2:} Substituting \eqref{eq:rem-step1} into
\eqref{eq:Psi} and subtracting
$\tau_Q=\EQ[Y^{\mathrm Q}]-\EQ[\mu_0(Z)]$ from
Proposition~\ref{prop:C-id},
\begin{equation}
  \Psi(\mu,r)-\tau_Q
  =\EQ\bigl[\mu_0(Z)-\mu(Z)\bigr]
   -\EP\Bigl[r(X)\bigl\{\mu_0(X)-\mu(X)\bigr\}\Bigr],
  \label{eq:rem-step2}
\end{equation}
the observed target means $\EQ[Y^{\mathrm Q}]$ cancelling exactly.
 
\emph{Step 3:} By the Riesz identity \eqref{eq:riesz} of
Lemma~\ref{lem:A-riesz}, applied with $h=\mu_0-\mu$,
\[
  \EQ\bigl[\mu_0(Z)-\mu(Z)\bigr]
  =\EP\Bigl[\rdag(X)\bigl\{\mu_0(X)-\mu(X)\bigr\}\Bigr].
\]
This holds \emph{whatever} $\mu$ is. Substituting into \eqref{eq:rem-step2},
\[
  \Psi(\mu,r)-\tau_Q
  =\EP\Bigl[\bigl\{\rdag(X)-r(X)\bigr\}\bigl\{\mu_0(X)-\mu(X)\bigr\}\Bigr],
\]
and reversing the sign of both factors leaves the product unchanged, which is
\eqref{eq:remainder}.
\end{proof}
 
\begin{remark}[Why the error is a product]
\label{rem:C-product}
 \corr{The inverse-randomisation identity}
converts the trial term into $\EP[r(\mu_0-\mu)]$ irrespective of $r$, while the
Riesz identity converts the target term into $\EP[\rdag(\mu_0-\mu)]$
irrespective of $\mu$ and subtracting produces the product. Three consequences are
used repeatedly below. First, \emph{double robustness}: the right-hand side of \eqref{eq:remainder}
vanishes as soon as either factor does, so $\Psi(\mu,r)=\tau_Q$ whenever
$\mu=\mu_0$ or $r=\rdag$, with no requirement on the other nuisance beyond
square-integrability. Second, \emph{the product-rate condition}: by Cauchy--Schwarz,
$|\Psi(\mu,r)-\tau_Q|\le\norm{r-\rdag}{L^2(P)}\norm{\mu-\mu_0}{L^2(P)}$, so
root-$N$ inference requires only that the \emph{product} of the two estimation
errors be $o(N^{-1/2})$, not that either be so individually. Third,  \eqref{eq:remainder} is \emph{the same}
identity for every procedure of the form \eqref{eq:Psi}, whether $r$ is built
from a fitted sampling score or from a transport plan. 
\end{remark}
 
\section{The Uncalibrated Entropic Transport Weight}
\label{app:D}
 
\corr{The previous section} reduces everything to the approximation of $\rdag$ by the transport
weight. In this section, we answer the question for the \emph{uncalibrated}
unbalanced entropic plan, before any Riesz constraints are imposed. We present here three facts, and together they motivate the calibrated construction of the next section. The source marginal has a closed form in which\corr{, at a fixed dual
potential,} the relaxation parameter $\rho$ \corr{acts} through the exponent
$\gamma=\varepsilon/(\varepsilon+\rho)$, a power-law shrinkage towards the
uniform weight (Section~\ref{sub:D-closed}). At $\rho=0$ that weight is
exactly a self-normalised Gibbs-kernel \corr{smoothing of the ratio
$q/S_\varepsilon p$}, so that entropic transport is, in a precise sense, a
smoothing construction (Section~\ref{sub:D-kernel}). Consequently the weight converges to $\rdag$ as $\varepsilon\downarrow0$ but carries, at any \emph{fixed} $\varepsilon>0$, a bias of order $\varepsilon$ \corr{on interior compacta} whose effect on the Riesz equations we compute in closed form (Section~\ref{sub:D-expansion}). That bias is a population object. It does not diminish as the sample grows, and removing it is exactly what the calibration constraints will do.
 
\subsection{The unbalanced program and its source marginal}
\label{sub:D-closed}
 
Recall that the covariate distributions are assumed to admit
densities $p$ and $q$ on $\calX$, and the transport plan is likewise
represented by a density $\pi(x,z)$ on $\calX\times\calX$, with marginals
$\pi_1$ and $\pi_2$ as in \eqref{eq:marginals}. Recall the unbalanced entropic optimal transport from Definition~\ref{def:D-ubot}.

Existence and uniqueness in \eqref{eq:ubot} of Definition~\ref{def:D-ubot} are standard \corr{ in the theory of unbalanced entropic transport \citep{liero2018optimal,chizat2018unbalanced}}.  The cost term is
linear in $\pi$, the feasible set $\{\pi\ge0:\pi_2=q\}$ is convex, the map
$\pi\mapsto\pi_1$ is linear so that the marginal divergence is convex, and
$\pi\mapsto\KL(\pi\|p\otimes q)$ is strictly convex on its effective domain.
Lower semicontinuity and \corr{weak-$L^1$ compactness of the
Kullback--Leibler sublevel sets (uniform integrability plus the
Dunford--Pettis theorem)} give existence, and strict convexity gives
uniqueness.
 
\begin{proposition}[Closed form of the source weight]
\label{prop:D-closed}
The problem \eqref{eq:ubot} has a unique solution. \corr{Suppose the
solution is almost everywhere positive and admits a dual potential $g$ for
the target-marginal constraint with
$\int e^{\{g(z)-c(x,z)\}/\varepsilon}q(z)\,dz<\infty$ for $P$-almost every
$x$; such interiority and duality hold under two-sided bounds on $p$ and $q$
on $\calX$ by the Fenchel--Rockafellar duality theory of entropic transport
\citep{liero2018optimal,chizat2018scaling}. Then}
\begin{equation}
  \frac{\pi(x,z)}{p(x)\,q(z)}
  =\exp\Bigl\{\frac{g(z)-c(x,z)}{\varepsilon}\Bigr\}\,
   r_{\varepsilon,\rho}(x)^{-\rho/\varepsilon},
  \label{eq:D-kkt}
\end{equation}
and consequently, with $\gamma=\varepsilon/(\varepsilon+\rho)$ as in
\eqref{eq:gamma},
\begin{equation}
  r_{\varepsilon,\rho}(x)
  =\Bigl[\int\exp\Bigl\{\frac{g(z)-c(x,z)}{\varepsilon}\Bigr\}\,q(z)\,dz
   \Bigr]^{\gamma},
  \label{eq:D-source}
\end{equation}
where $g$ is determined by the target-marginal constraint
\begin{equation}
  \int\exp\Bigl\{\frac{g(z)-c(x,z)}{\varepsilon}\Bigr\}\,
       r_{\varepsilon,\rho}(x)^{-\rho/\varepsilon}\,p(x)\,dx=1
  \qquad\text{for $Q$-almost every }z.
  \label{eq:D-target}
\end{equation}
Moreover $\EP[r_{\varepsilon,\rho}(X)]=1$ exactly, for every
$\varepsilon>0$ and $\rho\ge0$.
\end{proposition}
 
\begin{proof}
Existence and uniqueness were argued above. For the first-order condition,
introduce a multiplier $g(z)$ for the constraint $\pi_2=q$ and write
$r\coloneq \pi_1/p$, so that $\KL(\pi_1\|p)=\int\{r\log r-r+1\}\,p$. Since
$r(x)p(x)=\int\pi(x,z)\,dz$, an infinitesimal perturbation of $\pi(x,z)$
changes the source marginal at $x$ by the same amount, so the first variation
of the marginal divergence is $\log r(x)$; that of the plan divergence is
$\log\{\pi(x,z)/(p(x)q(z))\}$. Stationarity at an interior optimiser \corr{(interior by hypothesis)} therefore
requires
\[
  c(x,z)+\varepsilon\log\frac{\pi(x,z)}{p(x)q(z)}
  +\rho\log r(x)-g(z)=0,
\]
and exponentiating gives \eqref{eq:D-kkt}.
 
Integrating \eqref{eq:D-kkt} in $z$ and using $\int\pi(x,z)\,dz=r(x)p(x)$ gives
\[
  r(x)=r(x)^{-\rho/\varepsilon}
       \int e^{\{g(z)-c(x,z)\}/\varepsilon}\,q(z)\,dz,
\]
so that $r(x)^{1+\rho/\varepsilon}$ equals the integral\corr{, which, after raising} to the power
$\gamma=(1+\rho/\varepsilon)^{-1}$ yields \eqref{eq:D-source}. Integrating
\eqref{eq:D-kkt} in $x$ and imposing $\pi_2=q$ gives \eqref{eq:D-target}.
 
Finally, the constraint $\pi_2=q$ forces $\iint\pi=\int q=1$, and the total
mass of $\pi_1$ equals that of $\pi$; hence
$\EP[r_{\varepsilon,\rho}(X)]=\int\pi_1=1$.
\end{proof}
 
\begin{remark}[$\rho$ as a shrinkage exponent]
\label{rem:D-shrink}
Denote the part inside the bracket of \eqref{eq:D-source} by $\varrho(x)$, so that
$r_{\varepsilon,\rho}=\varrho^{\,\gamma}$ exactly. Holding the dual potential fixed, the exponent $\gamma\in(0,1)$ decreases with $\rho$ and as a consequence, compresses the
log-weight towards zero. \corr{This explains why} the relaxation
parameter performs a \emph{power-law shrinkage} of the weight towards the
uniform weight $1$. \corr{It is \emph{not} true that
$r_{\varepsilon,\rho}=r_{\varepsilon,0}^{\,\gamma}$. The dual potential $g$
depends on $\rho$ through the factor $r^{-\rho/\varepsilon}$ in
\eqref{eq:D-target}, so the displayed identity holds at fixed dual potential,
and to first order as $\rho\downarrow0$, but not exactly.} 
\end{remark}
 
\subsection{The source weight as a Gibbs-smoothed density ratio}
\label{sub:D-kernel}
 
Throughout this subsection $\calX\subseteq\R^d$ and
$c(x,z)=\norm{x-z}{2}^2$. Recall the Gibbs smoothing
operator $S_\varepsilon$ of \eqref{eq:gibbs}, and let
\begin{equation}
  \kappa_\varepsilon\coloneq \int_{\R^d}e^{-\norm{u}{2}^2/\varepsilon}\,du,
  \qquad\text{so that}\qquad
  \int e^{-\norm{x-z}{2}^2/\varepsilon}f(x)\,dx
  =\kappa_\varepsilon\,S_\varepsilon f(z)
  \label{eq:kappa}
\end{equation}
for every integrable $f$; the constant $\kappa_\varepsilon$ will cancel and is
introduced only to avoid writing the Gaussian normalisation explicitly.
 
\begin{proposition}[Kernel fixed-point representation]
\label{prop:D-kernel}
Let $\rho\ge0$ and $\gamma=\varepsilon/(\varepsilon+\rho)$. The source weight
of Definition~\ref{def:D-ubot} satisfies the exact nonlinear fixed-point
equation
\begin{equation}
  r_{\varepsilon,\rho}
  =\Bigl[
     S_\varepsilon\Bigl\{
       \frac{q}{S_\varepsilon\bigl(p\,r_{\varepsilon,\rho}^{-\rho/\varepsilon}\bigr)}
     \Bigr\}
   \Bigr]^{\gamma}.
  \label{eq:D-fixedpt}
\end{equation}
Equivalently, setting $u\coloneq r_{\varepsilon,\rho}^{1/\gamma}$, so that
$r_{\varepsilon,\rho}=u^{\gamma}$ and
$r_{\varepsilon,\rho}^{-\rho/\varepsilon}=u^{-(1-\gamma)}$,
\begin{equation}
  u=S_\varepsilon\Bigl\{\frac{q}{S_\varepsilon\bigl(p\,u^{-(1-\gamma)}\bigr)}
    \Bigr\}.
  \label{eq:D-fixedpt-u}
\end{equation}
In the special case $\rho=0$ one has $\gamma=1$, the fixed point becomes
explicit, and
\begin{equation}
  r_{\varepsilon,0}=S_\varepsilon\Bigl[\frac{q}{S_\varepsilon p}\Bigr].
  \label{eq:D-NW}
\end{equation}
\end{proposition}
 
\begin{proof}
\emph{Step 1:} Put
$h(z)\coloneq \exp\{g(z)/\varepsilon\}$. With the squared Euclidean cost, the
target-marginal condition \eqref{eq:D-target} reads
\[
  h(z)\int e^{-\norm{x-z}{2}^2/\varepsilon}\,
       p(x)\,r_{\varepsilon,\rho}(x)^{-\rho/\varepsilon}\,dx=1
  \qquad\text{for $Q$-a.e. }z,
\]
so that, by \eqref{eq:kappa} applied to
$f=p\,r_{\varepsilon,\rho}^{-\rho/\varepsilon}$,
\begin{equation}
  h(z)=\Bigl\{\kappa_\varepsilon\,
        S_\varepsilon\bigl(p\,r_{\varepsilon,\rho}^{-\rho/\varepsilon}\bigr)(z)
       \Bigr\}^{-1}.
  \label{eq:D-h}
\end{equation}
 
\emph{Step 2:} By \eqref{eq:D-source},
\[
  r_{\varepsilon,\rho}(x)
  =\Bigl[\int h(z)\,e^{-\norm{x-z}{2}^2/\varepsilon}\,q(z)\,dz\Bigr]^{\gamma},
\]
and inserting \eqref{eq:D-h} makes the integrand
$\kappa_\varepsilon^{-1}e^{-\norm{x-z}{2}^2/\varepsilon}$ times
$q(z)/S_\varepsilon(p\,r^{-\rho/\varepsilon})(z)$. By \eqref{eq:kappa} again,
now applied in the variable $z$, the factor $\kappa_\varepsilon$ cancels and
the integral is precisely
$S_\varepsilon\{q/S_\varepsilon(p\,r^{-\rho/\varepsilon})\}(x)$. This is
\eqref{eq:D-fixedpt}.
 
\emph{Step 3:} Since
$u=r_{\varepsilon,\rho}^{1/\gamma}$ we have
$r_{\varepsilon,\rho}=u^{\gamma}$ and
\[
  r_{\varepsilon,\rho}^{-\rho/\varepsilon}
  =u^{-\gamma\rho/\varepsilon}
  =u^{-\rho/(\varepsilon+\rho)}
  =u^{-(1-\gamma)} .
\]
Raising both sides of \eqref{eq:D-fixedpt} to the power $1/\gamma$ gives
\eqref{eq:D-fixedpt-u}. Taking $\rho=0$, so that $\gamma=1$ and
$r_{\varepsilon,0}^{-\rho/\varepsilon}\equiv1$, reduces \eqref{eq:D-fixedpt} to
\eqref{eq:D-NW}.
\end{proof}
 
\begin{remark}
\label{rem:D-interp}
Identity \eqref{eq:D-NW} illustrates the  sense in which the uncalibrated
entropic weight is a density-ratio smoother. The quantity $S_\varepsilon p$ is a Gibbs-kernel smoothing of the
trial density evaluated at the \emph{target} point $z$ and the ratio
$q/S_\varepsilon p$ is therefore a density-ratio-like object with a smoothed
denominator. The outer $S_\varepsilon$ transports that ratio back to the
\emph{trial} point $x$ using the same kernel. For the squared Euclidean cost, $S_\varepsilon$ is convolution with a
Gaussian of covariance $(\varepsilon/2)I_d$, so $\sqrt{\varepsilon/2}$ is the
smoothing scale, and the self-normalisation through $S_\varepsilon p$ is the
population counterpart of the Sinkhorn column-scaling step, which is what
enforces $\EP[r_{\varepsilon,0}(X)]=1$.
 
It is worth noting that \eqref{eq:D-NW} is not exactly the
Nadaraya--Watson density-ratio estimator, which would be
$S_\varepsilon q/S_\varepsilon p$ evaluated at a common point. The entropic
construction forms the ratio at the target variable and then smooths back to
the source variable. For general $\rho>0$ the situation
is genuinely more complicated. For instance, \eqref{eq:D-fixedpt} is a nonlinear fixed-point
equation, the inner object $S_\varepsilon(p\,r^{-\rho/\varepsilon})$ being a
smoothed \emph{adaptively tilted} trial density, and there is in general no
reduction of $r_{\varepsilon,\rho}$ to a fixed power of $r_{\varepsilon,0}$.
\end{remark}

\subsection{Small-regularisation expansion and the Riesz imbalance}
\label{sub:D-expansion}

The expansion below is the only place in the supplement where smoothness of the
underlying densities is used. The hypotheses are stated locally rather than
added to the standing assumptions of Appendix~\ref{app:A}, since they play no
role in the asymptotic theory.

Recall from Proposition~\ref{prop:D-kernel} that the uncalibrated population
weight satisfies
\[
r_{\varepsilon,\rho}
=
\left[
S_\varepsilon\left\{
\frac{q}
{S_\varepsilon(p\,r_{\varepsilon,\rho}^{-\rho/\varepsilon})}
\right\}
\right]^{\varepsilon/(\varepsilon+\rho)},
\]
while \(r_{\varepsilon,0}=S_\varepsilon\{q/(S_\varepsilon p)\}\). Hence
\[
r_{\varepsilon,\rho}-\rdag
=
\{r_{\varepsilon,\rho}-r_{\varepsilon,0}\}
+
\{r_{\varepsilon,0}-\rdag\}.
\]
We expand the source-relaxation and entropic-smoothing components separately.

\begin{theorem}[Expansion of the uncalibrated weight]
\label{thm:D-laplace}
Let \(\calK\) be \corr{a compact subset of the interior of} \(\calX\). Suppose that on a neighbourhood of
\(\calK\), the densities \(p\) and \(q\) are four times continuously
differentiable with derivatives bounded through order four,
\(p\ge\underline p>0\), and \(\rdag=q/p\) is bounded above and away from zero.
\corr{Suppose also that \(p\) and \(q\) are bounded globally, that
\(p\ge\underline p_0>0\) on all of \(\calX\) --- so that
\(q/S_\varepsilon p\) is bounded on \(\calX\) uniformly in small
\(\varepsilon\) --- and that all densities are extended by zero off
\(\calX\); the Gaussian tail outside the stated neighbourhood then
contributes \(O(e^{-c/\varepsilon})\) to every smoothing below.}

For fixed \(\varepsilon>0\), put \(\delta=\rho/\varepsilon\), and suppose that
for sufficiently small \(\delta\), the positive solution
\(\delta\mapsto r_{\varepsilon,\varepsilon\delta}\) is twice continuously
differentiable as an \corr{\(\ell^\infty(\calX)\)}-valued map and remains uniformly
bounded above and away from zero \corr{on \(\calX\) (the differentiations
under \(S_\varepsilon\) and under \(\EP\) in the proof integrate over all of
\(\calX\), so uniformity on \(\calK\) alone would not suffice)}. Define
\[
\mathcal B_\varepsilon(x)
\coloneq 
S_\varepsilon\left[
\frac{
q\,S_\varepsilon\{p\log r_{\varepsilon,0}\}
}{
(S_\varepsilon p)^2
}
\right](x)
-
r_{\varepsilon,0}(x)\log r_{\varepsilon,0}(x).
\]
Then, uniformly on \(\calK\) as \(\rho/\varepsilon\to0\),
\begin{equation}
r_{\varepsilon,\rho}-r_{\varepsilon,0}
=
\frac{\rho}{\varepsilon}\mathcal B_\varepsilon
+
O_\varepsilon\left\{
\left(\frac{\rho}{\varepsilon}\right)^2
\right\},
\label{eq:D-rho-expansion}
\end{equation}
where the remainder constant may depend on the fixed value of
\(\varepsilon\). Moreover,
\(\EP[\mathcal B_\varepsilon(X)]=0\).

For the entropic component, let
\(\Delta\coloneq \sum_{k=1}^d\partial^2/\partial x_k^2\) denote the Laplacian.
Uniformly on \(\calK\) as \(\varepsilon\downarrow0\),
\begin{equation}
r_{\varepsilon,0}
=
\rdag
+
\frac{\varepsilon}{4}
\left\{
\Delta\rdag
-
\rdag\,\frac{\Delta p}{p}
\right\}
+
O(\varepsilon^2).
\label{eq:D-expansion}
\end{equation}
Consequently, combining the two local expansions,
\begin{equation}
\begin{aligned}
r_{\varepsilon,\rho}-\rdag
={}
\frac{\rho}{\varepsilon}\mathcal B_\varepsilon
+
\frac{\varepsilon}{4}
\left\{
\Delta\rdag
-
\rdag\,\frac{\Delta p}{p}
\right\}+
O_\varepsilon\left\{
\left(\frac{\rho}{\varepsilon}\right)^2
\right\}
+
O(\varepsilon^2).
\end{aligned}
\label{eq:D-total-bias}
\end{equation}
If\corr{, in addition, the differentiability radius in \(\delta\), the
two-sided bounds on the branch, and} the second derivative of the
source-relaxation solution branch \corr{are} bounded uniformly for all
sufficiently small \(\varepsilon\)\corr{, and
\(\sup_{\calK}|\mathcal B_\varepsilon|=O(1)\) as
\(\varepsilon\downarrow0\)}, the first remainder is uniform in
\(\varepsilon\), and \eqref{eq:D-total-bias} may be read in the joint
regime \(\varepsilon\to0\), \(\rho/\varepsilon\to0\), with remainder
\(O\{(\rho/\varepsilon)^2+\varepsilon^2\}\).
\end{theorem}

\begin{proof}
We treat the source-relaxation and entropic-smoothing components separately.

\emph{Step 1:}
Fix \(\varepsilon>0\), set \(\delta\coloneq \rho/\varepsilon\), and write
\(R_\delta\coloneq r_{\varepsilon,\varepsilon\delta}\) and
\(R_0\coloneq r_{\varepsilon,0}\). Since
\(\varepsilon/(\varepsilon+\rho)=1/(1+\delta)\), the fixed-point equation
becomes
\[
R_\delta
=
\left[
S_\varepsilon\left\{
\frac{q}{S_\varepsilon(pR_\delta^{-\delta})}
\right\}
\right]^{1/(1+\delta)}.
\]
Define \(D_\delta\coloneq S_\varepsilon(pR_\delta^{-\delta})\) and
\(A_\delta\coloneq S_\varepsilon(q/D_\delta)\), so that
\(R_\delta=A_\delta^{1/(1+\delta)}\). At \(\delta=0\),
\(R_\delta^{-\delta}=1\), and hence
\(D_0=S_\varepsilon p\) and
\(A_0=S_\varepsilon\{q/(S_\varepsilon p)\}=R_0\).
Since
\(R_\delta^{-\delta}=\exp\{-\delta\log R_\delta\}\),
the chain rule gives
\[
\frac{d}{d\delta}R_\delta^{-\delta}
=
R_\delta^{-\delta}
\left\{
-\log R_\delta
-\delta\frac{R_\delta'}{R_\delta}
\right\}.
\]
At \(\delta=0\), \(R_\delta^{-\delta}=1\), and the term containing
\(R_\delta'\) vanishes because it is multiplied by \(\delta\). Therefore
\[
\left.
\frac{d}{d\delta}R_\delta^{-\delta}
\right|_{\delta=0}
=
-\log R_0.
\]
Thus, although \(R_\delta\) itself depends on \(\delta\), this dependence does
not contribute to the first-order derivative of \(R_\delta^{-\delta}\).

\emph{Step 2:}
By linearity of \(S_\varepsilon\),
\[
D_0'
=
S_\varepsilon\left[
p\left.
\frac{d}{d\delta}R_\delta^{-\delta}
\right|_{\delta=0}
\right]
=
-S_\varepsilon\{p\log R_0\}.
\]
Since \(A_\delta=S_\varepsilon(q/D_\delta)\),
\[
A_\delta'
=
S_\varepsilon\left(
-\frac{qD_\delta'}{D_\delta^2}
\right).
\]
Using \(D_0=S_\varepsilon p\) and
\(D_0'=-S_\varepsilon\{p\log R_0\}\) gives
\[
A_0'
=
S_\varepsilon\left[
\frac{
q\,S_\varepsilon\{p\log R_0\}
}{
(S_\varepsilon p)^2
}
\right].
\]
Taking logarithms of
\(R_\delta=A_\delta^{1/(1+\delta)}\) gives
\(\log R_\delta=(1+\delta)^{-1}\log A_\delta\).
Differentiating,
\[
\frac{R_\delta'}{R_\delta}
=
-\frac{\log A_\delta}{(1+\delta)^2}
+
\frac{1}{1+\delta}\frac{A_\delta'}{A_\delta}.
\]
Evaluating at \(\delta=0\) and using \(A_0=R_0\),
\[
\frac{R_0'}{R_0}
=
-\log R_0+\frac{A_0'}{R_0},
\]
and therefore
\[
R_0'
=
A_0'-R_0\log R_0
=
S_\varepsilon\left[
\frac{
q\,S_\varepsilon\{p\log R_0\}
}{
(S_\varepsilon p)^2
}
\right]
-
R_0\log R_0
=
\mathcal B_\varepsilon.
\]

\emph{Step 3:}
By the assumed twice continuous differentiability of
\(\delta\mapsto R_\delta\) as an
\(\ell^\infty(\calK)\)-valued map, Taylor's theorem around \(\delta=0\)
gives, uniformly on \(\calK\),
\[
R_\delta
=
R_0+\delta R_0'
+
O_\varepsilon(\delta^2)
=
R_0+\delta\mathcal B_\varepsilon
+
O_\varepsilon(\delta^2).
\]
Substituting \(\delta=\rho/\varepsilon\) proves
\eqref{eq:D-rho-expansion}. The notation
\(O_\varepsilon(\delta^2)\) emphasizes that this is a local expansion in
\(\delta\) with \(\varepsilon\) held fixed.

\corr{The differentiability of the solution branch is a genuine hypothesis;
we do not derive it from primitives here. At \(\delta=0\), the derivative
with respect to \(R_\delta\) inside \(R_\delta^{-\delta}\) is multiplied by
\(\delta\), which is the simplification underlying the preceding
calculation.

The identity \(\EP[\mathcal B_\varepsilon(X)]=0\) holds exactly and needs no
differentiability at all: \(S_\varepsilon\) is self-adjoint on \(L^2(dx)\)
(its kernel is symmetric), so
$$
\begin{aligned}
    \EP\left[S_\varepsilon\left\{\frac{q\,S_\varepsilon(p\log r_{\varepsilon,0})}
{(S_\varepsilon p)^2}\right\}\right]
&=\int (S_\varepsilon p)\,\frac{q\,S_\varepsilon(p\log r_{\varepsilon,0})}
{(S_\varepsilon p)^2}\\
& =\int \frac{q}{S_\varepsilon p}\,S_\varepsilon(p\log r_{\varepsilon,0})\\
&=\int S_\varepsilon\!\left(\frac{q}{S_\varepsilon p}\right)p\log r_{\varepsilon,0}\\
&=\EP[r_{\varepsilon,0}\log r_{\varepsilon,0}],
\end{aligned}
$$
and subtracting the second term of \(\mathcal B_\varepsilon\) gives zero.
(This computation is at \(\delta=0\); it is consistent with, but does not
require, differentiating the exact mass identity
\(\EP[R_\delta(X)]=1\) of Proposition~\ref{prop:D-closed}.)}
Thus the leading source-relaxation perturbation changes the shape of the
weight but has zero total mass. 

\emph{Step 4:}
We now set \(\rho=0\), so that
\(r_{\varepsilon,0}
=S_\varepsilon\{q/(S_\varepsilon p)\}\).
By Lemma~\ref{lem:A-heat},
\[
S_\varepsilon p
=
p+\frac{\varepsilon}{4}\Delta p+O(\varepsilon^2)
=
p\left\{
1+\frac{\varepsilon}{4}\frac{\Delta p}{p}
+O(\varepsilon^2)
\right\}
\]
uniformly on \(\calK\). Since \(p\ge\underline p>0\), the reciprocal may be
expanded uniformly using
\((1+u)^{-1}=1-u+O(u^2)\), yielding
\[
\frac{1}{S_\varepsilon p}
=
\frac{1}{p}
\left\{
1-\frac{\varepsilon}{4}\frac{\Delta p}{p}
\right\}
+
O(\varepsilon^2).
\]
Multiplying by \(q\) and using \(\rdag=q/p\) gives
\[
\frac{q}{S_\varepsilon p}
=
\rdag
-
\frac{\varepsilon}{4}
\rdag\frac{\Delta p}{p}
+
O(\varepsilon^2).
\]
Applying \(S_\varepsilon\) to the preceding expansion \corr{--- legitimate
because the expansion of \(q/S_\varepsilon p\) holds uniformly on a compact
neighbourhood \(\calK'\) of \(\calK\) with
\(\operatorname{dist}(\calK,\partial\calK')>0\), while outside \(\calK'\)
the integrand is bounded by the global hypothesis
\(p\ge\underline p_0\) and the Gaussian kernel contributes
\(O(e^{-c/\varepsilon})\) ---}
\[
r_{\varepsilon,0}
=
S_\varepsilon\rdag
-
\frac{\varepsilon}{4}
S_\varepsilon\left(
\rdag\frac{\Delta p}{p}
\right)
+
O(\varepsilon^2).
\]
Another application of Lemma~\ref{lem:A-heat} gives
\(S_\varepsilon\rdag
=
\rdag+(\varepsilon/4)\Delta\rdag+O(\varepsilon^2)\).
Moreover,
\[
S_\varepsilon\left(
\rdag\frac{\Delta p}{p}
\right)
=
\rdag\frac{\Delta p}{p}+O(\varepsilon).
\]
The latter remainder is multiplied by \(\varepsilon/4\), and therefore
contributes only \(O(\varepsilon^2)\). Consequently,
\[
r_{\varepsilon,0}
=
\rdag
+
\frac{\varepsilon}{4}\Delta\rdag
-
\frac{\varepsilon}{4}
\rdag\frac{\Delta p}{p}
+
O(\varepsilon^2),
\]
which proves \eqref{eq:D-expansion}.
The decomposition
\[
r_{\varepsilon,\rho}-\rdag
=
\{r_{\varepsilon,\rho}-r_{\varepsilon,0}\}
+
\{r_{\varepsilon,0}-\rdag\}
\]
is exact. Substituting \eqref{eq:D-rho-expansion} and
\eqref{eq:D-expansion} gives \eqref{eq:D-total-bias}.

The two expansions have different local parameters. The source-relaxation
expansion is in \(\delta=\rho/\varepsilon\) with \(\varepsilon\) fixed,
whereas the entropic expansion is in \(\varepsilon\). Thus neither limiting
condition is required for the exact decomposition itself; they are required
only for the corresponding first-order approximations. Under the additional
uniformity condition stated in the theorem, the two remainder terms can be
combined into \(O\{(\rho/\varepsilon)^2+\varepsilon^2\}\).
\end{proof}

Theorem~\ref{thm:D-laplace} separates the two population distortions induced
by uncalibrated regularized transport. For fixed \(\varepsilon\), source
relaxation contributes locally through
\((\rho/\varepsilon)\mathcal B_\varepsilon\), whereas entropic smoothing
contributes through the \(O(\varepsilon)\) differential term in
\eqref{eq:D-expansion}. Neither contribution is a finite-sample estimation
error.

\begin{corollary}[Riesz imbalance under regularization]
\label{cor:D-imbalance}
Under the hypotheses of Theorem~\ref{thm:D-laplace}, let
\corr{\(h\in C^2(\calX)\) have bounded derivatives and compact support in
the interior of \(\calX\), with the expansion hypotheses holding on a
neighbourhood of that support}. Then
\begin{equation}
\begin{aligned}
&\EP\bigl[r_{\varepsilon,\rho}(X)h(X)\bigr]
-
\EQ\bigl[h(X^Q)\bigr]\\
&=
\frac{\rho}{\varepsilon}
\EP\bigl[\mathcal B_\varepsilon(X)h(X)\bigr]
-
\frac{\varepsilon}{4}
\int
\nabla h\cdot
\bigl(p\nabla\rdag-\rdag\nabla p\bigr)\,dx+
O_\varepsilon\left\{
\left(\frac{\rho}{\varepsilon}\right)^2
\right\}
+
O(\varepsilon^2).
\end{aligned}
\label{eq:D-imbalance}
\end{equation}
For constant \(h\), the Riesz imbalance is exactly zero for every
\((\varepsilon,\rho)\).
\end{corollary}

\begin{proof}
By the Riesz identity, the left-hand side of \eqref{eq:D-imbalance} equals
\[
\EP\bigl[
\{r_{\varepsilon,\rho}(X)-\rdag(X)\}h(X)
\bigr].
\]
Substituting \eqref{eq:D-total-bias} gives the source-relaxation contribution
\((\rho/\varepsilon)\EP[\mathcal B_\varepsilon h]\) and the entropic
contribution
\[
\frac{\varepsilon}{4}
\int h
\bigl\{
p\Delta\rdag-\rdag\Delta p
\bigr\}\,dx.
\]
To rewrite the latter, note that
\[
\nabla\!\cdot(p\nabla\rdag)
=
\nabla p\cdot\nabla\rdag+p\Delta\rdag,
\qquad
\nabla\!\cdot(\rdag\nabla p)
=
\nabla\rdag\cdot\nabla p+\rdag\Delta p.
\]
Subtracting these identities cancels the cross terms and gives
\[
p\Delta\rdag-\rdag\Delta p
=
\nabla\!\cdot
\bigl(p\nabla\rdag-\rdag\nabla p\bigr).
\]
Integration by parts under the stated support condition\corr{ --- the
boundary terms vanish because $h$ vanishes near $\partial\calX$ ---}
therefore gives
\[
\int h
\bigl\{
p\Delta\rdag-\rdag\Delta p
\bigr\}\,dx
=
-\int
\nabla h\cdot
\bigl(p\nabla\rdag-\rdag\nabla p\bigr)\,dx,
\]
which proves \eqref{eq:D-imbalance}.

For \(h\equiv1\), Proposition~\ref{prop:D-closed} gives
\(\EP[r_{\varepsilon,\rho}(X)]=1=\EQ[1]\) exactly, so the Riesz imbalance
vanishes for every \((\varepsilon,\rho)\)\corr{, consistently with
\(\EP[\mathcal B_\varepsilon(X)]=0\); note that constant functions are not
covered by the display, whose hypotheses exclude them, and their exact
balance is enforced by the mass constraint rather than by the interior
expansion}.
\end{proof}

\corr{\begin{remark}[Boundary effects]
\label{rem:D-boundary}
The restriction to test functions supported in the interior is not cosmetic.
Within an $O(\sqrt\varepsilon)$ collar of $\partial\calX$ the expansion
\eqref{eq:D-expansion} fails by an $O(1)$ relative amount, because
$S_\varepsilon p$ under-counts $p$ near the boundary. For an $h$ whose
support meets the collar, the boundary layer contributes an additional term
of the same $O(\varepsilon)$ order as the interior term in
\eqref{eq:D-imbalance}, proportional to the normal derivative of $h$ at the
boundary; and in $L^2(P)$ the weight bias itself is generally of order
$\varepsilon^{1/4}$ rather than $\varepsilon$, a boundary layer of width
$\sqrt\varepsilon$ carrying an $O(1)$ error. None of this affects the use
made of these results: the bias is a population object whatever its exact
order, and the calibration constraints of Appendix~\ref{app:E} remove its
estimand-relevant part exactly, at fixed $(\varepsilon,\rho)$, with no
boundary hypothesis.
\end{remark}}

\begin{remark}[Why calibration is needed]
\label{rem:D-why}
 Note that both $r_{\varepsilon,\rho}$ and $\rdag$ are population objects.
Holding $(\varepsilon,\rho)$ fixed while increasing the sample size makes an
empirical uncalibrated transport weight converge towards
$r_{\varepsilon,\rho}$, however it does not move $r_{\varepsilon,\rho}$ towards
$\rdag$. The weight error
$b_{\varepsilon,\rho}\coloneq r_{\varepsilon,\rho}-\rdag$ therefore persists, and by
Proposition~\ref{prop:C-rem} it enters the estimator through
$\EP[b_{\varepsilon,\rho}(\mu-\mu_0)]$, which is \corr{bounded by, and in
general of the order of,}
$\norm{b_{\varepsilon,\rho}}{L^2(P)}\norm{\mu-\mu_0}{L^2(P)}$\corr{; it
therefore need not be $o(N^{-1/2})$, which} blocks
root-$N$ inference. At $\rho=0$, Theorem~\ref{thm:D-laplace} shows this error
is $O(\varepsilon)$ \corr{uniformly on interior compacta; in $L^2(P)$ the
boundary layer generally makes it of the larger order $\varepsilon^{1/4}$
(Remark~\ref{rem:D-boundary}), which only strengthens the point}; for
$\rho>0$ there is an additional source-relaxation component
$r_{\varepsilon,\rho}-r_{\varepsilon,0}$\corr{, which is not $O(\rho)$
without a separate perturbation argument (Remark~\ref{rem:D-shrink}).}
 
The route taken in Appendix~\ref{app:E} is to hold $(\varepsilon,\rho)$ fixed and remove
the \emph{estimand-relevant} part of this error by constraint rather than by
undersmoothing. Calibration does not ask the finite-$J$ weight to equal $\rdag$
pointwise; it imposes $\EP[r(X)b_J(X)]=\EQ[b_J(Z)]$, equivalently
$\EP[\{r-\rdag\}b_J]=0$, so that the imbalance is annihilated exactly on
$H_J$ even though a pointwise error remains.
 
\end{remark}

\section{Riesz Calibration: Dual Structure and Convergence Rates}
\label{app:E}
 
In the previous section, we have seen that OT based weights even at the population level is geometrically sensible but
misaligned with the estimand, i.e., at fixed regularisation it violates the Riesz
equations by an amount of order $\varepsilon$ that no sample size removes. This
section repairs that by imposing the finite-dimensional Riesz equations
\emph{inside} the transport program, and then shows what the repair buys.
 
The development has a definite shape, and it is worth describing before the
details. Section~\ref{sub:E-kkt} derives the Karush--Kuhn--Tucker (KKT) system of the calibrated program and finds that the calibration multiplier enters as a multiplicative exponential tilt of the uncalibrated weight.
Section~\ref{sub:E-tilt} makes the  observation that  the
tilt is constant in the target index and as result, it factors out of the log-sum-exp, and
the calibrated weight is \emph{exactly} an offset exponentially-tilted sieve estimator, whose calibration equations are the first-order conditions of a strictly convex $M$-estimation problem in a $J$-dimensional parameter where $J$ is the size of of sieve bases. Once
that is achieved, Section~\ref{sub:E-rate} obtains the $L^2(P)$ rate by an entirely
standard sieve $M$-estimation argument, with the regularisation parameters held
fixed throughout.
 
 
\subsection{The calibrated program and its dual structure}
\label{sub:E-kkt}
 Let \(i=1,\ldots,n\) index the entire randomized trial sample. Put
\begin{equation}
  a_i\coloneq \frac1n \quad (i=1,\ldots,n),
  \qquad
  \omega_j\coloneq \frac1m \quad (j=1,\ldots,m),
  \qquad
  C_{ij}\coloneq c(X_i,X_j^Q).
  \label{eq:E-masses}
\end{equation}
For a nonnegative array
\(\Pi=(\pi_{ij})_{1\le i\le n,\,1\le j\le m}\), write
\(\alpha_i(\Pi)\coloneq \sum_j\pi_{ij}\). The discrete generalised divergences are
those of \eqref{eq:genKL} with the counting reference measures $a\otimes\omega$
and $a$. Recall that \(b_J(x)=\{b_{J,1}(x),\ldots,b_{J,J}(x)\}^{\top}\) be a
vector of \(J\) basis functions with \(b_{J,1}\equiv1\), and let
\(H_J=\mathrm{span}\{b_{J,1},\ldots,b_{J,J}\}\).
 
\begin{definition}[Empirical Riesz-calibrated transport]
\label{def:emp}
The empirical calibrated plan $\widehat\Pi=(\widehat\pi_{ij})$ solves
\begin{equation}
  \min_{\Pi\ge0}\Bigl\{
    \sum_{i=1}^n\sum_{j=1}^m C_{ij}\pi_{ij}
    +\varepsilon\,\KL\bigl(\Pi\,\|\,a\otimes\omega\bigr)
    +\rho\,\KL\bigl(\alpha(\Pi)\,\|\,a\bigr)\Bigr\}
  \label{eq:E-program}
\end{equation}
subject to the target-marginal constraints
$\sum_{i=1}^n\pi_{ij}=\omega_j$ for $j=1,\dots,m$ and the \emph{Riesz
calibration constraints}
\begin{equation}
\sum_{i=1}^n\alpha_i(\Pi)\,b_J(X_i)=\sum_{j=1}^m\omega_j\,b_J(\corr{X_j^Q}).
  \label{eq:E-calib}
\end{equation}
The calibrated weight is $\rhat(X_i)\coloneq \alpha_i(\widehat\Pi)/a_i$ for
$i=1,\ldots,n$\corr{; its out-of-sample extension to arbitrary $x\in\calX$ is
the explicit formula} \eqref{eq:E-weight} below.
\end{definition}
 
It is worth noting two remarks on the formulation. First, because \(a_i=1/n\), the calibration constraint \eqref{eq:E-calib} is exactly the empirical Riesz equation \corr{(writing $r_\Pi(X_i)\coloneq \alpha_i(\Pi)/a_i$ for a generic feasible plan, and $\rhat$ at the optimum)}:
\begin{equation}
  \sum_{i=1}^n \alpha_i(\Pi)\,b_J(X_i)
  =
  \En\bigl[\rhat(X)b_J(X)\bigr].
  \label{eq:E-calib-E}
\end{equation}
Hence \eqref{eq:E-calib} reads
\[
\En\bigl[\rhat(X)b_J(X)\bigr]
=
\E_m\bigl[b_J(X^Q)\bigr],
\]
the sample analogue of \eqref{eq:calib-pop}. Second, calibration is enforced \emph{simultaneously} with the transport cost minimisation, and this is what produces the multiplicative structure below.
\corr{Note also that, since $b_{J,1}\equiv1$, the first calibration equation
duplicates the total-mass identity already implied by the target-marginal
constraints, so \eqref{eq:E-calib} carries $J-1$ effective constraints; this
is the one-dimensional multiplier indeterminacy resolved in
Proposition~\ref{prop:E-kkt}.}
 
\begin{proposition}[KKT structure of the calibrated program]
\label{prop:E-kkt}
Suppose \eqref{eq:E-program}--\eqref{eq:E-calib} admits a strictly interior solution, as in Assumption~\ref{ass:des-dual}. Then there exist multipliers $v\in\R^m$ for the target-marginal constraints and $\lambda\in\R^J$ for
\eqref{eq:E-calib} such that
\begin{equation}
  \widehat\pi_{ij}
  =a_i\,\omega_j\,
   \exp\Bigl\{\frac{v_j-C_{ij}-\lambda^\top b_J(X_i)}{\varepsilon}\Bigr\}\,
   \rhat_i^{\,-\rho/\varepsilon},
  \label{eq:E-plan}
\end{equation}
and consequently, with $\gamma=\varepsilon/(\varepsilon+\rho)$,
\begin{equation}
  \rhat(x)
  =\Bigl[\sum_{j=1}^m\omega_j
     \exp\Bigl\{\frac{v_j-c(x,X^Q_j)-\lambda^\top b_J(x)}{\varepsilon}\Bigr\}
   \Bigr]^{\gamma}
  =\exp\Bigl\{\gamma\Lhat(x)-\frac{\lambda^\top b_J(x)}{\varepsilon+\rho}\Bigr\},
  \label{eq:E-weight}
\end{equation}
where the \emph{empirical offset} is
\begin{equation}
  \Lhat(x)\coloneq \log\sum_{j=1}^m\omega_j
    \exp\Bigl\{\frac{v_j-c(x,X^Q_j)}{\varepsilon}\Bigr\}.
  \label{eq:offset-hat}
\end{equation}
\corr{Note that \eqref{eq:E-plan} determines the plan only implicitly, since
$\rhat_i$ appears on both sides; the closed form for the weight is
\eqref{eq:E-weight}. If, in addition, the empirical Gram matrix
$\widehat G_J\coloneq \En[b_J(X)b_J(X)^\top]$ is nonsingular, the multipliers are
unique once $v$ is normalised by $\Em(v)=0$.}
\end{proposition}
 
\begin{proof}
Introduce multipliers $v_j$ for the target-marginal constraints and
$\lambda\in\R^J$ for \eqref{eq:E-calib}, and write
$\alpha_i\coloneq \alpha_i(\Pi)=\sum_j\pi_{ij}$ inside the Lagrangian
\begin{align*}
  \mathcal L(\Pi,v,\lambda)
  &=\sum_{i,j}C_{ij}\pi_{ij}
   +\varepsilon\sum_{i,j}\Bigl\{\pi_{ij}\log\frac{\pi_{ij}}{a_i\omega_j}
      -\pi_{ij}+a_i\omega_j\Bigr\}
   +\rho\sum_i\Bigl\{\alpha_i\log\frac{\alpha_i}{a_i}-\alpha_i+a_i\Bigr\}\\
  &\qquad+\lambda^\top\Bigl(\sum_i\alpha_ib_J(X_i)-\sum_j\omega_jb_J(X^Q_j)\Bigr)
   -\sum_jv_j\Bigl(\sum_i\pi_{ij}-\omega_j\Bigr),
\end{align*}
all sums over \(i\) being over \(i=1,\ldots,n\). Since $\alpha_i=\sum_j\pi_{ij}$ we have
$\partial\alpha_i/\partial\pi_{ij}=1$, so setting
$\partial\mathcal L/\partial\pi_{ij}=0$ at an interior optimiser gives
\begin{equation}
  C_{ij}+\varepsilon\log\frac{\pi_{ij}}{a_i\omega_j}
  +\rho\log\frac{\alpha_i}{a_i}+\lambda^\top b_J(X_i)-v_j=0.
  \label{eq:E-foc}
\end{equation}
At the optimiser $\alpha_i/a_i=\rhat_i$, and solving \eqref{eq:E-foc} for
$\pi_{ij}$ yields \eqref{eq:E-plan}. Summing \eqref{eq:E-plan} over $j$ and using $\alpha_i=a_i\rhat_i$, we notice
\[
  a_i\rhat_i=a_i\rhat_i^{\,-\rho/\varepsilon}
   \sum_j\omega_j e^{\{v_j-C_{ij}-\lambda^\top b_J(X_i)\}/\varepsilon},
\]
so that $\rhat_i^{\,1+\rho/\varepsilon}$ equals the sum and finally, raising to the power $\gamma=(1+\rho/\varepsilon)^{-1}$ gives the first equality in
\eqref{eq:E-weight}. For the second, observe that
$\lambda^\top b_J(X_i)$ does not depend on $j$ and therefore factors out of the
sum:
\[
  \Bigl[e^{-\lambda^\top b_J(x)/\varepsilon}
        \sum_j\omega_je^{\{v_j-c(x,X_j^Q)\}/\varepsilon}\Bigr]^{\gamma}
  =e^{-\gamma\lambda^\top b_J(x)/\varepsilon}\,e^{\gamma\Lhat(x)},
\]
and $\gamma/\varepsilon=1/(\varepsilon+\rho)$.
 
{ For uniqueness, note first that the objective in \eqref{eq:E-program} is
strictly convex in $\Pi$, so the optimal plan $\widehat\Pi$ --- and with it
$\alpha$ and $\rhat$ --- is unique. Suppose $(v,\lambda)$ and $(v',\lambda')$
both satisfy \eqref{eq:E-foc} at $\widehat\Pi$, and set $\tilde v\coloneq v-v'$,
$\tilde\lambda\coloneq \lambda-\lambda'$. Subtracting the two systems gives
$\tilde\lambda^\top b_J(X_i)=\tilde v_j$ for all $i,j$; the left side is free
of $j$ and the right side free of $i$, so both equal a constant $c$. Thus
$\tilde v_j=c$ for all $j$ and, using $b_{J,1}\equiv1$,
$(\tilde\lambda-ce_1)^\top b_J(X_i)=0$ for all $i$, where $e_1$ is the first
coordinate vector. Nonsingularity of $\widehat G_J$ forces
$\tilde\lambda=ce_1$, and the normalisation $\Em(v)=0$ forces $c=0$, whence
$(v,\lambda)=(v',\lambda')$. (Without full rank, multiplier components along
directions vanishing on the design are undetermined, although the plan and the
weight remain unique.)}
\end{proof}
 
\begin{corollary}
\label{cor:E-tilt}
Fix the target dual vector $v$ and, for $\lambda\in\R^J$, let
$\rhat(x;\lambda,v)$ denote the right-hand side of \eqref{eq:E-weight} viewed
as a function of $(\lambda,v)$. Then
\begin{equation}
  \rhat(x;\lambda,v)
  =\rhat(x;0,v)\,
   \exp\Bigl\{-\frac{\lambda^\top b_J(x)}{\varepsilon+\rho}\Bigr\}.
  \label{eq:E-tiltratio}
\end{equation}
{ If $\norm{\lambda^\top b_J}{\infty}/(\varepsilon+\rho)\le C_0$ for some
finite constant $C_0$, then
\[
  e^{-C_0}\;\le\;\frac{\rhat(x;\lambda,v)}{\rhat(x;0,v)}\;\le\;e^{C_0}
  \qquad\text{for every }x\in\calX,
\]
so calibration modifies the raw transport weight by a multiplicative factor
bounded away from zero and infinity, uniformly in $x$.}
\end{corollary}

\begin{proof}
Identity \eqref{eq:E-tiltratio} is the second equality of
\eqref{eq:E-weight} read at $\lambda$ and at $0$ with the same $v$; note that
the factor $1/\varepsilon$ inside the bracket \corr{combines with} the outer power
$\gamma$ \corr{via $\gamma/\varepsilon=1/(\varepsilon+\rho)$}, which is why the exponent carries $\varepsilon+\rho$ and not
$\varepsilon$. \corr{The two-sided bound is immediate from the hypothesis.}
\end{proof}
 
\begin{remark}
\label{rem:E-mult}
 Firstly, \eqref{eq:E-tiltratio} is what makes calibration
compatible with the overlap adaptation of Remark~\ref{rem:D-shrink}: a small baseline weight
remains small after a uniformly bounded multiplicative tilt. Under a
squared-loss calibration the correction would instead be additive,
$\rhat=\rhat|_{\lambda=0}+\lambda^\top b_J$, which can assign positive mass to
units with zero uncalibrated mass and so destroy the geometry of the overlap
region. That the tilt is exponential is a consequence of the $\KL$ penalties in
\eqref{eq:E-program}, and is therefore not a matter of computational
convenience.

{ Second, the hypothesis of Corollary~\ref{cor:E-tilt} is a supremum-norm
bound on the tilt \emph{function}, not on its coefficient vector. The
distinction matters: the Cauchy--Schwarz route
$|\lambda^\top b_J(x)|\le B\sqrt J\,\norm{\lambda}{2}$ would require
$\norm{\lambda}{2}/(\varepsilon+\rho)=O(J^{-1/2})$, which fails in the sieve
regime\corr{: by the Gram bound of Assumption~\ref{ass:des-basis},
$\norm{\theta_J}{2}\ge\norm{\theta_J^\top b_J}{L^2(P)}/\sqrt{C_G}
\to\norm{\psi}{L^2(P)}/\sqrt{C_G}>0$ whenever the population tilt $\psi$ is
not identically zero, so the coefficient is of constant order}. What Theorem~\ref{thm:E-rate} delivers at the realised solution, with
probability tending to one, is precisely the supremum-norm hypothesis:
$\norm{\thetahat^\top b_J}{\infty}
 \le\norm{\theta_J^\top b_J}{\infty}+o_p(1)=O_p(1)$
by Step~0 of its proof together with \eqref{eq:E-supnorm}\corr{; thus $C_0$ may be taken $O_p(1)$, the two-sided bound holding with probability tending to one}.}

Third, the comparison in \eqref{eq:E-tiltratio} is made \emph{at fixed} $v$.
It is not a comparison between the calibrated solution and the separately
re-optimised uncalibrated Sinkhorn solution, whose target dual vector will in
general differ when the calibration constraints are removed. The statement
supplied by the $\KL$ geometry is the conditional one, and we do not claim more.
\end{remark}
 
The last item of this subsection records the boundedness of the target dual by invoking Assumption~\ref{ass:des-dual}. Its proof is a two-line consequence of \eqref{eq:E-foc}, and it is
the only place where \corr{the numerical behaviour of the transport solver is touched}.
 
\begin{lemma}[Boundedness of the target dual and the offset]
\label{lem:dual-bdd}
\corr{Suppose the empirical calibrated program admits a strictly interior solution whose plan and weight ratios are bounded away from zero and infinity, that is,} $0<c_h\le\widehat\pi_{ij}/(a_i\omega_j)\le C_h$ and
$0<c_r\le\rhat_i\le C_r$ for \(i=1,\ldots,n\), $j\le m$. Then, with the normalisation $\Em(v)=0$,
\[
  \norm{v}{\infty}=O(1),
\]
and consequently, if in addition the cost satisfies
$\sup_z\norm{c(\cdot,z)}{W^{s,\infty}(\calX)}\le C_{c,s}<\infty$, the empirical offset of \eqref{eq:offset-hat} obeys
$\norm{\Lhat}{W^{s,\infty}(\calX)}\le \corr{C_L'}$\corr{, a deterministic constant,} uniformly in $m$.
\end{lemma}
 
\begin{proof}
Write \eqref{eq:E-foc} as
$v_j-\lambda^\top b_J(X_i)
 =C_{ij}+\varepsilon\log\{\widehat\pi_{ij}/(a_i\omega_j)\}+\rho\log\rhat_i$.
The right-hand side is bounded in absolute value by a deterministic constant
$C_1$, since the cost is bounded on the compact $\calX$, the ratios are bounded
away from zero and infinity \corr{by hypothesis}, and $\varepsilon,\rho$ are bounded
by Assumption~\ref{ass:design}. Hence for any $j,j'$ and any \corr{\(i=1,\ldots,n\),}
\[
  |v_j-v_{j'}|
  \le|v_j-\lambda^\top b_J(X_i)|+|v_{j'}-\lambda^\top b_J(X_i)|\le2C_1,
\]
so all coordinates of $v$ lie within $2C_1$ of one another; the normalisation
$\Em(v)=0$ then forces $\norm{v}{\infty}\le2C_1$.
 
For the offset, apply Lemma~\ref{lem:A-lse} \corr{(with $\nu$ the empirical distribution
of the target covariates)}, $k=s$ and
$f_j(\cdot)=\{v_j-c(\cdot,X^Q_j)\}/\varepsilon$. The hypothesis of that lemma
holds because $\norm{v}{\infty}=O(1)$, the cost is uniformly bounded in
$W^{s,\infty}$, and $\varepsilon\ge\varepsilon_0>0$; the conclusion is the
stated bound, with a constant that does not depend on $m$.
\end{proof}
 
\subsection{The calibrated weight as an exponentially tilted sieve}
\label{sub:E-tilt}
 
Everything in this section now follows from a change of variable. Set
\begin{equation}
  \theta\coloneq -\frac{\lambda}{\varepsilon+\rho}\in\R^J,
  \label{eq:E-theta}
\end{equation}
so that, by the second equality in \eqref{eq:E-weight},
\begin{equation}
  \rhat(x)=\exp\bigl\{\gamma\Lhat(x)+\theta^\top b_J(x)\bigr\}.
  \label{eq:E-tiltform}
\end{equation}
The calibrated transport weight is thus an exponential family in $\theta$ with
a data-dependent \emph{offset} $\gamma\Lhat$ carrying all the transport geometry, and a \emph{tilt} $\theta^\top b_J$ carrying all the calibration. The population counterpart replaces $\Lhat$ by
\begin{equation}
  L(x)\coloneq \log\int\exp\Bigl\{\frac{g(x^Q)-c(x,x^Q)}{\varepsilon}\Bigr\}\,q(x^Q)\,dx^Q,
  \label{eq:offset}
\end{equation}
$g$ being the population target dual potential; Assumption~\ref{ass:des-dual} is
precisely the requirement that $\Lhat$ converge to this fixed $L$ at rate
$a_{L,N}$.

{ \begin{corollary}[Sobolev bound for the population offset]
\label{cor:E-Lpop}
Under Assumptions~\ref{ass:design} and~\ref{ass:des-dual}, the population
offset satisfies $\norm{L}{W^{s,\infty}(\calX)}\le C_L$, with $C_L$ depending
only on $(s,d,C_g,C_{c,s},\varepsilon_0)$. Consequently, on the event of
Assumption~\ref{ass:des-dual} the difference $\Lhat-L$ lies in the fixed
Sobolev ball $\{\ell:\norm{\ell}{W^{s,\infty}(\calX)}\le C_L+C_L'\}$, where
$C_L'$ is the constant of Lemma~\ref{lem:dual-bdd}, while
$\norm{\Lhat-L}{\infty}=O_p(a_{L,N})$.
\end{corollary}

\begin{proof}
Apply Lemma~\ref{lem:A-lse} with $\nu=\QX$ and
$f_z(\cdot)=\{g(z)-c(\cdot,z)\}/\varepsilon$; its hypothesis holds with
$F=(C_g+C_{c,s})/\varepsilon_0$. The second assertion combines this with
Lemma~\ref{lem:dual-bdd} and Assumption~\ref{ass:des-dual}.
\end{proof}

A word on the logical status of the pair $(v,\theta)$ is in order, because the
offset and the tilt are \emph{not} variation-independent: in the calibrated
program they are the multipliers of one optimisation problem and are
determined jointly. The analysis below is accordingly a \emph{profiled}
analysis. On the event that the program has an interior solution with target
dual $v$, the realised tilt $\thetahat=-\lambda/(\varepsilon+\rho)$ is the
unique minimiser of the convex objective $M_n$ built from the realised offset
$\gamma\Lhat$ (Proposition~\ref{prop:E-Mest} below). No independence between
$\Lhat$ and the trial sample is assumed anywhere; the rate argument of
Section~\ref{sub:E-rate} replaces $\Lhat$ by the deterministic $L$ at the
outset, at a cost of $O_p(a_{L,N})$ \corr{(the offset-convergence clause of Assumption~\ref{ass:des-dual})}, with the Sobolev envelope supplied by Corollary~\ref{cor:E-Lpop}, and is otherwise a statement about fixed
functions.}

The rate \(a_{L,N}\) in Assumption~\ref{ass:des-dual} can be related to more primitive quantities by separating two distinct sources of error in the empirical transport offset. Recall that
\(\Lhat(x)=
\log\left[
m^{-1}\sum_{j=1}^m
\exp\{[v_j-c(x,X_j^Q)]/\varepsilon\}
\right]\),
where \(v_j\) is the empirical target-side dual potential associated with the target-marginal constraint for \(X_j^Q\), whereas
\(L(x)=
\log\EQ[
\exp\{[g(X^Q)-c(x,X^Q)]/\varepsilon\}
]\),
where \(g\) is the corresponding population target dual potential.

To connect these two quantities, introduce the intermediate offset
\[\widetilde L_m(x)\coloneq 
\log\left[
m^{-1}\sum_{j=1}^m
\exp\{[g(X_j^Q)-c(x,X_j^Q)]/\varepsilon\}
\right].\] Thus \(\widetilde L_m\) uses the population dual potential \(g\), as does
\(L\), but evaluates the defining integral using the empirical target distribution, as does \(\Lhat\). It therefore provides a bridge between the empirical offset \(\Lhat\) and the population offset \(L\).
Notice by the triangle inequality,
\[\norm{\Lhat-L}{\infty}
\le
\norm{\Lhat-\widetilde L_m}{\infty}
+
\norm{\widetilde L_m-L}{\infty}.\]
The first term measures the error from replacing the population dual values \(g(X_j^Q)\) by their empirical counterparts \(v_j\). The second term measures the sampling error from replacing integration with respect to the target
covariate distribution \(Q_X\) by averaging over the observed target covariates. We control these two terms separately.

\begin{lemma}[Perturbation of the transport offset]
\label{lem:E-offset-perturb}
Define
\(\Delta_{g,N}\coloneq 
\max_{1\le j\le m}|v_j-g(X_j^Q)|\). Then, for fixed \(\varepsilon>0\),
\(\norm{\Lhat-\widetilde L_m}{\infty}
\le \Delta_{g,N}/\varepsilon\).
\end{lemma}

\begin{proof}
For each \(x\), let
\[S_m(x)\coloneq 
m^{-1}\sum_{j=1}^m
\exp\{[v_j-c(x,X_j^Q)]/\varepsilon\},\qquad \widetilde S_m(x)\coloneq 
m^{-1}\sum_{j=1}^m
\exp\{[g(X_j^Q)-c(x,X_j^Q)]/\varepsilon\}.\]
By definition of \(\Delta_{g,N}\),
\(-\Delta_{g,N}\le v_j-g(X_j^Q)\le\Delta_{g,N}\)
for every \(j\). Hence
\(e^{-\Delta_{g,N}/\varepsilon}\widetilde S_m(x)
\le S_m(x)\le
e^{\Delta_{g,N}/\varepsilon}\widetilde S_m(x)\).
Taking logarithms gives
\(-\Delta_{g,N}/\varepsilon
\le \Lhat(x)-\widetilde L_m(x)
\le \Delta_{g,N}/\varepsilon\).
Taking the supremum over \(x\in\calX\) proves the result.
\end{proof}

We next control the second term \(\norm{\widetilde L_m-L}{\infty}\). Define \(\phi_x(z)\coloneq \exp\{[g(z)-c(x,z)]/\varepsilon\}\) and
\(\mathcal F_\varepsilon\coloneq \{\phi_x:x\in\calX\}\).
By construction, \(\widetilde L_m(x)=\log\E_m[\phi_x(X^Q)]\),
whereas \(L(x)=\log\EQ[\phi_x(X^Q)]\).
Thus the only difference between \(\widetilde L_m\) and \(L\) is the replacement of the population expectation \(\EQ\) by the empirical target average \(\Em\).

Under the boundedness conditions on \(g\), \(c\), and
\(\varepsilon\ge\varepsilon_0>0\), there exist constants
\(0<\underline C<\overline C<\infty\) such that
\(\underline C\le\EQ[\phi_x(X^Q)]\le\overline C\)
uniformly over \(x\in\calX\). If
\(\sup_{x\in\calX}|(\E_m-\EQ)\phi_x|=o_p(1)\),
then \(\Em[\phi_x(X^Q)]\) is also bounded away from zero uniformly in \(x\) with probability tending to one. Since the logarithm is Lipschitz on compact subsets of \((0,\infty)\), it follows that \(\norm{\widetilde L_m-L}{\infty} \le C\sup_{x\in\calX}|(\E_m-\EQ)\phi_x|\)
with probability tending to one, for some finite constant \(C\).

Combining the preceding two bounds yields a primitive rate for the complete offset error.

\begin{proposition}[Primitive rate for the transport offset]
\label{prop:E-offset-rate}
Suppose that, \corr{under the normalisations $\Em(v)=0$ and $\EQ[g(X^Q)]=0$ of Assumption~\ref{ass:des-dual}},
\(\Delta_{g,N}
=
\max_{1\le j\le m}|v_j-g(X_j^Q)|
=
O_p(b_{g,N})\)
for some sequence \(b_{g,N}\to0\), and suppose that
\(\sup_{x\in\calX}|(\E_m-\EQ)\phi_x|
=
O_p(b_{Q,m})\)
for some sequence \(b_{Q,m}\to0\). Then
\(\norm{\Lhat-L}{\infty}
=
O_p\{b_{g,N}/\varepsilon+b_{Q,m}\}\).
Consequently, \corr{the offset-convergence clause of} Assumption~\ref{ass:des-dual} holds with
\(a_{L,N}\coloneq b_{g,N}/\varepsilon+b_{Q,m}\).
If \(\varepsilon\ge\varepsilon_0>0\) is fixed, one may equivalently take
\(a_{L,N}\coloneq b_{g,N}+b_{Q,m}\).
\end{proposition}

\begin{proof}
The decomposition above gives
\(\norm{\Lhat-L}{\infty}
\le
\norm{\Lhat-\widetilde L_m}{\infty}
+
\norm{\widetilde L_m-L}{\infty}\).
By Lemma~\ref{lem:E-offset-perturb}, the first term is bounded by \(\Delta_{g,N}/\varepsilon\), while the target empirical-process argument bounds the second by \(C\sup_{x\in\calX}|(\E_m-\EQ)\phi_x|\). The assumed rates therefore imply
\(\norm{\Lhat-L}{\infty}
=
O_p\{b_{g,N}/\varepsilon+b_{Q,m}\}\).
For fixed \(\varepsilon\ge\varepsilon_0>0\), the factor
\(1/\varepsilon\) is absorbed into the multiplicative constant, so \(a_{L,N}\coloneq b_{g,N}+b_{Q,m}\) is an admissible rate sequence.
\end{proof}

A useful special case arises when the target empirical-process component is root-\(m\). If the class \(\mathcal F_\varepsilon\) satisfies conditions ensuring \(\sup_{x\in\calX}|(\E_m-\EQ)\phi_x|=O_p(m^{-1/2})\),
then, for fixed \(\varepsilon>0\), Proposition~\ref{prop:E-offset-rate} permits the choice \(a_{L,N}\coloneq b_{g,N}+m^{-1/2}\).
For example, such a bound follows if \(\mathcal F_\varepsilon\) is \(Q\)-Donsker with a square-integrable envelope. More generally, the Donsker condition may be replaced by any entropy or concentration condition yielding the same uniform empirical-process rate.

Although \(v_j\) is indexed by the target observations, the quantity \(b_{g,N}\) is generally a two-sample rate. The empirical dual vector is obtained from the transport problem involving both the full randomized trial covariate sample and the target sample. Thus \(b_{g,N}\) may depend on both \(n\) and \(m\). With \(N=n\wedge m\) and \(n\asymp m\), the notation \(b_{g,N}\) emphasizes this dependence.

\begin{remark}
\label{rem:E-offset-rootN}
If, in addition, \(\max_{1\le j\le m}|v_j-g(X_j^Q)|=O_p(N^{-1/2})\), \(\sup_{x\in\calX}|(\E_m-\EQ)\phi_x|=O_p(m^{-1/2})\), and \(n\asymp m\asymp N\), then, for fixed \(\varepsilon>0\), \(\norm{\Lhat-L}{\infty}=O_p(N^{-1/2})\), so one may take \(a_{L,N}\coloneq N^{-1/2}\).
This stronger conclusion requires a separate root-\(N\) stability result for the empirical dual potential and does not follow from Lemma~\ref{lem:dual-bdd}, which establishes boundedness rather than a convergence rate. \corr{Under the normalisation pairing $\Em(v)=0$,
$\EQ[g(X^Q)]=0$, one has $b_{g,N}\gtrsim m^{-1/2}$ in general, so
$a_{L,N}\asymp N^{-1/2}$ is the best rate attainable from this
decomposition.}
\end{remark}

The decomposition also clarifies the role of the regularization parameters. For fixed positive \(\varepsilon\) and \(\rho\), these parameters affect the population offset, the empirical dual problem, and the constants in the preceding bounds, but they do not generate a separate deterministic bias term
in the calibrated weight rate. If either \(\varepsilon\) or \(\rho\) is allowed to vary with \(N\), their effect on \(b_{g,N}\) and on the constants in the offset bound must be tracked explicitly.

The preceding offset bound provides the link between the transport problem and the finite-dimensional calibration problem considered next. Recall that the calibrated weight has the representation \(\rhat(x)=\exp\{\gamma\Lhat(x)+\theta^\top b_J(x)\}\). Thus, once the transport problem has produced \(\Lhat\), calibration acts only through the \(J\)-dimensional tilt parameter \(\theta\), with \(\gamma\Lhat\) playing the role of an estimated offset. At the population level the corresponding representation replaces \(\Lhat\) by \(L\).
Consequently, the discrepancy between the empirical and population calibration problems contains, in addition to ordinary trial and target sampling variation, the perturbation induced by replacing \(L\) with \(\Lhat\). Proposition~\ref{prop:E-offset-rate} controls precisely this
additional term through \(\norm{\Lhat-L}{\infty}=O_p(a_{L,N})\).
This observation allows the calibration equations to be studied as a finite-dimensional convex \(M\)-estimation problem with an estimated offset, which we formalize next.

\begin{proposition}[Calibration as convex $M$-estimation]
\label{prop:E-Mest}
Define
\begin{equation}\label{eq:E-Mn}
\begin{aligned}
  M_n(\theta)
  &\coloneq \En\Bigl[\,e^{\gamma\Lhat(X)+\theta^\top b_J(X)}\Bigr]
    -\theta^\top\E_m\bigl[b_J(X^Q)\bigr],
  \\
  M(\theta)
  &\coloneq \EP\Bigl[e^{\gamma L(X)+\theta^\top b_J(X)}\Bigr]
    -\theta^\top\EQ\bigl[b_J(X^Q)\bigr].
\end{aligned}
\end{equation}
Then:
\begin{enumerate}[label=\textup{(\roman*)},leftmargin=2.6em,itemsep=2pt,topsep=3pt]
\item $M_n$ and $M$ are convex\corr{, with Hessians \eqref{eq:E-hessian}
below; $M$ is strictly convex under Assumption~\ref{ass:des-basis}, and $M_n$
is strictly convex if and only if
$\widehat G_J=\En[b_J(X)b_J(X)^\top]$ is nonsingular}.
\item The stationarity condition $\nabla M_n(\theta)=0$ is exactly the empirical calibration constraint \eqref{eq:E-calib-E}, and $\nabla M(\theta)=0$ is exactly the population calibration equation \eqref{eq:calib-pop}, for the weight \eqref{eq:E-tiltform}\corr{ and its population counterpart}.
\item \corr{Suppose $\widehat G_J$ is nonsingular. Then $M_n$ possesses a
minimiser if and only if $\E_m[b_J(X^Q)]$ lies in the relative interior of the
convex hull $\calC_J\coloneq \mathrm{conv}\{b_J(X_i):i=1,\ldots,n\}$; in that case
the minimiser $\thetahat$ is unique, and $\rhat$ of Definition~\ref{def:emp}
is given by \eqref{eq:E-tiltform} with $\theta=\thetahat$.}
\end{enumerate}
\end{proposition}

\begin{proof}
(i) Both maps are of the form
$\theta\mapsto\E[w\,e^{\theta^\top b_J}]-\theta^\top\varkappa$ with $w\ge0$. The exponential term is convex in $\theta$ as a nonnegative mixture of convex functions, and the linear term is affine. The Hessians are
\begin{equation}
  \nabla^2M_n(\theta)
  =\En\Bigl[\,e^{\gamma\Lhat+\theta^\top b_J}\,b_J b_J^\top\Bigr],
  \qquad
  \nabla^2M(\theta)
  =\EP\Bigl[e^{\gamma L+\theta^\top b_J}\,b_J b_J^\top\Bigr],
  \label{eq:E-hessian}
\end{equation}
both positive semidefinite. \corr{For fixed $\theta$ the weight
$e^{\gamma L+\theta^\top b_J}$ is bounded away from zero on the compact
$\calX$, so $\nabla^2M(\theta)\succeq c(\theta)\,G_J\succ0$ under
Assumption~\ref{ass:des-basis}; likewise
$\nabla^2M_n(\theta)\succeq c_n(\theta)\,\widehat G_J$ with $c_n(\theta)>0$,
which is positive definite exactly when $\widehat G_J$ is.}

(ii) Differentiating \eqref{eq:E-Mn},
\(\nabla M_n(\theta)
=
\En\bigl[e^{\gamma\Lhat+\theta^\top b_J}b_J\bigr]
-
\Em[b_J],\)
which vanishes exactly when
\(
\En[\rhat\,b_J]=\Em[b_J],
\)
with \(\rhat\) as in \eqref{eq:E-tiltform}. By
\eqref{eq:E-calib-E}, this is precisely \eqref{eq:E-calib}.
For the population version, differentiating \(M\) gives
\(
\nabla M(\theta)
=
\EP\bigl[e^{\gamma L+\theta^\top b_J}b_J\bigr]
-
\EQ[b_J],
\)
so \(\nabla M(\theta)=0\) is exactly the population calibration equation
\eqref{eq:calib-pop}.

{ (iii) Write $z\coloneq \E_m[b_J(X^Q)]$, and note that $z_1=1=b_{J,1}(X_i)$ for every
$i$, so $z$ and the design points all lie on the affine hyperplane
$\{u\in\R^J:u_1=1\}$; the hull $\calC_J$ has empty interior in $\R^J$, and the
relevant notion is its relative interior.

The recession function of the convex $M_n$ in the direction $u\ne0$,
\[
  M_n^\infty(u)\coloneq \lim_{t\to\infty}\frac{M_n(\theta+tu)-M_n(\theta)}{t}
  =\begin{cases}
     +\infty,&\text{if }u^\top b_J(X_i)>0\text{ for some }i,\\[2pt]
     -u^\top z,&\text{if }u^\top b_J(X_i)\le0\text{ for all }i,
   \end{cases}
\]
does not depend on $\theta$: the exponential term explodes in the first case
and is bounded in the second. If $M_n^\infty(u)>0$ for every $u\ne0$, then
$M_n$ is coercive, its sublevel sets are compact, a minimiser exists, and it
is unique by the strict convexity of part (i). Conversely, suppose
$M_n^\infty(u)\le0$ for some $u\ne0$, so that $u^\top b_J(X_i)\le0$ for all
$i$ and $u^\top z\ge0$. If $u^\top z>0$ then $M_n(\theta+tu)\to-\infty$ and no
minimiser exists. If $u^\top z=0$ then, since $\widehat G_J\succ0$ excludes
$u^\top b_J(X_i)=0$ for all $i$, at least one design point has
$u^\top b_J(X_i)<0$; hence $t\mapsto M_n(\theta+tu)$ is strictly decreasing
for \emph{every} $\theta$, and again no minimiser exists. A minimiser
therefore exists if and only if
\begin{equation}
  u^\top b_J(X_i)\le0\ \ \text{for all }i
  \quad\Longrightarrow\quad
  u^\top z<0,
  \qquad\text{for every }u\ne0.
  \label{eq:E-coercive}
\end{equation}

It remains to identify \eqref{eq:E-coercive} with
$z\in\mathrm{relint}\,\calC_J$. Suppose first that \eqref{eq:E-coercive}
holds. Then $z\in\calC_J$: otherwise, $\calC_J$ being closed, strict
separation provides $u$ with
$\max_{x\in\calC_J}u^\top x<u^\top z$, and the translated vector
$\tilde u\coloneq u-(u^\top z)e_1$ satisfies
$\tilde u^\top b_J(X_i)=u^\top b_J(X_i)-u^\top z<0$ for all $i$ (so
$\tilde u\ne0$) and $\tilde u^\top z=0$, contradicting
\eqref{eq:E-coercive} --- both computations using that all first coordinates
equal one. Similarly, if $z$ were a relative boundary point of $\calC_J$, a
supporting hyperplane proper on $\calC_J$ would provide $u$ with
$u^\top x\le u^\top z$ for all $x\in\calC_J$ and $u^\top b_J(X_i)<u^\top z$
for at least one $i$; the same translation gives $\tilde u\ne0$ with
$\tilde u^\top b_J(X_i)\le0$ for all $i$ and $\tilde u^\top z=0$, again a
contradiction. Hence $z\in\mathrm{relint}\,\calC_J$.

Conversely, suppose $z\in\mathrm{relint}\,\calC_J$ and let $u\ne0$ satisfy
$u^\top b_J(X_i)\le0$ for all $i$. Then $u^\top z\le0$, $z$ being a convex
combination of the design points. If $u^\top z=0$, the linear functional
$u^\top(\cdot)$ attains its maximum over $\calC_J$ at the relative-interior
point $z$ and is therefore constant on $\calC_J$, forcing
$u^\top b_J(X_i)=0$ for all $i$ and contradicting $\widehat G_J\succ0$. Hence
$u^\top z<0$, which is \eqref{eq:E-coercive}. Finally, the identification of
$\rhat$ with \eqref{eq:E-tiltform} at $\theta=\thetahat$ is
Proposition~\ref{prop:E-kkt} combined with \eqref{eq:E-theta}.}
\end{proof}
 
\begin{remark}[Why the tilt parametrisation]
\label{rem:why-tilt}
It is worth being explicit about the route not taken, since it is the natural first attempt and it fails. One would like to argue that $\rhat$ lies, with probability tending to one, in a fixed $W^{s,2}$ ball, and then bound $\norm{(I-\Pi_J)(\rhat-\rdag)}{L^2(P)}\le CJ^{-s/d}\norm{\rhat-\rdag}{W^{s,2}}$
for the $L^2(P)$ projection $\Pi_J$ onto $H_J$. By \eqref{eq:E-weight} this would require a bound of the form
$\sup_{\norm{\lambda}{2}\le1}\norm{\lambda^\top b_J}{W^{s,\infty}}\le B_s$ with $B_s$ independent of $J$.
 
No basis satisfies that bound together with Assumption~\ref{ass:des-basis}. Indeed, write $\mathcal B_J\coloneq \{\lambda^\top b_J:\norm{\lambda}{2}\le1\}$. The Gram condition gives $\norm{\lambda^\top b_J-\lambda'^\top b_J}{L^2(P)}^2
 =(\lambda-\lambda')^\top G_J(\lambda-\lambda')\ge c_G\norm{\lambda-\lambda'}{2}^2$,
so a $\delta$-packing of the unit sphere in $\R^J$ maps to a
$\sqrt{c_G}\,\delta$-separated subset of $\mathcal B_J$; the
\corr{packing number of $\mathcal B_J$ in $L^2(P)$ at scale $\delta$} is therefore at least $(\sqrt{c_G}/\corr{\delta})^{J-1}$, which is unbounded in $J$ at any fixed \corr{$\delta<\sqrt{c_G}$}. A uniform $W^{s,\infty}$ bound, on the other hand, would place every $\mathcal B_J$ inside one fixed Sobolev ball, which on a bounded Lipschitz domain is totally bounded in $L^2(P)$ with packing numbers \emph{independent} of $J$. The two requirements are contradictory. Concretely, for the trigonometric basis of Assumption~\ref{ass:des-basis}, a coefficient vector supported on the top frequency has $\norm{\lambda^\top b_J}{W^{s,\infty}}\asymp J^{s/d}$, and carrying that factor through the projection bound gives $J^{-s/d}\cdot J^{s/d}=O(1)$.  
 
The tilt parametrisation sidesteps this entirely. It never asks the tilt to be smooth. Smoothness is required only of the \emph{offset} $\gamma\Lhat$, which by Lemma~\ref{lem:dual-bdd} is uniformly bounded in $W^{s,\infty}$ because it is
a log-sum-exp of uniformly smooth costs; and approximation is required only of the fixed function $\log\rdag-\gamma L$, which is Assumption~\ref{ass:des-approx}. What replaces the projection argument is the $M$-estimation argument of Theorem~\ref{thm:E-rate}, which controls the $J$-dimensional parameter \corr{$\thetahat$} directly.
\end{remark}
 
\subsection{Rate of convergence for the calibrated weight}
\label{sub:E-rate}
 
Because the tilt acts on the log scale, the rate theory requires the density
ratio to be bounded away from zero as well as from above. We record this as a
separate hypothesis, since it is genuinely stronger than the overlap condition used for identification and is needed nowhere else in the supplement.
 
\begin{assumption}[Two-sided overlap]
\label{ass:E-twosided}
There are constants $0<\underline r\le\overline r<\infty$ with
$\underline r\le\rdag(x)\le\overline r$ for $P$-almost every $x$.
\end{assumption}
 
Assumption~\ref{ass:E-twosided} is used only through the boundedness of
$\log\rdag$, and hence of the population tilt $\log\rdag-\gamma L$.
 
\begin{theorem}[Calibration-driven $L^2$ rate]
\label{thm:E-rate}
Suppose Assumptions~\ref{ass:id}, \ref{ass:design} and~\ref{ass:E-twosided} hold, and \corr{let
$J=J_N\to\infty$ satisfy the growth condition
\begin{equation}
  \sqrt{J_N}\,\delta_N\to0,
  \qquad\text{where}\quad
  \delta_N\coloneq J_N^{-s/d}+\sqrt{\frac{J_N\log J_N}{N}}+a_{L,N}.
  \label{eq:E-growth}
\end{equation}}Then
\begin{equation}
  \norm{\rhat-\rdag}{L^2(P)}
  =O_p\Bigl(
      J^{-s/d}
      +\sqrt{\frac{J\log J}{N}}
      +a_{L,N}
    \Bigr)\corr{{}=O_p(\delta_N)}.
  \label{eq:E-rate}
\end{equation}
Moreover $\norm{\rhat}{\infty}=O_p(1)$, and $\rhat$ is bounded away from zero
with probability tending to one\corr{, and $\norm{\thetahat}{2}=O_p(1)$}.
\end{theorem}

{ Condition \eqref{eq:E-growth} strengthens the familiar sieve-growth
condition $J^2\log J/N\to0$ --- which is the requirement that $\sqrt J$ times
the middle term of $\delta_N$ tend to zero --- by two further requirements
that the proof genuinely uses: $\sqrt J\,J^{-s/d}\to0$, automatic under the
smoothness floor $s>d/2$ of Assumption~\ref{ass:des-approx}, and
$\sqrt J\,a_{L,N}\to0$, a restriction on the offset error relative to the
sieve dimension. Both enter where the $\ell_2$ control of the tilt coefficient
is converted into supremum-norm control of the tilt (Steps~3 and~4 below).}

\begin{proof}
 Throughout, $C$ denotes a constant that may change from line to line and depends only on the constants appearing in Assumptions~\ref{ass:id}, \ref{ass:design}
and~\ref{ass:E-twosided}. Write
\[
  \psi\coloneq \log\rdag-\gamma L,
  \qquad
  \corr{W_L(x;\theta)\coloneq \exp\bigl\{\gamma L(x)+\theta^\top b_J(x)\bigr\},}
  \qquad
  W(x;\theta)\coloneq \exp\bigl\{\gamma\Lhat(x)+\theta^\top b_J(x)\bigr\},
\]
so that $\rhat=W(\cdot;\thetahat)$ and $\rdag=\exp(\gamma L+\psi)$\corr{; the
function $W_L$ is deterministic, while $W$ carries the estimated offset}. Let $\theta_J$ be the approximating vector of Assumption~\ref{ass:des-approx}, so that
$\norm{\psi-\theta_J^\top b_J}{\infty}\le C_sJ^{-s/d}$. \corr{We work
throughout on the event of Assumption~\ref{ass:des-dual}, whose probability
tends to one.}

\emph{Step 0: \corr{bounded exponents and design concentration}.} By
Assumption~\ref{ass:E-twosided}, we have $\norm{\log\rdag}{\infty}\le C$. \corr{By
Corollary~\ref{cor:E-Lpop}, $\norm{L}{\infty}\le C_L$, and therefore
$\norm{\psi}{\infty}\le C$ and, by the approximation bound,
$\norm{\theta_J^\top b_J}{\infty}\le C$ uniformly in $J$; moreover
$\norm{\Lhat}{\infty}=O(1)$ by Lemma~\ref{lem:dual-bdd} and
$\norm{\Lhat-L}{\infty}=O_p(a_{L,N})$.} Consequently
\corr{$W_L(\cdot;\theta_J)$, $W(\cdot;\theta_J)$} and $\rdag$ take values in a fixed compact sub-interval of $(0,\infty)$, on which the exponential and logarithm are Lipschitz. We use this repeatedly in the form
\begin{equation}
  C^{-1}|a-b|\le|e^a-e^b|\le C|a-b|
  \qquad\text{for }a,b\text{ in that interval.}
  \label{eq:E-lip}
\end{equation}
We furthermore notice that since
$c_G\norm{\theta_J}{2}^2\le\theta_J^\top G_J\theta_J
 =\norm{\theta_J^\top b_J}{L^2(P)}^2\le\norm{\theta_J^\top b_J}{\infty}^2\le C$
by Assumption~\ref{ass:des-basis}, the approximating coefficient vectors are bounded,
\begin{equation}
  \sup_J\norm{\theta_J}{2}<\infty .
  \label{eq:E-thetaJ-bdd}
\end{equation}
\corr{Finally, the empirical Gram matrix
$\widehat G_J\coloneq \En\bigl[b_J(X)b_J(X)^\top\bigr]$ is an average of $n$
independent positive semidefinite matrices with $\EP[\widehat G_J]=G_J$ and
operator norm bounded by $B^2J$; the matrix Bernstein inequality gives
\begin{equation}
  \norm{\widehat G_J-G_J}{\mathrm{op}}
  =O_p\Bigl(\sqrt{\frac{J\log J}{n}}+\frac{J\log J}{n}\Bigr)=o_p(1)
  \label{eq:E-Ghat}
\end{equation}
under \eqref{eq:E-growth}, so that, with probability tending to one,
$c_G/2\le\lambda_{\min}(\widehat G_J)\le\lambda_{\max}(\widehat G_J)\le2C_G$.}

\emph{Step 1: the score at $\theta_J$ is small.} { Using the Riesz identity
$\EP[\rdag b_J]=\EQ[b_J]$ (Lemma~\ref{lem:A-riesz}) to insert and remove
$\rdag$, decompose
\begin{equation}\label{eq:E-scoredecomp}
\begin{aligned}
  \nabla M_n(\theta_J)
  &=\underbrace{(\En-\EP)\bigl\{W_L(X;\theta_J)b_J(X)\bigr\}}_{\mathrm{(Ia)}}
  +\underbrace{\En\bigl[\{W(X;\theta_J)-W_L(X;\theta_J)\}b_J(X)\bigr]}_{\mathrm{(Ib)}}
  \\ &\qquad
  +\underbrace{\EP\bigl[\{W_L(X;\theta_J)-\rdag(X)\}b_J(X)\bigr]}_{\mathrm{(II)}}
  +\underbrace{(\EQ-\E_m)\bigl\{b_J(X^Q)\bigr\}}_{\mathrm{(III)}}.
  \end{aligned}
\end{equation}
The point of the decomposition is that the empirical process in (Ia) is
indexed by the \emph{deterministic} function $W_L(\cdot;\theta_J)$, while the
estimated offset appears only in (Ib), where it is controlled in supremum
norm; no conditioning on the dual vector is used anywhere.

For (Ia), $W_L(\cdot;\theta_J)$ is a fixed function, bounded by Step 0, so
each coordinate of $W_L(X;\theta_J)b_J(X)$ is an i.i.d.\ average of centred
variables bounded by $CB$; Hoeffding's inequality and a union bound over the
$J$ coordinates give, since $\norm{u}{2}\le\sqrt J\norm{u}{\infty}$,
$\norm{\mathrm{(Ia)}}{2}=O_p(\sqrt{J\log J/n})$.

For (Ib), by \eqref{eq:E-lip} and Step 0,
$\norm{W(\cdot;\theta_J)-W_L(\cdot;\theta_J)}{\infty}
 \le C\gamma\norm{\Lhat-L}{\infty}=O_p(a_{L,N})$; hence, for any unit vector
$u\in\R^J$, by Cauchy--Schwarz and \eqref{eq:E-Ghat},
\[
  |u^\top\mathrm{(Ib)}|
  \le\norm{W(\cdot;\theta_J)-W_L(\cdot;\theta_J)}{\infty}\,\En|u^\top b_J(X)|
  \le O_p(a_{L,N})\,\bigl\{u^\top\widehat G_Ju\bigr\}^{1/2}
  =O_p(a_{L,N}),
\]
so $\norm{\mathrm{(Ib)}}{2}=O_p(a_{L,N})$.

For (II), we avoid the crude coordinatewise bound, which would cost a spurious
factor $\sqrt J$. For any unit vector $u\in\R^J$,
\begin{align*}
  \bigl|\EP[\{W_L(X;\theta_J)-\rdag(X)\}u^\top b_J(X)]\bigr|
  &\le\norm{W_L(\cdot;\theta_J)-\rdag}{L^2(P)}\,\norm{u^\top b_J}{L^2(P)}
  \\ &\le\sqrt{C_G}\,\norm{W_L(\cdot;\theta_J)-\rdag}{L^2(P)},
\end{align*}
since $\norm{u^\top b_J}{L^2(P)}^2=u^\top G_Ju\le C_G$ by
Assumption~\ref{ass:des-basis}, and by \eqref{eq:E-lip},
$\norm{W_L(\cdot;\theta_J)-\rdag}{\infty}
 \le C\norm{\theta_J^\top b_J-\psi}{\infty}\le CJ^{-s/d}$.
Hence $\norm{\mathrm{(II)}}{2}\le CJ^{-s/d}$.

For (III), each coordinate $b_{J,j}(X^Q_\ell)-\EQ[b_{J,j}(X^Q)]$ is centred
and bounded by $2B$, so the same Hoeffding-plus-union-bound argument gives
$\norm{\mathrm{(III)}}{2}=O_p(\sqrt{J\log J/m})$.

Collecting the four bounds and using $n\asymp m\asymp N$,
\begin{equation}
  \norm{\nabla M_n(\theta_J)}{2}
  =O_p\Bigl(\sqrt{J\log J/N}+a_{L,N}+J^{-s/d}\Bigr)=O_p(\delta_N).
  \label{eq:E-score}
\end{equation}}

\emph{Step 2: local strong convexity.} By \eqref{eq:E-hessian}, for
$\theta$ in the Euclidean ball $\mathcal B\coloneq \{\norm{\theta-\theta_J}{2}\le
c_0/(B\sqrt J)\}$, where $c_0>0$ is a fixed constant, we have
$\norm{(\theta-\theta_J)^\top b_J}{\infty}\le c_0$, so that
$e^{\gamma\Lhat+\theta^\top b_J}\ge e^{-c_0}\,W(\cdot;\theta_J)\ge
\underline w>0$ on $\mathcal B$\corr{, with probability tending to one,} by Step 0. Hence
\[
  \nabla^2M_n(\theta)\succeq\underline w\,\widehat G_J,
  \qquad
  \widehat G_J\coloneq \En\bigl[b_J(X)b_J(X)^\top\bigr]
  \qquad\text{for }\theta\in\mathcal B .
\]
\corr{By \eqref{eq:E-Ghat}}, with probability tending to one,
$\lambda_{\min}(\nabla^2M_n)\ge\underline w\,c_G/2=:c_*>0$ on $\mathcal B$.

\emph{Step 3: from score to parameter.} \corr{On the event of
Assumption~\ref{ass:des-dual} the program has an interior solution, whose
multiplier supplies a stationary point of the convex $M_n$
(Proposition~\ref{prop:E-kkt} with Proposition~\ref{prop:E-Mest}(ii)); since
$\widehat G_J\succ0$ with probability tending to one by \eqref{eq:E-Ghat},
strict convexity (Proposition~\ref{prop:E-Mest}(i)) makes this stationary
point the unique minimiser $\thetahat$.} Write
\[
  s_N\coloneq \norm{\nabla M_n(\theta_J)}{2},
  \qquad
  \varrho_N\coloneq \corr{\frac{4s_N}{c_*}},
\]
so that $s_N=O_p(\delta_N)$ by \eqref{eq:E-score}. \corr{If $s_N=0$ then
$\thetahat=\theta_J$ by strict convexity and there is nothing to prove, so
suppose $s_N>0$ and} work on the event of Step~2. For a unit vector $u\in\R^J$ put $\theta=\theta_J+\varrho_Nu$; a second-order
Taylor expansion with the integral form of the remainder gives\corr{, provided
the segment stays in $\mathcal B$,}
\[
  M_n(\theta)-M_n(\theta_J)
  \ge \varrho_N\,u^\top\nabla M_n(\theta_J)+\tfrac12c_*\varrho_N^2
  \ge -\varrho_Ns_N+\tfrac12c_*\varrho_N^2\corr{= -\frac{4s^2_N}{c_*}+\frac{8s^2_N}{c_*}
  =\frac{4s^2_N}{c_*}>0 .}
\]
\corr{The segment condition $\varrho_N\le c_0/(B\sqrt J)$ holds with
probability tending to one, because $\varrho_N=O_p(\delta_N)$ and
$\sqrt J\,\delta_N\to0$ by the growth condition \eqref{eq:E-growth}.} Since $M_n$ is convex, its value \corr{everywhere} on the sphere of radius $\varrho_N$ \corr{strictly} exceeding its value at the centre forces the minimiser to lie strictly inside that sphere, so
\begin{equation}
  \norm{\thetahat-\theta_J}{2}\le\varrho_N=O_p(\delta_N).
  \label{eq:E-theta-rate}
\end{equation}

\emph{Step 4: from parameter to function.} Write
\[
  \log\rhat-\log\rdag
  =\gamma(\Lhat-L)
   +(\thetahat-\theta_J)^\top b_J
   -\bigl(\psi-\theta_J^\top b_J\bigr).
\]
The three terms are controlled respectively by
$\gamma\norm{\Lhat-L}{\infty}=O_p(a_{L,N})$, by
\[
  \norm{(\thetahat-\theta_J)^\top b_J}{L^2(P)}
  =\bigl\{(\thetahat-\theta_J)^\top G_J(\thetahat-\theta_J)\bigr\}^{1/2}
  \le\sqrt{C_G}\,\norm{\thetahat-\theta_J}{2}=O_p(\delta_N),
\]
and by $C_sJ^{-s/d}$. Hence $\norm{\log\rhat-\log\rdag}{L^2(P)}=O_p(\delta_N)$. To pass from the log scale to the weight itself we need the exponent of $\rhat$ to remain in a fixed compact interval, and this is where the \corr{growth} condition is used\corr{ a second time}. By
\eqref{eq:E-theta-rate},
\begin{equation}
  \norm{(\thetahat-\theta_J)^\top b_J}{\infty}
  \le B\sqrt J\,\norm{\thetahat-\theta_J}{2}
  =O_p\bigl(\sqrt J\,\delta_N\bigr)=o_p(1)
  \label{eq:E-supnorm}
\end{equation}
\corr{by \eqref{eq:E-growth}}, so that by Step~0 the exponent
$\gamma\Lhat+\thetahat^\top b_J$ lies, with probability tending to one, in a fixed
compact interval. Applying \eqref{eq:E-lip} once more therefore gives
\eqref{eq:E-rate}.

\emph{Step 5: boundedness.} This is now immediate. Notice that  \eqref{eq:E-supnorm} and
Step~0 show that the exponent $\gamma\Lhat+\thetahat^\top b_J$ is bounded in
probability, uniformly in $x$, and exponentiating gives
$\norm{\rhat}{\infty}=O_p(1)$ together with a positive lower bound holding with
probability tending to one. By \eqref{eq:E-thetaJ-bdd} and
\eqref{eq:E-theta-rate} we also have $\norm{\thetahat}{2}=O_p(1)$, which is the
form in which the conclusion is used in Appendix~\ref{app:F}.
\end{proof}
 
\begin{remark}[Reading the three terms]
\label{rem:E-three}
The rate \eqref{eq:E-rate} has one term for each of the three ways the
construction can go wrong, and they are worth naming. The term $J^{-s/d}$ is
\emph{sieve approximation error}: the calibration space is too coarse to
represent the population tilt, and no amount of data helps. The term
$\sqrt{J\log J/N}$ is \emph{calibration estimation error}: the $J$ moment
equations are solved on finite samples, and the error grows with $J$. These two
pull in opposite directions and are balanced in Corollary~\ref{cor:E-J}. The
third term, $a_{L,N}$, is new relative to a pure sieve analysis \corr{and is
\emph{transport-offset estimation error}: it measures how well the empirical
offset of the calibrated program approximates its population counterpart}, and it is the only
channel through which the numerical transport solve affects the statistical
rate. It is worth noting what does \emph{not} appear: \corr{no term requires
$\varepsilon$ or $\rho$ to vanish}. The $O(\varepsilon)$ imbalance of
Corollary~\ref{cor:D-imbalance} has been removed by the constraints, not by undersmoothing, and
the regularisation parameters survive only inside the offset and the exponent
$\gamma$.
\end{remark}
 
\begin{corollary}[Rate-balancing sieve dimension]
\label{cor:E-J}
\corr{Let Assumptions~\ref{ass:id}, \ref{ass:design} and~\ref{ass:E-twosided} hold}\corr{, and suppose
\begin{equation}
  a_{L,N}=o\Bigl\{\bigl(\tfrac{\log N}{N}\bigr)^{d/\{2(2s+d)\}}\Bigr\},
  \label{eq:E-aL-J}
\end{equation}}the choice
\begin{equation}
  J_N\asymp\Bigl(\frac{N}{\log N}\Bigr)^{d/(2s+d)}
  \label{eq:E-Jopt}
\end{equation}
\corr{satisfies the growth condition \eqref{eq:E-growth} --- so that Theorem~\ref{thm:E-rate} applies --- and} balances the first two terms of \eqref{eq:E-rate}, yielding
\begin{equation}
  \norm{\rhat-\rdag}{L^2(P)}
  =O_p\Bigl\{\bigl(\tfrac{\log N}{N}\bigr)^{s/(2s+d)}+a_{L,N}\Bigr\}.
  \label{eq:E-rateopt}
\end{equation}
If a vector-valued concentration inequality is used in Step~1 of the proof in
place of the coordinatewise union bound, the factors $\log J$ and $\log N$ may
be deleted, giving $J_N\asymp N^{d/(2s+d)}$ and
$\norm{\rhat-\rdag}{L^2(P)}=O_p\{N^{-s/(2s+d)}+a_{L,N}\}$ \corr{(the logarithm in the design concentration \eqref{eq:E-Ghat} remains, but affects no rate)}.
\end{corollary}

\begin{proof}
Ignoring constants, write $\mathrm a(J)=J^{-s/d}$ and
$\mathrm b(J)=\sqrt{J\log J/N}$. The optimiser is polynomial in $N$, so
$\log J\asymp\log N$ and it suffices to solve
$J^{-s/d}\asymp\sqrt{J\log N/N}$. Squaring gives
$J^{-2s/d}\asymp J\log N/N$, so $J^{1+2s/d}\asymp N/\log N$; since
$1+2s/d=(2s+d)/d$ this is \eqref{eq:E-Jopt}. Substituting into
$\mathrm a(J_N)$ gives
$J_N^{-s/d}\asymp(\log N/N)^{s/(2s+d)}$, and $\mathrm b(J_N)$ has the same
order by construction, which is \eqref{eq:E-rateopt}.

{ It remains to verify \eqref{eq:E-growth} for this choice. First,
$\sqrt{J_N}\,J_N^{-s/d}\asymp(N/\log N)^{(d-2s)/\{2(2s+d)\}}\to0$ because
$s>d/2$ (Assumption~\ref{ass:des-approx}). Second, observe the fact that
$J_N^2\log J_N/N\asymp N^{(d-2s)/(2s+d)}(\log N)^{(2s-d)/(2s+d)}\to0$, again
because $s>d/2$, so $\sqrt{J_N}\cdot\sqrt{J_N\log J_N/N}\to0$. Third,
$\sqrt{J_N}\,a_{L,N}\asymp(N/\log N)^{d/\{2(2s+d)\}}a_{L,N}\to0$ is exactly
\eqref{eq:E-aL-J}.} Deleting the logarithmic
factors and repeating the computation gives the sharper statement.
\end{proof}
 
\begin{corollary}[Admissible sieve dimensions for efficiency]
\label{cor:E-window}
Suppose in addition that Assumption~\ref{ass:nuisance} holds with rate exponent $\beta>0$,
that $a_{L,N}=o(N^{\beta-1/2})$, and that $J=J_N$ satisfies \corr{\eqref{eq:E-growth} together with}
\begin{equation}
  J_N^{-s/d}=o\bigl(N^{\beta-1/2}\bigr)
  \qquad\text{and}\qquad
  J_N\log J_N=o\bigl(N^{2\beta}\bigr).
  \label{eq:E-window}
\end{equation}
Then the product-rate condition
$\norm{\rhat-\rdag}{L^2(P)}\,\norm{\widehat\mu_0-\mu_0}{L^2(P)}=o_p(N^{-1/2})$
holds. For polynomial $J_N=N^{\kappa}$ the admissible window is, up to
logarithmic slack,
\begin{equation}
  \frac{d(1-2\beta)}{2s}<\kappa<\corr{\min\Bigl\{2\beta,\ \frac12\Bigr\}},
  \label{eq:E-kappa}
\end{equation}
\corr{the upper edge $1/2$ arising from the growth condition
\eqref{eq:E-growth}. For $\beta\ge1/4$ the window is nonempty whenever
$s>d(1-2\beta)$, which is implied by the smoothness floor $s>d/2$; for
$0<\beta<1/4$ it is nonempty} precisely when
$s>d(1-2\beta)/(4\beta)$. In particular, with the rate-optimal choice
\eqref{eq:E-Jopt}, the product-rate condition holds whenever
\begin{equation}
  \beta+\frac{s}{2s+d}>\frac12 .
  \label{eq:E-eff}
\end{equation}
If the outcome regression attains the usual series rate
$N^{-s_\mu/(2s_\mu+d)}$ for a smoothness index $s_\mu>0$, condition
\eqref{eq:E-eff} becomes
\begin{equation}
  \frac{s}{2s+d}+\frac{s_\mu}{2s_\mu+d}>\frac12
  \qquad\Longleftrightarrow\qquad
  s\,s_\mu>\frac{d^2}{4}.
  \label{eq:E-twosob}
\end{equation}
\end{corollary}

\begin{proof}
By Theorem~\ref{thm:E-rate} and Assumption~\ref{ass:nuisance}, the product is
$O_p\{(J^{-s/d}+\sqrt{J\log J/N}+a_{L,N})N^{-\beta}\}$, and we require each of
the three terms multiplied by $N^{-\beta}$ to be $o(N^{-1/2})$. For the first
this is $J^{-s/d}=o(N^{\beta-1/2})$; for the second,
$\sqrt{J\log J/N}\,N^{-\beta}=o(N^{-1/2})$ becomes, after multiplying by
$N^{1/2+\beta}$ and squaring, $J\log J=o(N^{2\beta})$; for the third it is the
assumed bound on $a_{L,N}$. This gives \eqref{eq:E-window}.

With $J_N=N^\kappa$ the first condition reads $\kappa s/d>1/2-\beta$, that is
$\kappa>d(1-2\beta)/(2s)$, and the second reads $\kappa<2\beta$ up to
logarithmic slack. { The growth condition \eqref{eq:E-growth}
additionally requires $J_N^2\log J_N=o(N)$, that is $\kappa<1/2$ up to
logarithms (at $\kappa=1/2$ the middle term of $\sqrt J\,\delta_N$ is of
order $\sqrt{\log N}$ and fails). Together these give \eqref{eq:E-kappa}. The
window is nonempty exactly when $d(1-2\beta)/(2s)<\min\{2\beta,1/2\}$: for
$\beta<1/4$ the binding edge is $2\beta$ and the condition is
$s>d(1-2\beta)/(4\beta)$; for $\beta\ge1/4$ the binding edge is $1/2$ and the
condition is $s>d(1-2\beta)$, which $s>d/2$ implies since
$d(1-2\beta)\le d/2$ there.}

Under \eqref{eq:E-Jopt} the weight error has exponent $s/(2s+d)$ up to
logarithms, so the product with $N^{-\beta}$ is $o(N^{-1/2})$ when
$\beta+s/(2s+d)>1/2$, which is \eqref{eq:E-eff}. Substituting
$\beta=s_\mu/(2s_\mu+d)$ and multiplying \eqref{eq:E-eff} by the positive
quantity $2(2s+d)(2s_\mu+d)$ gives
$2s(2s_\mu+d)+2s_\mu(2s+d)>(2s+d)(2s_\mu+d)$; expanding both sides, \corr{the
terms $2sd$ and $2s_\mu d$ appear on both sides and cancel, the term $d^2$
moves to the right,} and what remains is
$4ss_\mu>d^2$, which is \eqref{eq:E-twosob}. At equality the polynomial orders
only match $N^{-1/2}$ and the logarithmic factors are no longer negligible, so
strict inequality or additional undersmoothing is required.

\corr{Finally, the component $\sqrt{J_N}\,a_{L,N}\to0$ of \eqref{eq:E-growth} restricts the offset
error relative to the sieve dimension; under the root-$N$ offset stability of
Remark~\ref{rem:E-offset-rootN} it holds automatically throughout the window
\eqref{eq:E-kappa}, since $\kappa<1$.}
\end{proof}
 
\begin{remark}[Where $\varepsilon$, $\rho$ and $J$ each act]
\label{rem:E-tuning}
The theory assigns the three tuning parameters sharply different roles, and
this is the practical content of the appendix. The calibration dimension $J$ is
the \emph{rate-driving} parameter: it alone appears in \eqref{eq:E-rate}
through terms that must be balanced, and Corollaries~\ref{cor:E-J}
and~\ref{cor:E-window} determine it. The entropic parameter $\varepsilon$ and
the relaxation $\rho$ are \emph{not} required to vanish; they select which
stable member of the calibrated family is computed, and they enter the
asymptotics only through the offset, whose stability is
Assumption~\ref{ass:des-dual}. Accordingly no condition of the form
$\varepsilon_N=o(N^{-\varkappa})$ appears anywhere above. Should one
nevertheless wish to let $\varepsilon_N$ or $\rho_N$ vary with $N$, the
constants in Lemma~\ref{lem:dual-bdd} and the sequence $a_{L,N}$ must be
tracked explicitly; there is no universal power law in $\varepsilon$ or $\rho$
implied by calibration alone.
 
\corr{The growth condition \eqref{eq:E-growth} imposed in
Theorem~\ref{thm:E-rate} is, up to logarithms, the requirement
$J_N=o(\sqrt{N/\log N})$ familiar from the main text, together with the
smoothness floor $s>d/2$ and the offset restriction
$\sqrt J\,a_{L,N}\to0$.} It
arises for a concrete reason, and it is worth being explicit that it is needed
for the $L^2(P)$ rate itself and not merely for the supremum-norm conclusion:
the parametrisation is exponential, so passing from the log scale --- on which
the $M$-estimation argument operates --- to the weight requires the exponent to
remain in a fixed compact interval, and converting the $\ell_2$ bound on
$\thetahat-\theta_J$ into the necessary supremum-norm bound on the tilt costs a
factor $\sqrt J$. A sharper argument, replacing the uniform Lipschitz step by a
localised exponential inequality in $L^2(P)$, \corr{might relax the middle
requirement towards} $J\log J/N\to0$; we have preferred the conservative hypothesis. This is also the sense in which a sieve that is
too rich for the sample is diagnosable: by
Proposition~\ref{prop:E-Mest}(iii), the calibration program ceases to have a
solution as soon as $\Em[b_J(Z)]$ leaves the \corr{relative interior of the}
convex hull of the design points, and the residual of \eqref{eq:E-calib} is exactly the numerical
signature of that failure.
\end{remark}

\section{Asymptotic Linearity, Efficiency, and Inference}
\label{app:F}
 
The ingredients are now in place: Appendix~\ref{app:C} supplies the influence function
and the exact remainder identity, Appendix~\ref{app:E} the rate at which the calibrated
weight approaches the Riesz representer. This appendix assembles them. We
begin with an exact algebraic decomposition of the estimation error
(Section~\ref{sub:F-expansion}), pass to the efficient central limit theorem
and a variance estimator requiring no resampling
(Section~\ref{sub:F-clt}), and close with double robustness and with a result
delimiting what calibration on a \emph{fixed} sieve can and cannot deliver
(Section~\ref{sub:F-dr}). The last of these is the sharpest statement in the
appendix: calibration is an orthogonality device, and orthogonality on a
finite-dimensional space buys consistency but not efficiency.
 
\subsection{The one-step estimator and its exact expansion}
\label{sub:F-expansion}
 
\begin{definition}[One-step estimator]
\label{def:F-onestep}
Let $\muhat$ be the cross-fitted outcome regression of Assumption~\ref{ass:nuisance} and
$\rhat$ the calibrated weight of Definition~\ref{def:emp}. The one-step estimator is
\begin{equation}
  \that
  \coloneq \Em\bigl[Y^{\mathrm Q}-\muhat(Z)\bigr]
    -\En\Bigl[A\,\rhat(X)\bigl\{Y-\muhat(X)\bigr\}\Bigr].
  \label{eq:F-onestep}
\end{equation}
\corr{On the trial side the foldwise convention of Section~\ref{sub:A-notation} applies:
$\muhat(X_i)$ means $\muhat^{(-k)}(X_i)$ for $i\in\calI_k$. On the target
side $\muhat$ denotes the fold average $\bar\mu_0=K^{-1}\sum_k\muhat^{(-k)}$,
as in the main text, so that \eqref{eq:F-onestep} coincides with the
estimator of Section~\ref{sec:estimation} of the main text.}
\end{definition}

Throughout we write $\dmu\coloneq \muhat-\mu_0$ and $\dr\coloneq \rhat-\rdag$\corr{, and
$\EP$, $\EQ$ applied to expressions containing fitted nuisances denote
integration over a fresh observation with the fitted functions held fixed
(the frozen-argument convention); with this convention every identity below
is exact}.
 
\begin{theorem}[Exact stochastic decomposition]
\label{thm:F-decomp}
Under Assumption~\ref{ass:id} \corr{(with $\dmu,\dr\in L^2(P)$, supplied on an event
of probability tending to one by Assumption~\ref{ass:nuisance} and Theorem~\ref{thm:E-rate})},
\begin{equation}
  \that-\tau_Q
  =(\Em-\EQ)\bigl\{Y^{\mathrm Q}-\mu_0(Z)\bigr\}
   -(\En-\EP)\Bigl[A\,\rdag(X)\bigl\{Y-\mu_0(X)\bigr\}\Bigr]
   +R_{1n}+R_{2n}+R_{3n},
  \label{eq:F-decomp}
\end{equation}
where
\begin{align}
  R_{1n}&=-(\Em-\EQ)\bigl\{\dmu(Z)\bigr\},
  \label{eq:F-R1}\\
  R_{2n}&=(\En-\EP)\bigl\{A\,\rhat(X)\dmu(X)\bigr\}
          -(\En-\EP)\Bigl[A\,\dr(X)\bigl\{Y-\mu_0(X)\bigr\}\Bigr],
  \label{eq:F-R2}\\
  R_{3n}&=\EP\bigl[\dr(X)\,\dmu(X)\bigr].
  \label{eq:F-R3}
\end{align}
The identity is exact: no expansion or approximation is involved.
\end{theorem}
 
\begin{proof}
Substitute $\muhat=\mu_0+\dmu$ in the target term of \eqref{eq:F-onestep} and
centre:
\begin{equation}
  \Em\bigl[Y^{\mathrm Q}-\muhat(Z)\bigr]
  =(\Em-\EQ)\bigl\{Y^{\mathrm Q}-\mu_0(Z)\bigr\}+\tau_Q
   -\Em\bigl[\dmu(Z)\bigr],
  \label{eq:F-target}
\end{equation}
using $\EQ[Y^{\mathrm Q}-\mu_0(Z)]=\tau_Q$ from Proposition~\ref{prop:C-id}.
 
For the trial term, substitute $\rhat=\rdag+\dr$ and
$Y-\muhat=(Y-\mu_0)-\dmu$ and expand into four pieces:
\begin{equation}
  -\En\bigl[A\rhat(Y-\muhat)\bigr]
  =-\En\bigl[A\rdag(Y-\mu_0)\bigr]
   +\En\bigl[A\rdag\dmu\bigr]
   -\En\bigl[A\dr(Y-\mu_0)\bigr]
   +\En\bigl[A\dr\dmu\bigr].
  \label{eq:F-fourpieces}
\end{equation}
We centre each piece in turn, using \corr{\eqref{eq:invrand} and} \eqref{eq:invrand2} of Lemma~\ref{lem:A-invrand} to compute
the population values.
 
The first piece has population value
$\EP[A\rdag(Y-\mu_0)]=\EP[\rdag(\mu_0-\mu_0)]=0$, so it equals
$-(\En-\EP)[A\rdag(Y-\mu_0)]$, the leading trial term of \eqref{eq:F-decomp}.
The third piece likewise has population value zero and equals
$-(\En-\EP)[A\dr(Y-\mu_0)]$; note that this is legitimate for any fixed $\dr$,
since \eqref{eq:invrand2} holds whatever the multiplier.
 
The second piece has population value
$\EP[A\rdag\dmu]=\EP[\rdag\dmu]=\EQ[\dmu(Z)]$, the first equality by
\corr{\eqref{eq:invrand} with $G$ a function of $X$ alone} and the second by the Riesz
identity of Lemma~\ref{lem:A-riesz}. Hence
$\En[A\rdag\dmu]=(\En-\EP)[A\rdag\dmu]+\EQ[\dmu(Z)]$. The fourth piece has
population value $\EP[A\dr\dmu]=\EP[\dr\dmu]=R_{3n}$, so
$\En[A\dr\dmu]=(\En-\EP)[A\dr\dmu]+R_{3n}$.
 
It remains to collect. The term $-\Em[\dmu]$ from \eqref{eq:F-target} and the
term $+\EQ[\dmu]$ from the second piece combine to
$-(\Em-\EQ)\{\dmu\}=R_{1n}$. The two centred empirical processes
$(\En-\EP)[A\rdag\dmu]$ and $(\En-\EP)[A\dr\dmu]$ combine, by linearity, into
the single term $(\En-\EP)[A\rhat\dmu]$, which together with
$-(\En-\EP)[A\dr(Y-\mu_0)]$ is exactly $R_{2n}$; each of these two empirical
processes is used once and once only. What is left is $R_{3n}$ and the two
leading terms, giving \eqref{eq:F-decomp}.
\end{proof}
 
\begin{lemma}[Negligibility of the remainders]
\label{lem:F-rem}
Suppose Assumptions~\ref{ass:id}, \ref{ass:design}, \ref{ass:nuisance} and~\ref{ass:E-twosided} hold, that $J=J_N$ satisfies
\corr{the growth condition \eqref{eq:E-growth} of Theorem~\ref{thm:E-rate} together with
$J_N\log N=o(\sqrt N)$}, that the trial residual is almost surely bounded,
\begin{equation}
  |Y-\mu_0(X)|\le C_R<\infty \quad\text{almost surely on }\{T=0\},
  \label{eq:F-bddres}
\end{equation}
and that\corr{, with $\delta_N$ the deterministic rate of Theorem~\ref{thm:E-rate} and
$N^{-\beta}$ that of Assumption~\ref{ass:nuisance},}
\begin{equation}
  \sqrt{J_N\log N}\;\corr{\delta_N}\to0,
  \qquad
  \sqrt{J_N\log N}\;\corr{N^{-\beta}}\to0.
  \label{eq:F-EPcond}
\end{equation}
Then $R_{1n}=o_p(N^{-1/2})$ and $R_{2n}=o_p(N^{-1/2})$. If in addition
\begin{equation}
  \norm{\dr}{L^2(P)}\,\norm{\dmu}{L^2(P)}=o_p(N^{-1/2}),
  \label{eq:F-product}
\end{equation}
then $R_{3n}=o_p(N^{-1/2})$ as well.
\end{lemma}
 
\begin{proof}
{ Throughout we work on the event $\mathcal E_N$ --- of probability tending to one
--- on which the conclusions of Theorem~\ref{thm:E-rate} and Corollary~\ref{cor:E-Lpop}
hold with a fixed constant $K$:
$\norm{\thetahat-\theta_J}{2}\le K\delta_N$,
$\norm{(\thetahat-\theta_J)^\top b_J}{\infty}\le1$ by \eqref{eq:E-supnorm},
$\norm{\Lhat-L}{\infty}\le K a_{L,N}$, and $\Lhat-L$ lies in the fixed
Sobolev ball $\mathcal B_W\coloneq \{\ell:\norm{\ell}{W^{s,\infty}(\calX)}\le
C_L+C_L'\}$. Since every assertion is probabilistic, this loses no
generality.

\emph{Step 0: a deterministic function class.} Factor the fitted weight as
$\rhat=e^{\gamma(\Lhat-L)}\,W_L(\cdot;\thetahat)$, with
$W_L(x;\theta)=e^{\gamma L(x)+\theta^\top b_J(x)}$ as in the proof of
Theorem~\ref{thm:E-rate}, and define
\small{\[
  \mathcal R_J\coloneq \Bigl\{x\mapsto e^{\gamma\ell(x)}W_L(x;\theta):\
  \ell\in\mathcal B_W,\ \norm{\ell}{\infty}\le K a_{L,N},\
  \norm{\theta-\theta_J}{2}\le K\delta_N,\
  \norm{(\theta-\theta_J)^\top b_J}{\infty}\le1\Bigr\}.
\]}
The class is deterministic, and on $\mathcal E_N$ it contains $\rhat$; the
randomness of $(\thetahat,\Lhat)$ enters only through membership. Three
properties are used repeatedly. \emph{Envelope:} by Step~0 of Theorem~\ref{thm:E-rate}
the exponent of any member is bounded, so
$\sup_{r\in\mathcal R_J}\norm{r}{\infty}\le U<\infty$ with $U$ fixed.
\emph{Radius:} by \eqref{eq:E-lip},
$\norm{r-\rdag}{L^2(P)}\le
 C\{\norm{\ell}{\infty}+\norm{(\theta-\theta_J)^\top b_J}{L^2(P)}
 +\norm{\theta_J^\top b_J-\psi}{\infty}\}
 \le C\{Ka_{L,N}+\sqrt{C_G}\,K\delta_N+C_sJ^{-s/d}\}\le\sigma_N$
for a deterministic $\sigma_N\asymp\delta_N$, and $\sigma_N\ge cJ^{-s/d}$ by
construction. \emph{Entropy:} the tilt component is a $J$-dimensional
family, Lipschitz in $\theta$ on the stated ball, so its $L^2(P)$ bracketing
numbers obey $\log N_{[\,]}(\epsilon)\le CJ\log(CBJ/\epsilon)$; the offset
component $\{e^{\gamma\ell}:\ell\in\mathcal B_W\}$ obeys
$\log N_{[\,]}(\epsilon)\le C\epsilon^{-d/s}$ by the classical entropy bound
for Sobolev balls on a bounded Lipschitz domain
\citep[Corollary~2.7.2]{vaart1996weak}; and entropies of uniformly bounded
classes add under products. Hence
\begin{equation}
  \log N_{[\,]}\bigl(\epsilon,\mathcal R_J,L^2(P)\bigr)
  \le C\bigl\{J\log(CBJ/\epsilon)+\epsilon^{-d/s}\bigr\},
  \label{eq:F-entropy}
\end{equation}
and since $s>d/2$ (Assumption~\ref{ass:des-approx}) the entropy integral
$\int_0^{\sigma}\sqrt{1+\log N_{[\,]}(\epsilon)}\,d\epsilon
 \lesssim\sigma\sqrt{J\log N}+\sigma^{1-d/(2s)}$ converges.

\emph{(i) The term $R_{1n}$.} Let $\calT_R$ denote the $\sigma$-field
generated by the entire trial sample. The target-side prediction
$\bar\mu_0$ is $\calT_R$-measurable, the target sample is independent of
$\calT_R$, and conditionally on $\calT_R$ the summands of
$-m^{-1}\sum_j\{\dmu(X^Q_j)-\EQ\dmu\}$ are centred and independent, so
\[
  \Var\bigl[\sqrt m\,R_{1n}\mid\calT_R\bigr]
  \le\EQ\bigl[\dmu(Z)^2\bigr]
  \le\overline r\,\norm{\dmu}{L^2(P)}^2=o_p(1)
\]
by \eqref{eq:PtoQ}. For $u>0$,
$\P(\sqrt m|R_{1n}|>u)
 \le\E\bigl[\min\{u^{-2}\Var(\sqrt mR_{1n}\mid\calT_R),1\}\bigr]\to0$
by dominated convergence, whence $R_{1n}=o_p(m^{-1/2})$.

\emph{(ii) The second term of $R_{2n}$.} Consider the deterministic class
\[
  \mathcal G_J
  \coloneq \Bigl\{(x,t,y)\mapsto\tfrac{1-t}{e_0(x)}\{r(x)-\rdag(x)\}\{y-\mu_0(x)\}:
     r\in\mathcal R_J\Bigr\}.
\]
Every member is centred given $X$ (by \eqref{eq:invrand}, each $r-\rdag$
being bounded), has envelope
$\underline e^{-1}(U+\overline r)C_R$ by \eqref{eq:F-bddres}, and $L^2(P)$
norm at most $\sqrt{C_Y/\underline e}\,\sigma_N$; multiplying by the fixed
bounded function $(x,t,y)\mapsto\tfrac{1-t}{e_0(x)}\{y-\mu_0(x)\}$ rescales
brackets by its supremum and so preserves the entropy bound
\eqref{eq:F-entropy} up to constants. The bracketing maximal inequality
\citep[Lemma~3.4.2]{vaart1996weak} gives, unconditionally,
\begin{equation}
  \E\Bigl[\sup_{g\in\mathcal G_J}\bigl|\sqrt n(\En-\EP)g\bigr|\Bigr]
  \;\lesssim\;
  \sigma_N\sqrt{J\log N}+\sigma_N^{1-d/(2s)}
  +\frac{J\log N+\sigma_N^{-d/s}}{\sqrt n}.
  \label{eq:F-maximal}
\end{equation}
The first term vanishes by the first half of \eqref{eq:F-EPcond}; the second
because $s>d/2$ and $\sigma_N\to0$; in the third, $J\log N/\sqrt n\to0$ is
the hypothesis $J_N\log N=o(\sqrt N)$, and
$\sigma_N^{-d/s}\le CJ$ because $\sigma_N\ge cJ^{-s/d}$. Since
$\rhat\in\mathcal R_J$ on $\mathcal E_N$,
$(\En-\EP)[A\dr(Y-\mu_0)]=o_p(n^{-1/2})$. No conditioning on the dual
vector, and no independence between $\Lhat$ and the trial sample, is used
anywhere: the class is fixed in advance.

\emph{(iii) The first term of $R_{2n}$.} Decompose
$(\En-\EP)[A\rhat\dmu]=(\En-\EP)[A\rdag\dmu]+(\En-\EP)[A\dr\dmu]$ and
treat both summands foldwise, writing $\En=\sum_k(n_k/n)\E_{n,k}$. For the
first, condition on the training $\sigma$-field $\calT_k$ of fold $k$:
$\dmu^{(-k)}$ is then a fixed function, the fold-$k$ observations are
i.i.d.\ and independent of $\calT_k$, and conditional Chebyshev with
$\EP[A^2f(X)]\le\underline e^{-1}\EP[f(X)]$ gives
\[
  \Var\bigl[\sqrt{n_k}\,(\E_{n,k}-\EP)[A\rdag\dmu^{(-k)}]\mid\calT_k\bigr]
  \le\underline e^{-1}\,\overline r^{\,2}\,
     \norm{\dmu^{(-k)}}{L^2(P)}^2
     \;+\;\underline e^{-1}\overline r^{\,2}C_Y\norm{\dmu^{(-k)}}{L^2(P)}^2
  =o_p(1),
\]
where the two contributions bound the conditional variances of the
$X$-measurable part and of the residual part. For the second summand,
condition again on $\calT_k$: the function $\dmu^{(-k)}$ is fixed and bounded
by $C_\mu+C_0$ (Assumptions~\ref{ass:nuisance} and \ref{ass:id-mom}), and the class
$\{A(r-\rdag)\dmu^{(-k)}:r\in\mathcal R_J\}$ is deterministic given
$\calT_k$, with bounded envelope
$\underline e^{-1}(U+\overline r)(C_\mu+C_0)$, entropy
\eqref{eq:F-entropy}, and $L^2(P)$ radius
$\varsigma_k\coloneq C\bigl(\norm{\dmu^{(-k)}}{L^2(P)}\vee J^{-s/d}\bigr)$, a
$\calT_k$-measurable quantity to which the conditional form of
\eqref{eq:F-maximal} applies. The resulting bound,
$\varsigma_k\sqrt{J\log N}+\varsigma_k^{1-d/(2s)}
 +\{J\log N+\varsigma_k^{-d/s}\}/\sqrt{n_k}$,
is $o_p(1)$ by the second half of \eqref{eq:F-EPcond} (note
$\sqrt{J\log N}\,J^{-s/d}\to0$ under \eqref{eq:F-EPcond} as well), $s>d/2$,
$J\log N=o(\sqrt N)$, and $\varsigma_k^{-d/s}\le CJ$ by the floor
$J^{-s/d}$ built into the radius.

\emph{(iv) The term $R_{3n}$.} Under the frozen-argument convention,
Cauchy--Schwarz gives the upper bound 
$|R_{3n}|\le\norm{\dr}{L^2(P)}\norm{\dmu}{L^2(P)}$, which is
$o_p(N^{-1/2})$ by \eqref{eq:F-product}.}
\end{proof}
 
\begin{remark}[Why $\rhat$ need not be cross-fitted, and what it costs]
\label{rem:F-crossfit}
The two nuisances make different demands, and the proof above shows exactly
why. Cross-fitting is required for $\muhat$, whose errors are correlated with
the very residuals it is used to centre; it is what makes the conditional
Chebyshev arguments in (i) and (iii) available. It is \emph{not} required for
$\rhat$, which is a function of covariates and of the target sample only and is
therefore free of the trial outcomes. What replaces cross-fitting is the
empirical-process control \eqref{eq:F-maximal}, available because
Theorem~\ref{thm:E-rate} \corr{and Corollary~\ref{cor:E-Lpop} confine $\rhat$, with
probability tending to one, to the deterministic class $\mathcal R_J$, whose
entropy \eqref{eq:F-entropy} combines a $J$-dimensional tilt with a fixed
Sobolev ball for the estimated offset}.
 
The price is the pair of conditions \eqref{eq:F-EPcond}, and it is worth
recording what they amount to. Under the rate-optimal choice
$J_N\asymp(N/\log N)^{d/(2s+d)}$ of Corollary~\ref{cor:E-J} and the rate
\eqref{eq:E-rateopt}, the first condition reads, up to logarithmic factors,
\[
  N^{d/\{2(2s+d)\}}\cdot N^{-s/(2s+d)}
  =N^{(d-2s)/\{2(2s+d)\}}\to0
  \qquad\Longleftrightarrow\qquad
  s>\frac d2 ,
\]
\corr{together with the offset component
$\sqrt{J_N\log N}\,a_{L,N}\to0$, a $\sqrt{\log N}$-strengthening of the
$\sqrt{J_N}\,a_{L,N}\to0$ already contained in \eqref{eq:E-growth},
automatic under the root-$N$ offset stability of
Remark~\ref{rem:E-offset-rootN}; the smoothness floor $s>d/2$ is part of
Assumption~\ref{ass:des-approx}, so at the rate-optimal sieve the first condition is
automatic}. The second
condition is implied by \corr{$J_N\log N=o(N^{2\beta})$, which for
polynomial $J_N$ coincides with the second condition in \eqref{eq:E-window}
and follows from \eqref{eq:E-Jopt} whenever $\beta+s/(2s+d)>1/2$}. It should be said plainly that the assertion
``$\rhat$ lies in a finite-dimensional and hence Donsker class'' is correct at
\emph{fixed} $J$ but is not by itself enough when $J=J_N\to\infty$: the entropy
of $\mathcal R_{J_N}$ grows, and \eqref{eq:F-EPcond} is what compensates.
 
A second price is the bounded-residual condition \eqref{eq:F-bddres}, which the
maximal inequality needs and which the $(2+\delta)$-moment condition of
Assumption~\ref{ass:id-mom} does not supply, a bounded second moment not being a bounded
envelope. Condition \eqref{eq:F-bddres} is harmless for the bounded and count
outcomes of the applications we have in mind, and it may be dispensed with in
either of two ways: by a truncation argument combined with a maximal inequality
for classes with an $L^{2+\delta}$ envelope, or --- more simply --- by
cross-fitting $\rhat$ as well. In the latter case $\dr$ is fixed on the
evaluation fold given \corr{the foldwise training $\sigma$-field $\calT_k$},
the empirical-process argument of (ii) is replaced by the \corr{foldwise}
conditional Chebyshev bound
\[
  \Var\Bigl\{\sqrt {n_k}(\E_{n,k}-\EP)\bigl[A\,\dr(X)\{Y-\mu_0(X)\}\bigr]
    \;\Big|\;\calT_k\Bigr\}
  \le C_Y\underline e^{-\corr{1}}\norm{\dr}{L^2(P)}^2=o_p(1),
\]
and both \eqref{eq:F-bddres} and \eqref{eq:F-EPcond} may be dropped. We have
preferred the non-cross-fitted version because it is what one computes in
practice and because it costs a factor $K$ less in solver time.
\end{remark}
 
\subsection{The efficient central limit theorem and variance estimation}
\label{sub:F-clt}
 
\begin{theorem}[Asymptotic linearity and efficient CLT]
\label{thm:F-clt}
Suppose the hypotheses of Lemma~\ref{lem:F-rem} hold. Then
\begin{equation}
  \that-\tau_Q
  =\frac1m\sum_{j=1}^m\varphi_Q(O_j^{\mathrm Q})
   +\frac1n\sum_{i=1}^n\varphi_P(O_i^{\mathrm R})
   +o_p(N^{-1/2}),
  \label{eq:F-AL}
\end{equation}
with $\varphi_Q$ and $\varphi_P$ the efficient influence functions
\eqref{eq:eifQ}--\eqref{eq:eifP} of Theorem~\ref{thm:C-eif}, and consequently
\begin{equation}
  \sqrt m\,(\that-\tau_Q)\;\xrightarrow{d}\;
  \mathcal N\bigl(0,\,V_{\mathrm{eff}}\bigr),
  \qquad
  V_{\mathrm{eff}}=\Var_Q[\varphi_Q]+\eta^{-1}\Var_P[\varphi_P],
  \label{eq:F-CLT}
\end{equation}
so that $\that$ attains the semiparametric efficiency bound.
\end{theorem}
 
\begin{proof}
By Theorem~\ref{thm:F-decomp} and Lemma~\ref{lem:F-rem}, the three remainders
are $o_p(N^{-1/2})$, so it suffices to identify the two leading terms. Centering
the target term by $\EQ[Y^{\mathrm Q}-\mu_0(Z)]=\tau_Q$ gives
\[
  (\Em-\EQ)\bigl\{Y^{\mathrm Q}-\mu_0(Z)\bigr\}
  =\frac1m\sum_{j=1}^m\bigl\{Y_j^{\mathrm Q}-\mu_0(X^Q_j)-\tau_Q\bigr\}
  =\frac1m\sum_{j=1}^m\varphi_Q(O_j^{\mathrm Q}),
\]
and since $\EP[A\rdag(Y-\mu_0)]=0$ the trial term equals
$n^{-1}\sum_i\varphi_P(O_i^{\mathrm R})$ with $\varphi_P$ as in
\eqref{eq:eifP}. This is \eqref{eq:F-AL}.
 
Multiplying by $\sqrt m$,
\[
  \sqrt m\,(\that-\tau_Q)
  =\underbrace{\frac1{\sqrt m}\sum_{j}\varphi_Q(O_j^{\mathrm Q})}_{S_m^{\mathrm Q}}
   +\sqrt{\tfrac mn}\;
    \underbrace{\frac1{\sqrt n}\sum_{i}\varphi_P(O_i^{\mathrm R})}_{S_n^{\mathrm P}}
   +o_p(1).
\]
The summands of $S_m^{\mathrm Q}$ are i.i.d., centred \corr{(the
``Centring'' step in the proof of Theorem~\ref{thm:C-eif})}, with finite \corr{variance}
by Assumption~\ref{ass:id-mom}, so the \corr{classical i.i.d.}\ central limit theorem gives
$S_m^{\mathrm Q}\xrightarrow{d}\mathcal N(0,\Var_Q\varphi_Q)$; the same argument
gives $S_n^{\mathrm P}\xrightarrow{d}\mathcal N(0,\Var_P\varphi_P)$, and
$\sqrt{m/n}\to\eta^{-1/2}$. The two samples are independent, so
$S_m^{\mathrm Q}$ and $S_n^{\mathrm P}$ are independent, and \corr{joint
convergence follows directly, the characteristic functions factorising};
Slutsky's lemma disposes of the $o_p(1)$ term. The limit variance is \eqref{eq:F-CLT}, which is the bound
of Theorem~\ref{thm:C-eif} because $(\varphi_Q,\varphi_P)$ is the canonical gradient.
\end{proof}
 
\begin{theorem}[Variance consistency and Wald intervals]
\label{thm:F-var}
Under the hypotheses of Theorem~\ref{thm:F-clt}, define
\begin{equation}
  \widehat\varphi_{Q,j}\coloneq Y_j^{\mathrm Q}-\muhat(X^Q_j)-\that,
  \qquad
  \widehat\varphi_{P,i}
  \coloneq -\,\frac{1-T_i}{e_0(X_i)}\,\rhat(X_i)\bigl\{Y_i-\muhat(X_i)\bigr\},
  \label{eq:F-ifhat}
\end{equation}
and
\begin{equation}
  \widehat V
  \coloneq \frac1m\sum_{j=1}^m\widehat\varphi_{Q,j}^{\,2}
   +\frac mn\cdot\frac1n\sum_{i=1}^n\widehat\varphi_{P,i}^{\,2}.
  \label{eq:F-Vhat}
\end{equation}
Then $\widehat V\xrightarrow{p}V_{\mathrm{eff}}$, and consequently\corr{,
provided $V_{\mathrm{eff}}>0$,}
$\sqrt m(\that-\tau_Q)/\widehat V^{1/2}\xrightarrow{d}\mathcal N(0,1)$, so that
$\that\pm z_{\corr{1-}\alpha/2}\widehat V^{1/2}/\sqrt m$ is an asymptotically valid
$(1-\alpha)$ confidence interval.
\end{theorem}
 
\begin{proof}
We treat the two components separately, and in each case the pattern is the
same: show that the estimated influence function converges to the true one in
$L^2$, transfer this to the empirical second moments by Cauchy--Schwarz, and
apply an ordinary law of large numbers to the true influence function.
 
\emph{Target component.} From \eqref{eq:F-ifhat} and \eqref{eq:eifQ},
$\widehat\varphi_Q-\varphi_Q=-\dmu-(\that-\tau_Q)$, so by
$(a+b)^2\le2a^2+2b^2$,
\[
  \Em\bigl[(\widehat\varphi_Q-\varphi_Q)^2\bigr]
  \le2\,\Em\bigl[\dmu(Z)^2\bigr]+2(\that-\tau_Q)^2=o_p(1),
\]
the first term by \eqref{eq:PtoQ} and a conditional law of large numbers, the
second by the consistency implied by Theorem~\ref{thm:F-clt}. Cauchy--Schwarz
then gives
\[
  \bigl|\Em[\widehat\varphi_Q^{\,2}]-\Em[\varphi_Q^2]\bigr|
  \le\bigl(\Em[(\widehat\varphi_Q-\varphi_Q)^2]\bigr)^{1/2}
     \bigl(\Em[(\widehat\varphi_Q+\varphi_Q)^2]\bigr)^{1/2}
  =o_p(1),
\]
the second factor being $O_p(1)$ under the $(2+\delta)$-moment condition, and
the law of large numbers gives
$\Em[\varphi_Q^2]\xrightarrow{p}\Var_Q[\varphi_Q]$.
 
\emph{Trial component.} Adding and subtracting,
\[
  \widehat\varphi_P-\varphi_P
  =-A\,\dr(X)\bigl\{Y-\mu_0(X)\bigr\}+A\,\rhat(X)\,\dmu(X).
\]
For the first term, positivity and the conditional residual-variance bound of
Assumption~\ref{ass:id-mom} give
$\EP[A^2\dr^2(Y-\mu_0)^2]\le C_Y\underline e^{-\corr{1}}\norm{\dr}{L^2(P)}^2=o_p(1)$;
for the second, positivity and $\norm{\rhat}{\infty}=O_p(1)$ from
Theorem~\ref{thm:E-rate} give
$\EP[A^2\rhat^2\dmu^2]=O_p(\norm{\dmu}{L^2(P)}^2)=o_p(1)$. Hence
$\norm{\widehat\varphi_P-\varphi_P}{L^2(P)}=o_p(1)$ \corr{in the
frozen-argument sense. Because $\rhat$ is not cross-fitted, transferring this
to the empirical second moment requires one further uniform step: on the
event $\mathcal E_N$ of Lemma~\ref{lem:F-rem},
$(\widehat\varphi_P-\varphi_P)^2
 \le2A^2(Y-\mu_0)^2(r-\rdag)^2+2A^2r^2\dmu^2$ with $r=\rhat\in\mathcal R_J$,
and the classes $\{A^2(y-\mu_0)^2(r-\rdag)^2:r\in\mathcal R_J\}$ and,
foldwise, $\{A^2r^2(\dmu^{(-k)})^2:r\in\mathcal R_J\}$ are uniformly bounded
--- by \eqref{eq:F-bddres} and the bounded ranges of $\muhat$ and $\mu_0$ ---
with bracketing entropy \eqref{eq:F-entropy} up to constants; a bracketing
Glivenko--Cantelli bound then gives
$\En[(\widehat\varphi_P-\varphi_P)^2]
 =\EP[(\widehat\varphi_P-\varphi_P)^2]+o_p(1)=o_p(1)$ under
$J\log N=o(n)$, which holds a fortiori. Cauchy--Schwarz then transfers this
to $|\En[\widehat\varphi_P^{\,2}]-\En[\varphi_P^2]|=o_p(1)$}, while
$\En[\varphi_P^2]\xrightarrow{p}\Var_P[\varphi_P]$ \corr{by the law of large
numbers, $\EP[\varphi_P^2]$ being finite by Assumption~\ref{ass:id-mom}}.
 
\emph{Combination.} Since $m/n\to\eta^{-1}$,
$\widehat V\xrightarrow{p}\Var_Q[\varphi_Q]+\eta^{-1}\Var_P[\varphi_P]
 =V_{\mathrm{eff}}$, and the interval statement follows from
Theorem~\ref{thm:F-clt} and Slutsky's lemma.
\end{proof}
 
\begin{corollary}[Primitive sufficient conditions]
\label{cor:F-primitive}
Suppose Assumptions~\ref{ass:id}, \ref{ass:design}, \ref{ass:nuisance} and~\ref{ass:E-twosided} hold, that \corr{the
bounded-residual condition \eqref{eq:F-bddres} holds}, that
$\norm{\muhat-\mu_0}{L^2(P)}=O_p(N^{-\beta})$, that
\corr{$\sqrt{J_N\log N}\,a_{L,N}\to0$ and} $a_{L,N}=o(N^{\beta-1/2})$, and
that $J_N$ is chosen according to \eqref{eq:E-Jopt} with
$\beta+s/(2s+d)>1/2$ \corr{(the floor $s>d/2$ being part of
Assumption~\ref{ass:des-approx})}. Then \corr{all growth conditions of
Lemma~\ref{lem:F-rem}}, \eqref{eq:F-EPcond} and \eqref{eq:F-product} hold,
and consequently \eqref{eq:F-CLT} and
$\widehat V\xrightarrow{p}V_{\mathrm{eff}}$ hold.
\end{corollary}
 
\begin{proof}
Theorem~\ref{thm:E-rate} and Corollary~\ref{cor:E-J} give
$\delta_N=O\{(\log N/N)^{s/(2s+d)}+a_{L,N}\}$
\corr{(their hypotheses hold: the growth condition \eqref{eq:E-growth}
follows from \eqref{eq:E-Jopt}, $s>d/2$ and $\sqrt{J_N}\,a_{L,N}\to0$, the
latter implied by $\sqrt{J_N\log N}\,a_{L,N}\to0$)}. Condition
\eqref{eq:F-product} is then Corollary~\ref{cor:E-window} under
$\beta+s/(2s+d)>1/2$ and $a_{L,N}=o(N^{\beta-1/2})$. The first half of
\eqref{eq:F-EPcond} is the computation displayed in
Remark~\ref{rem:F-crossfit}, valid because $s>d/2$ \corr{together with the
assumed $\sqrt{J_N\log N}\,a_{L,N}\to0$}; the second half follows from
\corr{\eqref{eq:E-Jopt} and $2\beta>d/(2s+d)$, which is
$\beta+s/(2s+d)>1/2$}. Finally $J_N\asymp(N/\log N)^{d/(2s+d)}$ satisfies
\corr{$J_N\log N=o(\sqrt N)$} precisely when \corr{$d/(2s+d)<1/2$}, that is $s>d/2$.
\end{proof}
 
\subsection{Double robustness, and the limits of fixed-sieve calibration}
\label{sub:F-dr}
 
\begin{theorem}[Double robustness]
\label{thm:F-dr}
Suppose Assumption~\ref{ass:id} holds and that
\corr{$\that-\Psi(\muhat,\rhat)=o_p(1)$, with $\Psi$ the population
estimating functional \eqref{eq:Psi} evaluated at the fitted nuisances under
the frozen-argument convention --- a stochastic-equicontinuity condition
supplied, for instance, by the empirical-process bounds of
Lemma~\ref{lem:F-rem}, by cross-fitting both nuisances, or by
$\norm{\rhat}{\infty}=O_p(1)$ combined with foldwise conditional laws of
large numbers}. Then $\that\xrightarrow{p}\tau_Q$ if
\emph{either}
\begin{enumerate}[label=\textup{(\Alph*)},leftmargin=2.6em,itemsep=2pt,topsep=3pt]
\item $\norm{\dr}{L^2(P)}=o_p(1)$ and $\norm{\dmu}{L^2(P)}=O_p(1)$, \emph{or}
\item $\norm{\dmu}{L^2(P)}=o_p(1)$ and $\norm{\dr}{L^2(P)}=O_p(1)$.
\end{enumerate}
\corr{Beyond this displacement condition, }neither Assumption~\ref{ass:E-twosided} nor any rate condition is required.
\end{theorem}
 
\begin{proof}
By Proposition~\ref{prop:C-rem}, $\Psi(\muhat,\rhat)-\tau_Q=\EP[\dr\,\dmu]$ exactly, so
Cauchy--Schwarz gives
$|\Psi(\muhat,\rhat)-\tau_Q|\le\norm{\dr}{L^2(P)}\norm{\dmu}{L^2(P)}$. Under
(A) the first factor is $o_p(1)$ and the second $O_p(1)$; under (B) the roles
are exchanged. Either way the product is $o_p(1)$, and adding the
$o_p(1)$ empirical fluctuation gives $\that\xrightarrow{p}\tau_Q$.
\end{proof}
 
The remaining results delimit what calibration on a fixed sieve achieves. The
first says that a limiting weight other than $\rdag$ is incompatible with
regularity in the unrestricted model; the second draws the consequence for
fixed $J$. Both require that the trial-side tangent space be rich enough to
generate arbitrary perturbations of the control regression, which we state as a
hypothesis of the proposition rather than as a standing assumption.
 
\begin{proposition}[A noncanonical limiting weight cannot be efficient]
\label{prop:F-noncanonical}
Suppose that for every bounded $h\in L^2(P)$ there is a regular submodel
$t\mapsto P_t$, with $Q$ held fixed, whose control-regression derivative
satisfies $\dot\mu_0=h$ almost surely. Let $r_\infty\in L^2(P)$ and put
\[
  \varphi_P^{\,r_\infty}(O^{\mathrm R})
  \coloneq -A\,r_\infty(X)\bigl\{Y-\mu_0(X)\bigr\}.
\]
If an estimator is regular in the unrestricted two-sample model and
asymptotically linear with target component $\varphi_Q$ and trial component
$\varphi_P^{\,r_\infty}$, then $r_\infty=\rdag$ almost surely.
\end{proposition}
 
\begin{proof}
Fix a bounded $h\in L^2(P)$ and a submodel as in the hypothesis. By
\eqref{eq:id} the true pathwise derivative along this submodel is
$\dot\tau_Q=-\EQ[h(Z)]=-\EP[\rdag(X)h(X)]$, the second equality by the
Riesz identity of Lemma~\ref{lem:A-riesz}.
 
On the other hand, regularity with trial-side influence function
$\varphi_P^{\,r_\infty}$ forces the induced derivative to equal
$\EP[\varphi_P^{\,r_\infty}\zeta_P]$ \corr{\citep[Lemma~25.23]{vaart1998asymptotic};
along these submodels $Q$ is held fixed, so $\zeta_Q=0$ and the target
component contributes nothing}. Computing this exactly as in the proof
of Theorem~\ref{thm:C-eif}, first conditioning on $\{X,T=0\}$ via Lemma~\ref{lem:A-invrand} and then
applying Lemma~\ref{lem:C-deriv}(ii),
\[
  \EP\bigl[\varphi_P^{\,r_\infty}\zeta_P\bigr]
  =-\EP\Bigl[r_\infty(X)\,
     \EP\bigl[\{Y-\mu_0(X)\}\zeta_P\mid X,\,T=0\bigr]\Bigr]
  =-\EP\bigl[r_\infty(X)h(X)\bigr].
\]
Equating the two expressions gives $\EP[\{r_\infty-\rdag\}h]=0$ for every
bounded $h\in L^2(P)$, hence for every $h\in L^2(P)$ by density. Taking
$h=r_\infty-\rdag$ yields $\norm{r_\infty-\rdag}{L^2(P)}^2=0$.
\end{proof}
 
\begin{corollary}[Fixed-sieve calibration gives restricted representation only]
\label{cor:F-fixedJ}
Let $r_J\in L^2(P)$ satisfy the calibration equations \eqref{eq:calib-pop} for
a \emph{fixed} $J$. Then $\EP[r_Jh]=\EQ[h]$ for every $h\in H_J$, so the
trial-side score built from $r_J$ correctly represents the pathwise derivative
along every submodel with $\dot\mu_0\in H_J$. If, however, $r_J\ne\rdag$, then
under the hypothesis of Proposition~\ref{prop:F-noncanonical} that score is
\emph{not} the efficient influence-function component in the unrestricted
model. \corr{Thus the pair $(\varphi_Q,\varphi_P^{\,r_J})$ represents the
pathwise derivative only on the submodel $\{\dot\mu_0\in H_J\}$: an estimator
that is asymptotically linear with this influence pair is regular
\emph{within that submodel} but, by
Proposition~\ref{prop:F-noncanonical}, cannot be regular in the unrestricted
model; the convolution theorem therefore does not apply to it, and its
asymptotic variance $\Var_Q[\varphi_Q]+\eta^{-1}\Var_P[\varphi_P^{\,r_J}]$ is
not comparable to the bound --- it can even fall below $V_{\mathrm{eff}}$,
the hallmark of irregularity. Whether the fixed-$J$ estimator is
asymptotically linear at all is a separate question: by the calibration
orthogonality the product remainder is
$R_{3n}=\EP[\dr\,(I-\Pi_{H_J})\dmu]$, of order
$c_J\norm{(I-\Pi_{H_J})\dmu}{L^2(P)}$ with $\Pi_{H_J}$ the $L^2(P)$
projection, and this is $o_p(N^{-1/2})$ only under additional conditions on
$\dmu$; we do not pursue them.}
\end{corollary}
 
\begin{proof}
The first assertion is the sieve half of Lemma~\ref{lem:A-riesz}, and the derivative
computation in the proof of Proposition~\ref{prop:F-noncanonical} shows that
for $\dot\mu_0=h\in H_J$ the induced derivative $-\EP[r_Jh]$ agrees with the
true derivative $-\EQ[h]$. The second assertion is
Proposition~\ref{prop:F-noncanonical} in contrapositive form: since
$r_J\ne\rdag$, the pair $(\varphi_Q,\varphi_P^{\,r_J})$ cannot be the canonical
gradient of $\tau_Q$ in the unrestricted model, so the corresponding estimator
is not efficient there.
\end{proof}
 
\begin{remark}[Calibration, consistency, and efficiency]
\label{rem:F-summary}
It is worth putting the three preceding results side by side, because together
they explain why the sieve must grow and what is lost if it does not.
Calibration is an \emph{orthogonality device}: it enforces
$\EP[\{r_J-\rdag\}h]=0$ for $h\in H_J$ and thereby \corr{sharpens the product
remainder \eqref{eq:remainder} to
$\EP[\dr\,(I-\Pi_{H_J})\dmu]$, annihilating its component in the calibration
space --- this refinement is what a fixed calibration space genuinely buys,
though the proof of Theorem~\ref{thm:F-dr} itself uses only the
Cauchy--Schwarz bound}. \corr{Consistency under branch (B) of
Theorem~\ref{thm:F-dr} then survives at fixed $J$}: unless the sieve already
identifies the Riesz representer, the calibrated weight converges
\corr{in probability} to some
$r_J\ne\rdag$, so that $\norm{\rhat-\rdag}{L^2(P)}$ need not vanish but tends
to \corr{$c_J\coloneq \norm{r_J-\rdag}{L^2(P)}>0$}; double robustness through branch (B) survives regardless,
the remainder being \corr{at most} $c_J\norm{\dmu}{L^2(P)}\to0$.
 
It is not enough for anything more. Branch (A) of Theorem~\ref{thm:F-dr},
asymptotic linearity with the canonical gradient, and efficiency all require
$\norm{\rhat-\rdag}{L^2(P)}\to0$, which by Corollary~\ref{cor:F-fixedJ} cannot
happen at fixed $J$ unless the sieve already contains the representer. This is
why Theorem~\ref{thm:E-rate} is stated with $J=J_N\to\infty$, and why the growth of the
calibration space --- not the vanishing of the regularisation --- is the
asymptotic mechanism of the whole construction.
\end{remark}

\section{Competing Estimators}
\label{app:estimators}

This appendix collects the estimators compared in Sections~\ref{sec:sims}
and~\ref{sec:realdata}. All target the same estimand
$\tau_Q = \E_Q[Y] - \E_{Q}[\mu_0(X^Q)]$ and differ only in which nuisances
they use and how the weight is constructed. Throughout,
$\bar{Y}_{\mathrm{Q}} = m^{-1}\sum_{j=1}^m Y_j^{\mathrm{Q}}$ is the observed
treated mean in the target sample; $\Pnz$ denotes the empirical measure of
the $n_0$ RCT controls, $\Pnz f = n_0^{-1}\sum_{i:T_i=0} f(O_i)$; $\muhat$
is the cross-fitted control-arm outcome regression; and each weighting
estimator uses a weight $\hat{r}$ evaluated at the control-arm covariates.
Because the trial is randomized, we adopt the arm-conditional convention
throughout, normalising control-arm sums by the realised arm size $n_0$;
with $e_0$ known this is equivalent to weighting by
$(1-T)/e_0(X)$ up to the realised-versus-expected arm-size ratio, and it is
the convention implemented in the software.

\paragraph{Unadjusted (\textsc{Naive}).} The raw difference between the
target treated mean and the trial control mean,
\begin{equation}
\hat\tau_{\textsc{n}}
= \bar{Y}_{\mathrm{Q}} - \Pnz Y .
\label{eq:app-naive}
\end{equation}
It estimates $\E_Q[Y(1)] - \E_P[Y(0)]$, which differs from $\tau_Q$ by
$\E_{\PX}[(\rdag - 1)\,\mu_0]$: the covariate-shift gap. It is consistent
only when $\QX = \PX$.

\paragraph{G-computation (\textsc{G-comp}).} Outcome modelling alone,
\begin{equation}
\hat\tau_{\textsc{g}}
= \frac{1}{m}\sum_{j=1}^m \bigl\{Y_j^{\mathrm{Q}} - \muhat(X^Q_j)\bigr\},
\label{eq:app-gcomp}
\end{equation}
the plug-in version of the identification formula \eqref{eq:identification}.
Consistent if and only if $\muhat$ is consistent for $\mu_0$; no weight
enters.

\paragraph{Sampling-score weighting (\textsc{IPW-PS}).} A logistic
membership model $\hat\pi_S$ is fitted to the pooled covariate sample (full
trial versus target), and the density-ratio weight is the scaled odds,
\begin{equation}
\hat{r}_{\mathrm{ps}}(x)
= \frac{\hat\pi_S(x)}{1 - \hat\pi_S(x)}\cdot\frac{n}{m},
\qquad
\hat\tau_{\textsc{ips}}
= \bar{Y}_{\mathrm{Q}} - \Pnz\bigl[\hat{r}_{\mathrm{ps}}\,Y\bigr].
\label{eq:app-ipwps}
\end{equation}
Consistent when the sampling score is correctly specified, i.e.\ when
$\log\rdag$ lies in the span of the score model; under misspecification
$\hat r_{\mathrm{ps}}$ converges to a wrong limit and the bias is $O(1)$.

\paragraph{Transport weighting (\textsc{IPW-OT}).} The same
Horvitz--Thompson form with the calibrated transport weight of
Definition~\ref{def:rcot},
\begin{equation}
\hat\tau_{\textsc{iot}}
= \bar{Y}_{\mathrm{Q}} - \Pnz\bigl[\hat{r}\,Y\bigr].
\label{eq:app-ipwot}
\end{equation}
No outcome model enters; consistency requires only the weight leg.

\paragraph{Augmented sampling-score weighting (\textsc{AIPW-PS}).} The
one-step form of Definition~\ref{eq:drot} with
$\hat{r}_{\mathrm{ps}}$ in place of the calibrated weight,
\begin{equation}
\hat\tau_{\textsc{aps}}
= \frac{1}{m}\sum_{j=1}^m \bigl\{Y_j^{\mathrm{Q}} - \muhat(X^Q_j)\bigr\}
- \Pnz\bigl[\hat{r}_{\mathrm{ps}}\{Y - \muhat(X)\}\bigr].
\label{eq:app-aipwps}
\end{equation}
Doubly robust in the classical sense: consistent if the sampling score
\emph{or} the outcome model is correct. When both are misspecified the
product remainder of Proposition~\ref{prop:remainder} is $O(1)$ and no
amount of data removes the bias.

\paragraph{Proposed (\textsc{RiCOT}).} The one-step estimator of
Definition~\ref{eq:drot} with the Riesz-calibrated transport weight,
\begin{equation}
\hat\tau_{\textsc{drot}}
= \frac{1}{m}\sum_{j=1}^m \bigl\{Y_j^{\mathrm{Q}} - \muhat(X^Q_j)\bigr\}
- \Pnz\bigl[\hat{r}\{Y - \muhat(X)\}\bigr],
\label{eq:app-drot}
\end{equation}
whose properties are the subject of Sections~\ref{sec:riesz}
and~\ref{sec:implementation}.

Three implementation conventions make the comparison fair. First, the same
cross-fitted $\muhat$ enters every estimator that uses an outcome model, so
\textsc{AIPW-PS} versus \textsc{RiCOT} differ in exactly one ingredient: how
the weight is built. Second, both weightings are learned from the same
covariate data --- the full trial sample against the full target sample,
which randomization licenses since both trial arms share the covariate law
$\PX$ --- and evaluated on the control arm. The calibration constraints of
Definition~\ref{def:rcot} accordingly pin the moments of $\hat r$ over the
full trial sample; restricted to the control arm the weighted mean deviates
from one by sampling noise only, and self-normalised (H\'ajek) versions of
\eqref{eq:app-ipwps}--\eqref{eq:app-ipwot} put the two weightings on equal
footing when that deviation matters in small samples. Third, interval
estimates for the augmented pair use the plug-in variance
of Equation~\ref{eq:plugin} applied with the respective weight; in the
simulations coverage is reported for this pair, the pair for which the
theory under study makes a coverage claim.
\section{Rate-Scale Estimands for Count Outcomes}
\label{app:G}

Many clinical and pharmacoepidemiologic endpoints are recurrent event counts, such as hospitalizations, observed over follow-up times that differ \corr{across individuals}. Transporting such endpoints requires additional care because a population incidence rate is a ratio of population means rather than a mean of individual event-to-time ratios. In addition, person-time may itself depend on treatment.

We therefore consider two settings. We first treat the simpler case in which person-time is predetermined or can otherwise validly be incorporated with the baseline transport covariates. We then consider the more general setting in which person-time may depend on treatment, as can occur when follow-up is affected by death, treatment discontinuation, or another post-treatment event. In the latter case, person-time is itself a counterfactual quantity and must be transported together with the event count.

\subsection{Predetermined or treatment-independent person-time}
\label{sub:G-setup}

Each trial participant contributes $(X_i,T_i,t_i,Y_i)$, where $t_i>0$ is person-time and $Y_i\in\{0,1,2,\ldots\}$ is an event count. Each target participant contributes $(X_j^Q,t_j^Q,Y_j^Q)$. In this subsection, person-time is assumed to be predetermined with respect to treatment, or more generally to satisfy the conditions required for it to be included in the transport covariates.

Define the augmented covariate
\begin{equation}
W=(X,t),
\qquad
e_{0,W}(w)=P(T=0\mid W=w),
\qquad
A_W=\frac{1-T}{e_{0,W}(W)},
\label{eq:G-augmented}
\end{equation}
and let
\[
\mu_{0,W}(w)=\E_P(Y\mid W=w,T=0),
\qquad
r_W^\dagger=\frac{dQ_W}{dP_W}.
\]
When the same person-time applies under treatment and control, define
\begin{equation}
\lambda_a
=
\frac{\E_Q\{Y^Q(a)\}}{\E_Q(t)},
\qquad a\in\{0,1\},
\label{eq:G-rates}
\end{equation}
with
$\mathrm{IRR}=\lambda_1/\lambda_0$
and
$\mathrm{IRD}=\lambda_1-\lambda_0$.

\begin{assumption}[Identification on the rate scale]
\label{ass:G}
The following conditions hold.
\begin{enumerate}[label=\textup{(\roman*)},
                  leftmargin=2.6em,itemsep=2pt,topsep=3pt]

\item \corr{Treatment positivity holds given $W$:
$\underline e\le e_{0,W}(w)\le\overline e$ for $P_W$-almost every $w$. When
person-time is predetermined and randomization depends on $X$ alone --- the
design case --- $e_{0,W}=e_0$ and $A_W=A$, so this is nothing beyond the
mean-scale positivity. Overlap holds at the augmented covariate:
$Q_W\ll P_W$ with $\norm{r_W^\dagger}{\infty}\le\overline r_W<\infty$; where
the weighting branch of Proposition~\ref{prop:G-dr} is supported through the
rate theory of Appendix~\ref{app:E}, the two-sided bound
$0<\underline r_W\le r_W^\dagger$ of Assumption~\ref{ass:E-twosided} is required as well.}

\item The conditional control count regression transports across populations,
so that
\[
\E_Q\{Y^Q(0)\mid W=w\}
=
\E_P(Y\mid W=w,T=0)
=
\mu_{0,W}(w)
\]
for $Q_W$-almost every $w$.

\item \corr{Person-time is unaffected by treatment, $t(0)=t(1)=t$ almost
surely, and --- by the trial design --- treatment is randomized independently
of $\{Y(0),Y(1),t\}$ given $X$. This is what permits $t$ to enter the
transport covariate, and it makes the common denominator in
\eqref{eq:G-rates} well defined.}

\item \corr{$0<\E_Q(t)<\infty$, $\Var_Q(t)<\infty$, and the $W$-analogues
of the moment conditions in Assumption~\ref{ass:id-mom} hold; in particular
$\mu_{0,W}$ is bounded.}
\end{enumerate}
\end{assumption}

\corr{Consistency ($Y=TY(1)+(1-T)Y(0)$ in the trial and
$Y^{\mathrm Q}=Y^{\mathrm Q}(1)$ in the target) is inherited from
Assumption~\ref{ass:id}, which remains in force; it is used below when the treated
rate is identified from observed target means. As on the mean scale,
randomization itself is a property of the trial design and is not listed.}

Let $\widehat\mu_{0,W}$ denote the estimated control-count regression and let
$\widehat r_W$ denote the calibrated transport weight constructed with $W$ in
place of $X$. As in the mean-scale analysis, the transport weight is estimated
from the full trial and target covariate samples; treatment and outcome enter
separately through the augmentation term. The transported control-count mean is
estimated by
\begin{equation}
\widehat\theta_0
=
\E_m\!\left[\widehat\mu_{0,W}(X^Q,t^Q)\right]
+
\En\!\left[
A_W\widehat r_W(W)
\{Y-\widehat\mu_{0,W}(W)\}
\right].
\label{eq:G-theta0}
\end{equation}
The corresponding rate estimators are
\begin{equation}
\widehat\lambda_0
=
\frac{\widehat\theta_0}{\E_m(t^Q)},
\qquad
\widehat\lambda_1
=
\frac{\E_m(Y^Q)}{\E_m(t^Q)},
\label{eq:G-ratehat}
\end{equation}
with
$\widehat{\mathrm{IRR}}
=\widehat\lambda_1/\widehat\lambda_0$
and
$\widehat{\mathrm{IRD}}
=\widehat\lambda_1-\widehat\lambda_0$.

\corr{Because $b_{J,1}\equiv1$ (Assumption~\ref{ass:des-basis}), the empirical
Riesz calibration constraint gives}
\begin{equation}
\En\{\widehat r_W(W)\}=1
\label{eq:G-intercept}
\end{equation}
whenever the calibration program is feasible. The inverse-randomization factor
$A_W$ is part of the augmentation term and not of the transport-weight
normalization. \corr{In particular, if $\widehat\mu_{0,W}$ lies in the
calibration span $H_J$, the constraint gives
$\En[\widehat r_W\widehat\mu_{0,W}]=\E_m[\widehat\mu_{0,W}]$ exactly, so the
target- and trial-standardized versions of the regression term then
coincide; no exact normalization of $\En[A_W\widehat r_W]$ holds, and none
is needed --- the augmentation is centred by the inverse-randomization
identity, not by a normalization.}

\begin{proposition}[Double robustness on the rate scale]
\label{prop:G-dr}
Under Assumption~\ref{ass:G} and the same nuisance-estimation and empirical
process or cross-fitting conditions used for the mean-scale doubly robust
estimator, with $W$ replacing $X$,
$\widehat\lambda_0\xrightarrow{p}\lambda_0$
if either the calibrated transport weight is consistently estimated or the
control-count regression is consistently estimated. Consequently,
$\widehat{\mathrm{IRD}}\xrightarrow{p}\mathrm{IRD}$ and, provided
$\lambda_0>0$,
$\widehat{\mathrm{IRR}}\xrightarrow{p}\mathrm{IRR}$.
\end{proposition}

\begin{proof}
The denominator follows directly from the law of large numbers,
\[
\E_m(t^Q)\xrightarrow{p}\E_Q(t)\in(0,\infty).
\]

For the numerator, \eqref{eq:G-theta0} is the same augmented transported-mean
estimator used on the mean scale, with $W$ replacing $X$. For generic nuisance
functions $(\mu,r)$, its population version is
\[
\E_Q\{\mu(W)\}
+
\E_P\!\left[
A_Wr(W)\{Y-\mu(W)\}
\right].
\]
The inverse-randomization identity and the Riesz identity continue to hold
under Assumption~\ref{ass:G}. Hence
\begin{align}
&
\E_Q\{\mu(W)\}
+
\E_P\!\left[A_Wr(W)\{Y-\mu(W)\}\right]
-
\E_Q\{Y^Q(0)\}
\nonumber\\
&\qquad
=
-
\E_P\!\left[
\{r(W)-r_W^\dagger(W)\}
\{\mu(W)-\mu_{0,W}(W)\}
\right].
\label{eq:G-remainder}
\end{align}

The remainder is therefore a product of errors in the transport-weight and
outcome-regression components. Under the weighting branch it vanishes by
consistency of the calibrated transport weight, together with the stated
regularity conditions \corr{(the passage from the population identity to the
sample statement is the empirical-process content of Appendix~\ref{app:F}, applied at
$W$)}. Under the outcome-regression branch it vanishes by
consistency of $\widehat\mu_{0,W}$ and the corresponding boundedness or
integrability conditions. Thus
\[
\widehat\theta_0
\xrightarrow{p}
\E_Q\{Y^Q(0)\}.
\]
Combining numerator and denominator consistency gives
$\widehat\lambda_0\xrightarrow{p}\lambda_0$ by the continuous mapping theorem.
The treated rate is a ratio of two observed target means and is therefore
consistent for $\lambda_1$\corr{, using target-sample consistency
$Y^{\mathrm Q}=Y^{\mathrm Q}(1)$ from Assumption~\ref{ass:id}}. The results for the
incidence-rate difference and ratio follow by continuous mapping.
\end{proof}

\corr{\begin{remark}[What ``with $W$ replacing $X$'' involves]
\label{rem:G-transfer}
The hypothesis of Proposition~\ref{prop:G-dr} instantiates the mean-scale
conditions at the augmented covariate, and four things change materially.
The dimension grows: with $W\in\R^{d+1}$, the approximation exponent of
Assumption~\ref{ass:des-approx} becomes $J^{-s/(d+1)}$, the smoothness floor
$s>(d+1)/2$, and the balanced sieve
$J_N\asymp(N/\log N)^{(d+1)/(2s+d+1)}$. The design condition requires the
$t$-range to be a compact interval bounded away from zero. The quadratic
transport cost becomes unit-dependent --- $\norm{(x,t)-(x',t')}{2}^2$
changes with the units of person-time --- so $t$ must be standardized with
the other covariates by the prespecified scaling of Algorithm~\ref{alg:main} of the main text; the
population limit $r_W^\dagger$ is unit-invariant, but the finite-sample
weight and the diagnostics are not. Finally, when trial follow-up is nearly
deterministic --- administratively fixed follow-up being the leading case ---
the $t$-coordinate of the calibration design degenerates, the Gram condition
of Assumption~\ref{ass:des-basis} fails at $W$, and the calibration program becomes
infeasible unless the target person-time moments happen to match; this is
detected by the calibration residual of Section~\ref{sec:implementation} of the main text, and it is precisely the
situation in which the construction of
Section~\ref{sub:G-cautions}, which keeps $t$ out of the transport
covariate, is the more robust choice.
\end{remark}}

\subsection{Two cautions and treatment-dependent person-time}
\label{sub:G-cautions}

The first caution concerns the estimand itself. It may be tempting to apply
the transported-mean estimator directly to the individual rate $Y/t$, but
doing so generally targets a different population quantity.

\begin{proposition}[Unit-level rates target a different estimand]
\label{prop:G-unitrate}
\corr{Assume $0<\E_Q(t)<\infty$ and $\E_Q\{Y^Q(0)/t\}<\infty$ --- the latter
is not implied by Assumption~\ref{ass:G}, since person-time is not bounded
away from zero. Then}
\begin{equation}
\E_Q\!\left\{\frac{Y^Q(0)}{t}\right\}
=
\lambda_0
-
\frac{
\operatorname{Cov}_Q\{t,Y^Q(0)/t\}
}{
\E_Q(t)
}.
\label{eq:G-unitrate}
\end{equation}
Thus, the mean of individual rates equals the ratio-of-means incidence rate
only when person-time is uncorrelated with the individual event rate.
\end{proposition}

\begin{proof}
Let $R=Y^Q(0)/t$, so that $Y^Q(0)=tR$. The covariance identity gives
\[
\E_Q\{Y^Q(0)\}
=
\E_Q(t)\E_Q(R)
+
\operatorname{Cov}_Q(t,R).
\]
Dividing by $\E_Q(t)$ and rearranging yields
\eqref{eq:G-unitrate}.
\end{proof}

\corr{The proposition compares two population functionals. That the
transported-mean estimator applied to the individual rate targets the
left-hand side of \eqref{eq:G-unitrate} follows from
Assumption~\ref{ass:G}(ii) precisely because $t$ is a coordinate of $W$:
conditioning on $W=(x,t)$ fixes $t$, so
$\E_Q\{Y^Q(0)/t\mid W\}=\mu_{0,W}(W)/t$ transports whenever $\mu_{0,W}$
does. The covariance in \eqref{eq:G-unitrate} involves the counterfactual
joint law under $Q$ and is displayed for interpretation, not as an estimable
correction.}

The second caution concerns where the outcome-regression standardization is
performed. For double robustness, the regression component must be averaged
over the target population. Replacing this target standardization by a
weighted trial average causes the augmentation term to cancel algebraically.

\begin{proposition}[Trial-side standardization collapses the augmentation]
\label{prop:G-collapse}
Suppose the target standardization in \eqref{eq:G-theta0} is replaced by its
trial-weighted analogue,
\[
\widehat\theta_0^{\,\mathrm{trial}}
=
\En\!\left[
A_W\widehat r_W(W)\widehat\mu_{0,W}(W)
\right]
+
\En\!\left[
A_W\widehat r_W(W)
\{Y-\widehat\mu_{0,W}(W)\}
\right].
\]
Then
\[
\widehat\theta_0^{\,\mathrm{trial}}
=
\En\!\left[A_W\widehat r_W(W)Y\right]
\]
identically. The resulting estimator therefore has no
outcome-regression robustness branch.
\end{proposition}

\begin{proof}
Expanding the augmentation term gives
\begin{align*}
\widehat\theta_0^{\,\mathrm{trial}}
&=
\En[A_W\widehat r_W\widehat\mu_{0,W}]
+
\En[A_W\widehat r_WY]
-
\En[A_W\widehat r_W\widehat\mu_{0,W}]\\
&=
\En[A_W\widehat r_WY].
\end{align*}
The cancellation is exact. Since the resulting estimator no longer depends
on $\widehat\mu_{0,W}$, correctness of the outcome regression cannot compensate
for a misspecified transport weight. This is why the standardization term in
\eqref{eq:G-theta0} must be evaluated under the target empirical distribution.
\end{proof}

\corr{The exact cancellation is special to the $A_W$-weighted trial average.
The variant without $A_W$, which replaces the standardization term by
$\En[\widehat r_W\widehat\mu_{0,W}]$, does not collapse algebraically ---
indeed if $\widehat\mu_{0,W}\in H_J$ it coincides with $\widehat\theta_0$
exactly, by the calibration constraint. But its \emph{population} target
does collapse: by the inverse-randomization identity,
$\E_P[r\mu]+\E_P[A_Wr\{Y-\mu\}]=\E_P[r\,\mu_{0,W}]$ for every $(\mu,r)$,
which is free of $\mu$, so the outcome-regression branch is destroyed by
\emph{any} trial-side standardization: the error
$\E_P[(r-r_W^\dagger)\mu_{0,W}]$ then vanishes only if the weight is
consistent. This population argument is what underlies the statement in
Section~\ref{sec:rate} of the main text that moving the standardization to the weighted
trial population eliminates the outcome-regression robustness branch.}

When observed person-time can depend on treatment, it cannot be included in the baseline transport covariates. This situation is particularly relevant for recurrent-event endpoints when follow-up may be affected by death, treatment discontinuation, or other post-treatment events. The counterfactual control event count and counterfactual control person-time must then both be transported from the randomized trial.

\begin{assumption}[Identification with treatment-dependent person-time]
\label{ass:G-time}
The following conditions hold.
\begin{enumerate}[label=\textup{(\roman*)},
                  leftmargin=2.6em,itemsep=2pt,topsep=3pt]

\item \corr{Treatment positivity and trial-to-target overlap hold for the
baseline covariates $X$ as in Assumption~\ref{ass:id}; by the trial design, treatment
is randomized independently of the joint potential outcomes
$\{Y(0),Y(1),t(0),t(1)\}$ given $X$; and person-time obeys consistency,
$t=T\,t(1)+(1-T)\,t(0)$.}

\item The conditional control event-count mean transports across populations,
so that
\[
\E_Q\{Y^Q(0)\mid X=x\}
=
\E_P(Y\mid X=x,T=0)
=
\mu_0(x)
\]
for $Q_X$-almost every $x$.

\item The conditional control person-time mean also transports,
\[
\E_Q\{t(0)\mid X=x\}
=
\E_P(t\mid X=x,T=0)
=
\nu_0(x)
\]
for $Q_X$-almost every $x$.

\item \corr{The relevant count and person-time moments are finite, and
$0<\E_Q\{t(0)\}<\infty$, $0<\E_Q\{t(1)\}<\infty$.}
\end{enumerate}
\end{assumption}

Let
\[
A=\frac{1-T}{P(T=0\mid X)}
\]
and let $\widehat r$ be the calibrated target-to-trial transport weight based
only on the baseline covariates. Let $\widehat\mu_0$ estimate the conditional
control event-count mean and let $\widehat\nu_0$ estimate the conditional
control person-time mean. Define
\begin{align}
\widehat\theta_Y
&=
\E_m\{\widehat\mu_0(X^Q)\}
+
\En\!\left[
A\widehat r(X)
\{Y-\widehat\mu_0(X)\}
\right],
\label{eq:G-thetaX}
\\
\widehat\theta_t
&=
\E_m\{\widehat\nu_0(X^Q)\}
+
\En\!\left[
A\widehat r(X)
\{t-\widehat\nu_0(X)\}
\right].
\label{eq:G-thetat}
\end{align}
The transported counterfactual control rate is
\begin{equation}
\widehat\lambda_0^{\,\mathrm{tt}}
=
\frac{\widehat\theta_Y}{\widehat\theta_t}.
\label{eq:G-rate-tt}
\end{equation}
The treated rate remains directly estimable from the target sample as
\[
\widehat\lambda_1^{\,\mathrm{tt}}
=
\frac{\E_m(Y^Q)}{\E_m(t^Q)}.
\]
\corr{The corresponding estimands are the arm-specific rates
$\lambda_a^{\mathrm{tt}}=\E_Q\{Y^Q(a)\}/\E_Q\{t(a)\}$, with
$\mathrm{IRR}^{\mathrm{tt}}=\lambda_1^{\mathrm{tt}}/\lambda_0^{\mathrm{tt}}$
and
$\mathrm{IRD}^{\mathrm{tt}}=\lambda_1^{\mathrm{tt}}-\lambda_0^{\mathrm{tt}}$;
the superscript records that each arm now carries its own person-time
denominator, so these are contrasts of rates rather than of counts over a
common exposure, and no cancellation of person-time is available in the
ratio (contrast Remark~\ref{rem:G-irr} below).}

\begin{proposition}[Rate with treatment-dependent person-time]
\label{prop:G-time}
Under Assumption~\ref{ass:G-time} and the corresponding regularity conditions
for the two augmented transported-mean estimators,
\[
\widehat\lambda_0^{\,\mathrm{tt}}
\xrightarrow{p}
\frac{\E_Q\{Y^Q(0)\}}{\E_Q\{t(0)\}}.
\]
In particular, consistency holds if either the calibrated transport weight is
consistent, or both the control event-count regression and the control
person-time regression are consistent. Consequently, the corresponding
incidence-rate difference and incidence-rate ratio are consistent whenever
their denominators are positive.
\end{proposition}

\begin{proof}
Both \eqref{eq:G-thetaX} and \eqref{eq:G-thetat} are augmented
transported-mean estimators based on the baseline covariates $X$. Applying the
mean-scale double-robustness argument to the event count gives $\widehat\theta_Y
\xrightarrow{p}
\E_Q\{Y^Q(0)\}$ whenever either $\widehat r$ or $\widehat\mu_0$ is consistent. Applying the
same argument to person-time gives
$\widehat\theta_t
\xrightarrow{p}
\E_Q\{t(0)\}$
whenever either $\widehat r$ or $\widehat\nu_0$ is consistent.

Because the same transport weight enters both components, both limits hold if
$\widehat r$ is consistent. \corr{If the weighting model is not consistent,
it suffices that both $\widehat\mu_0$ and $\widehat\nu_0$ be consistent; the
exact remainders can also vanish through orthogonality, as in
\eqref{eq:G-remainder}, so joint consistency is sufficient rather than
necessary.} Since
$\corr{\E_Q}\{t(0)\}>0$, the continuous mapping theorem gives
\[
\widehat\lambda_0^{\,\mathrm{tt}}
\xrightarrow{p}
\frac{\E_Q\{Y^Q(0)\}}{\E_Q\{t(0)\}}.
\]
The treated rate is a ratio of observed target means, and consistency of the
rate difference and rate ratio follows by continuous mapping.
\end{proof}

\subsection{Influence-function inference on the rate scale}
\label{sub:G-if}

For the predetermined-person-time case, let
\[
M_0=\E_Q\{Y^Q(0)\},
\qquad
M_t=\E_Q(t),
\]
\corr{and record the efficient influence-function pair for the transported
control event-count mean $M_0$, obtained from Theorem~\ref{thm:C-eif} by dropping the
observed treated-arm term and reversing the sign of the control component:
\[
\varphi_Q^{M_0}(O^{\mathrm Q})=\mu_{0,W}(W^Q)-M_0,
\qquad
\varphi_P^{M_0}(O^{\mathrm R})=A_W\,r_W^\dagger(W)\{Y-\mu_{0,W}(W)\},
\]
with $W^Q=(X^Q,t^Q)$.} The target-sample influence function for $M_t$ is
$t-M_t$, with no trial-sample component. Assuming \corr{$M_0>0$ and $M_t>0$
(the latter is Assumption~\ref{ass:G}(iv))}, the delta method
applied to
$\log\lambda_0=\log M_0-\log M_t$
gives
\begin{equation}
\varphi_Q^{\log\lambda_0}
=
\frac{\varphi_Q^{M_0}}{M_0}
-
\frac{t-M_t}{M_t},
\qquad
\varphi_P^{\log\lambda_0}
=
\frac{\varphi_P^{M_0}}{M_0}.
\label{eq:G-logIF}
\end{equation}
{ The corresponding asymptotic variance keeps the target-side covariance ---
the two target components are functions of the same observation:
\[
V_{\log\lambda_0}
=\Var_Q\Bigl[\frac{\varphi_Q^{M_0}}{M_0}-\frac{t-M_t}{M_t}\Bigr]
+\eta^{-1}\,\frac{\Var_P\bigl[\varphi_P^{M_0}\bigr]}{M_0^{2}},
\]
whose first term expands with the cross term
$-2\,\Cov_Q\{\mu_{0,W}(W^Q),\,t^Q\}/(M_0M_t)$ --- generically nonzero
precisely in the settings Proposition~\ref{prop:G-unitrate} warns about. A
plug-in estimator follows by replacing nuisance functions and population
moments by their estimates.}

Influence functions for the incidence-rate difference follow by
\corr{pointwise subtraction of} the treated- and control-rate influence
functions \corr{within each sample; the two rates share the target sample and
the denominator $\E_m(t^Q)$, so their sampling errors are dependent and
variances may not simply be added}. For the incidence-rate ratio,
it is convenient to work on the log scale and subtract the corresponding
log-rate influence functions. Confidence intervals can then be
back-transformed to the original scale.

{ \begin{remark}[The rate ratio does not involve person-time]
\label{rem:G-irr}
When person-time is predetermined, $\E_Q(t)$ cancels exactly in the ratio:
$\mathrm{IRR}=\lambda_1/\lambda_0
 =\E_Q\{Y^Q(1)\}/\E_Q\{Y^Q(0)\}=M_1/M_0$ with $M_1=\E_Q[Y^{\mathrm Q}]$.
Consequently $\log\mathrm{IRR}=\log M_1-\log M_0$ has influence functions
\[
\varphi_Q^{\log\mathrm{IRR}}
=\frac{Y^{\mathrm Q}-M_1}{M_1}-\frac{\varphi_Q^{M_0}}{M_0},
\qquad
\varphi_P^{\log\mathrm{IRR}}=-\frac{\varphi_P^{M_0}}{M_0},
\]
the $(t-M_t)/M_t$ terms cancelling identically: the incidence-rate ratio
requires no person-time moments at all, and only the incidence-rate
difference genuinely involves $M_t$. This sharpens the contrast with the
treatment-dependent case, where the two arms carry different denominators and
nothing cancels.
\end{remark}

For treatment-dependent person-time, write $\theta_Y=\E_Q\{Y^Q(0)\}$ and
$\theta_t=\E_Q\{t(0)\}$, with influence pairs
\[
\varphi_Q^{Y}=\mu_0(X^Q)-\theta_Y,\quad
\varphi_P^{Y}=A\,\rdag(X)\{Y-\mu_0(X)\},
\quad
\varphi_Q^{t}=\nu_0(X^Q)-\theta_t,\quad
\varphi_P^{t}=A\,\rdag(X)\{t-\nu_0(X)\},
\]
obtained exactly as for $M_0$. The delta method applied to
$\log\lambda_0^{\mathrm{tt}}=\log\theta_Y-\log\theta_t$ gives, in each
channel, $\varphi^{\log}=\varphi^{Y}/\theta_Y-\varphi^{t}/\theta_t$, and
\[
V_{\log\lambda_0^{\mathrm{tt}}}
=\Var_Q\Bigl[\frac{\varphi_Q^{Y}}{\theta_Y}
              -\frac{\varphi_Q^{t}}{\theta_t}\Bigr]
+\eta^{-1}\,
\Var_P\Bigl[\frac{\varphi_P^{Y}}{\theta_Y}
             -\frac{\varphi_P^{t}}{\theta_t}\Bigr].
\]
Both channels carry covariances that must not be dropped: on the trial side
$\Cov_P\{\varphi_P^{Y},\varphi_P^{t}\}
 =\E_P\bigl[A^2\rdag(X)^2\{Y-\mu_0(X)\}\{t-\nu_0(X)\}\bigr]$ is generically
large and positive --- longer follow-up accompanies more events --- and on
the target side $\mu_0(X^Q)$ and $\nu_0(X^Q)$ are correlated through the
shared covariate. The requisite moments and the positivity of $\theta_t$ and
of $\E_Q\{t(1)\}$ are Assumption~\ref{ass:G-time}(iv).}

\begin{figure}[H]
\centering
\includegraphics[width=0.8\textwidth]{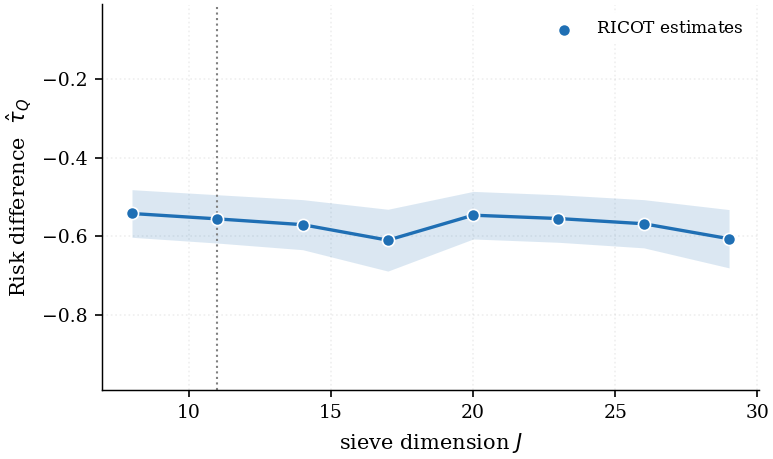}
\caption{Sensitivity of the estimated risk difference $\hat{\tau}_Q$ to the
sieve dimension $J$. Points show the point estimate at each value of $J$ for
which the calibration problem was solved, and the shaded band the associated
95\% confidence interval; the dotted vertical line marks the value used in the
main analysis. The estimate varies only marginally over the range considered,
remaining well within the confidence bands throughout.}
\label{fig:sieve-stable-real}
\end{figure}

\end{document}